\documentclass[reqno,oneside,11pt]{amsart}

\usepackage[english]{babel}
\usepackage{lmodern}
\usepackage{newpxtext}			
\usepackage[scaled]{helvet}	
\usepackage{courier}				
\usepackage[OT1,euler-digits]{eulervm}			
\usepackage[bb=boondox,bbscaled=1.05,scr=dutchcal]{mathalfa}	%
\usepackage{dsfont}
\usepackage[T1]{fontenc}
\normalfont
\usepackage{float}

\usepackage[a4paper,margin=2.5cm]{geometry}
\usepackage[dvipsnames]{xcolor}
\usepackage[hypertexnames=false,colorlinks=true,linkcolor=blue,%
citecolor=purple,filecolor=magenta,urlcolor=cyan]{hyperref}
\DeclareRobustCommand{\SkipTocEntry}[5]{} 

\usepackage{amssymb,anyfontsize,enumerate,layout,mathrsfs,mdwlist,stmaryrd,thmtools,tikz}
\usepackage{dsfont} 
\usepackage[textsize=footnotesize]{todonotes}
\presetkeys%
    {todonotes}%
    {color=Apricot}{}%
\usepackage[style=alphabetic,sorting=nty,maxnames=99,maxalphanames=5]{biblatex}
\usepackage{csquotes}
\usepackage{mathtools}

\usepackage[english,noabbrev]{cleveref}
\crefname{lemma}{lemma}{lemmata}
\Crefname{lemma}{Lemma}{Lemmata}

\allowdisplaybreaks 

\theoremstyle{plain}                          
\newtheorem{theorem}{Theorem}[section]
    
\newtheorem{lemma}[theorem]{Lemma}

\theoremstyle{definition}
\newtheorem{definition}[theorem]{Definition}
\newtheorem{prop-defin}[theorem]{Proposition-definition} 
\newtheorem{example}[theorem]{Example}
\newtheorem*{exercise}{Exercise}

\newtheorem*{claim}{Claim}
\theoremstyle{remark}
\newtheorem{remark}[theorem]{Remark}

\let\frontmatter\relax
\def\mainmatter{\def\baselinestretch{1.1}\normalfont \setlength{\parskip}{0.5em}}

\numberwithin{equation}{section}

\makeatletter
\def\l@section{\@tocline{2}{-3pt}{1pc}{5pc}{}}
\def\l@subsection{\@tocline{2}{-6pt}{2.5pc}{5pc}{}}
\def\l@subsubsection{\@tocline{2}{-6pt}{4pc}{5pc}{}}

\renewcommand{\tocsection}[3]{%
  \indentlabel{\@ifnotempty{#2}{\bfseries \small\ignorespaces#1 #2.~}}\small\bfseries#3}
\renewcommand{\tocsubsection}[3]{%
 \indentlabel{\@ifnotempty{#2}{\small\ignorespaces#1 #2.~}}\small#3}
\renewcommand{\tocsubsubsection}[3]{%
 \indentlabel{\@ifnotempty{#2}{\small\ignorespaces#1 #2.~}}\small#3}

\makeatother

\makeatletter
\renewcommand{\section}{\@startsection
{section}
{1}
{\z@}
{-1\baselineskip}
{0.8\baselineskip}
{\centering\scshape\large}} 
\renewcommand{\subsection}{\@startsection
{subsection}
{2}
{\z@}
{-1.2\baselineskip}
{0.5\baselineskip}
{\normalfont \bf \normalsize}} 
\renewcommand{\subsubsection}{\@startsection
{subsubsection}
{3}
{\z@}
{-1\baselineskip}
{0.5\baselineskip}
{\normalfont \it \normalsize}} 
\makeatother

\DeclareMathOperator*{\Res}{Res}

\usepackage{mathtools}

\let\hslash\hslash
\newcommand{\I}{\mathcal{I}}
\newcommand{\Ra}{\text{Ram}}
\newcommand{\DA}{\mathcal{D}_A}
\newcommand{\DAh}{\mathcal{D}_A^{\hslash}}
\newcommand{\DAhC}{\widehat{\mathcal{D}}_A^{\hslash}}

\newcommand{\bea}{\begin{eqnarray}}
\newcommand{\eea}{\end{eqnarray}}
\newcommand{\beq}{\begin{equation}}
\newcommand{\eeq}{\end{equation}}

\begin{document}
\frontmatter

\title{Les Houches lecture notes on topological recursion}

\author{Vincent Bouchard}
\address{Department of Mathematical \& Statistical Sciences,
	University of Alberta, U-Commons 5-234\\
	Edmonton, Alberta, Canada T6G 2N8}
\email{vincent.bouchard@ualberta.ca}

\begin{abstract}
You may have seen the words ``topological recursion'' mentioned in papers on matrix models, Hurwitz theory, Gromov-Witten theory, topological string theory, knot theory, topological field theory, JT gravity, cohomological field theory, free probability theory, gauge theories, to name a few. The goal of these lecture notes is certainly not to explain all these applications of the topological recursion framework. Rather, the intention is to provide a down-to-earth (and hopefully accessible) introduction to topological recursion itself, so that when you see these words mentioned, you can understand what it is all about. These lecture notes accompanied a series of lectures at the Les Houches school ``Quantum Geometry (Mathematical Methods for Gravity, Gauge Theories and Non-Perturbative Physics)'' in Summer 2024. 
\end{abstract}

\maketitle
\mainmatter

\newpage
\tableofcontents

\newpage

\section{Introduction}

Topological recursion was originally obtained by Eynard and Orantin in \cite{EO07,EO08}, inspired by previous work with Chekhov \cite{CE06}, as a recursive method for computing correlators of Hermitian matrix models. However, what made \cite{EO07} such an important paper is the authors' deep insight that the topological recursion framework may be much more general than previously anticipated. The authors proposed that, given any spectral curve, topological recursion can be used to construct a sequence of correlators that may have interesting applications, regardless of whether there is a matrix model behind the scenes or not. In essence, Eynard and Orantin proposed in \cite{EO07} a framework for constructing collections of objects attached to spectral curves, and asked the mathematical physics community: are these objects useful?

It turns out that topological recursion now has applications in a wide range of areas in mathematics and physics. You may have seen the words ``topological recursion'' mentioned in papers on matrix models, Hurwitz theory, Gromov-Witten theory, topological string theory, knot theory, topological field theory, JT gravity, cohomological field theory, free probability theory, gauge theories, to name a few. The goal of these lecture notes is certainly not to explain all these applications of the topological recursion framework. Rather, the intention is to provide an introduction to topological recursion itself, so that when you see these words mentioned, you can understand what it is all about.

These notes accompanied a series of three 90 minute lectures at the Les Houches school in Quantum Geometry in the summer of 2024. As such, they are rather limited in scope. My intention is to provide a down-to-earth (and hopefully accessible) introduction to the topological recursion framework. Given the extensive literature on the subject, and the fact that this research area often involves lengthy technical calculations, I hope that having an introductory set of lecture notes on topological recursion itself may be of use for newcomers to the subject.

Instead of introducing the framework as it was first formulated by Eynard and Orantin, I decided to start from the point of view of Airy ideals. Airy ideals are particular classes of differential constraints that naturally generalize the Virasoro constraints appearing in Witten's conjecture. Airy ideals were first proposed by Kontsevich and Soibelman in 2017 as an algebraic reformulation of topological recursion \cite{KS17} (see also \cite{ABCO17}).\footnote{Note that they were initially called "Airy structures" by Kontsevich and Soibelman, but I prefer to call them "Airy ideals" as they are a characterization of left ideals in the Rees Weyl algebra, as we will see.} I will however introduce Airy ideals following the approach in \cite{BCJ22}. I find that Airy ideals provide a clean, clear formulation of the type of recursion relations that we are interested in, and, as such, a good starting point to dive into the theory of topological recursion. The theory of Airy ideals is explored in Section \ref{s:airy}.

I will then introduce the original topological recursion framework, as a residue formulation on spectral curves, in Section \ref{s:TR}. I tried to be fairly general, but at the same time remain accessible and pedagogical. For instance, I focused mostly on spectral curves with simple ramification points, which is the case originally covered in \cite{EO07}, to avoid the complicated and technical combinatorics that obscure the general case (see \cite{BE13,BE17,BBCCN18,BKS20,BBCKS23}).

In Section \ref{s:loop} I connect the two frameworks via loop equations. Loop equations are in fact the origin of topological recursion. In matrix models, loop equations are also known as Schwinger-Dyson equations, and topological recursion was obtained as the relevant solution of loop equations that calculates the correlators of the matrix model. More abstractly, it can be shown that topological recursion is the unique solution to the loop equations that satisfies a particular property known as the projection property. This is explained in Section \ref{s:loop}. 

The connection with Airy ideals goes as follows. After imposing the projection property on correlators, loop equations can be reformulated as differential constraints for an associated partition function. Those differential constraints can be shown to generate an Airy ideal. As a result, topological recursion can be thought of as a particular case of the differential constraints provided by Airy ideals. In fact, the picture is very nice. Topological recursion is a method for constructing correlators on spectral curves; it does so via residue calculations locally at the ramification points of the spectral curve. Through the connection via loop equations, we see that topological recursion can be recast as attaching a collection of Virasoro constraints (or more generally $\mathcal{W}$-constraints) at each ramification point of the spectral curve. Neat!

Finally, I explore in Section \ref{s:connections} a few relations between topological recursion and the rest of the world. What makes topological recursion so interesting in particular applications is that it provides a direct connection  between various topics, such as integrability, enumerative geometry, differential constraints, cohomological field theory, quantum curves, WKB analysis and resurgence. These established connections may lead to new results in a particular context of application.

Given the length of this series of lectures, it is impossible to cover all connections. Instead, I decided in Section \ref{s:connections} to explore a little bit two of them, namely enumerative geometry and quantum curves. I already discuss the intimate connection between enumerative geometry on the moduli space of curves and topological recursion in Section \ref{s:TR}, but in Section \ref{s:connections} I explore further connections with enumerative geometry, such as Hurwitz theory and Gromov-Witten theory. In particular, I explain how the expansion of the same correlators at different points on the spectral curve may connect distinct enumerative geometric interpretations and yield interesting results such as ELSV-type formulae for Hurwitz numbers and localization formulae for Gromov-Witten invariants. I also briefly explore in Section \ref{s:connections} the connection between topological recursion and quantum curves. In particular, resurgence and Borel resummation of WKB solutions to quantum curves obtained via topological recursion is a very active area of research, which I briefly mention.

\subsection{Word of caution}

Given the limited amount of time in this lecture series, these lecture notes are limited in scope. One could easily write a whole book on topological recursion and applications (and perhaps one should do that!), but this is not the goal of these lecture notes. To remain fairly concise, especially given my tendency to be rather verbose, I had to make choices. 

I decided to stay focused on the topological recursion framework itself. As such, I only briefly mention applications of topological recursion, and omit entirely many topics of interest. For instance, I do not discuss at all knot invariants, cohomological field theory, free probability, matrix models, JT gravity, gauge theory and Whittaker vectors, to name a few. I only briefly mention the very important connection with integrability. I also do not discuss many generalizations of topological recursion and Airy ideals, such as non-commutative topological recursion, $Q$-deformed topological recursion, geometric recursion, super Airy structures and super topological recursion, etc.

This may have been a poor decision, since after all what makes topological recursion interesting is its applications and the connections that it provides. But to understand the applications of topological recursion, one should first understand what topological recursion is all about in the first place. I hope that these lecture notes may be little bit helpful in this direction.

I also apologize in advance for the many papers that I did not cite in this review. The literature surrounding topological recursion is now rather extensive, and thus it is impossible to do justice to all the interesting papers out there. I am also aware that my references throughout the text may be skewed towards my own papers; this is not because I think they are particularly interesting, it's just that these are the ones  I know the best! :-)

\subsection{Acknowledgments}

I would like to thank the Les Houches School of Physics for a fantastic learning environment and in particular the organizers of the summer school ``Quantum Geometry (Mathematical Methods for Gravity, Gauge  Theories and Non-Perturbative Physics)'' for the invitation to give this series of lectures. I would also like to thank Rapha\"el Belliard, Swagata Datta, Bertrand Eynard, Alessandro Giacchetto and Reinier Kramer for discussions and comments on some aspects of these lectures and notes. Finally, I would like to acknowledge support from the  National Science and Engineering Research Council of Canada.

\section{Airy ideals (Airy structures)}

\label{s:airy}

\subsection{Introduction: Witten's conjecture}

Witten's conjecture \cite{Wi91}, which was proved by Kontsevich in \cite{Ko91}, is perhaps one of the most celebrated result in physical mathematics. Witten was interested in studying two-dimensional quantum gravity. His deep insight was that two theories describing the same physics should actually be the same; more precisely, they should have the same partition function.

One of the two models of two-dimensional quantum gravity he was interested in relates the partition function to a fascinating problem in enumerative geometry, namely intersection theory over the moduli space of curves. More precisely, using the conventions that will be relevant later on, the partition function $Z$ of the theory can be expanded as
\begin{equation}\label{eq:ZZ}
Z = \exp\left( \sum_{\substack{g \in \mathbb{N}, n \in \mathbb{N}^* \\ 2g-2+n > 0}} \frac{\hslash^{2g-2+n}}{n!}  \sum_{k_1, \ldots, k_n \in \mathbb{N}}  F_{g,n}[2k_1+1, \ldots,2k_n+1]  t_{2k_1+1} \cdots t_{2k_n+1}\right),
\end{equation}
where the variables $t_1, t_3, t_5, \ldots$ are just formal variables for us. The coefficients $F_{g,n}[2k_1+1, \ldots,2k_n+1]$ have a natural interpretation as integrals of particular cohomology classes (so-called ``$\psi$-classes'') over the compactified moduli space of curves:
\begin{equation}
F_{g,n}[2k_1+1, \ldots,2k_n+1] = \prod_{i=1}^n (2 k_i+1)!! \int_{\overline{\mathcal{M}}_{g,n}} \psi_1^{k_1} \cdots \psi_n^{k_n},
\end{equation}
with the double factorial defined by
\begin{equation}
(2k+1) !! = (2k+1) (2k-1) (2k-3) \cdots (3)(1).
\end{equation}
I will not define these objects here; I refer the reader to the lecture series on moduli spaces of Riemann surfaces at this school \cite{GL24}.

The second model of two-dimensional quantum gravity relates the partition function to the so-called KdV integrable system, which is an infinite system of partial differential equations in the variables $t_a$. A priori, the KdV integrable system appears to have nothing to do with the physics at hand; the KdV equation, which is the first equation in the hierarchy, is used as a mathematical model of waves on shallow water surfaces! But this particular model of quantum gravity states that the partition function $Z$ of the theory should be a particular tau-function for the KdV integrable system. What this means is that $\partial_{t_1}^2 \log Z$ should be a solution of the KdV integrable system, and not just any solution; it should be the unique solution that satisfies an initial condition known as the string equation.

Witten's insight is that both theories should be describing the same partition function! Therefore, the $Z$ of \eqref{eq:ZZ}, which can be understood as a generating series for intersection numbers on $\overline{\mathcal{M}}_{g,n}$, should be a tau-function for the KdV integrable system. What a remarkable and unexpected statement, which was ultimately proved by Kontsevich using matrix models (another subject of a lecture series at this school).

There is a third way to rephrase Witten's conjecture, which was proposed by Dijkgraaf-Verlinde-Verlinde \cite{DVV91}. Instead of writing down the KdV integrable system for which $Z$ is a tau-function of, one can instead write directly an infinite set of differential constraints for $Z$, which take the form $L_k Z = 0$, $k \in \mathbb{Z}_{\geq -1}$, for some differential operators $L_k$ in the variables $t_a$ (we will see the explicit form of these operators later). Moreover, these operators are not random; they form a representation of a subalgebra of the Virasoro algebra:
\begin{equation}
[L_m, L_n] = \hslash^2 (m-n) L_{m+n}, \qquad m,n \geq -1.
\end{equation}
What is great with these differential constraints is that they uniquely fix $Z$. Furthermore, by applying the operators on $Z$, one obtains a very explicit recursive formula for the coefficients $F_{g,n}$, which are (up to rescaling) intersection numbers on $\overline{\mathcal{M}}_{g,n}$.

These differential constraints and the resulting recursion relations are what Airy ideals generalize. We move away from the particular geometric problem studied by Witten, and ask a much more general question. Given a general partition function, which we can think of as a generating series for some coefficients $F_{g,n}$ (which may have a particular interpretation in a given geometric or physical problem), when is it the case that there exists differential constraints that uniquely fix the partition function recursively?

\subsection{Formulating differential constraints more generally}

Our main object of study will be what we will call a ``partition function'':
\begin{equation}\label{eq:pf}
Z = \exp \left(\sum_{k=1}^\infty \hslash^{k} q^{(k+2)}(t_A)\right),
\end{equation}
where the $ q^{(k+2)} (t_A) $ are polynomials of degree $ \leq k+2 $ in a set of variables $ t_A $ (normalized such that $q^{(k+2)}(0) = 0$). We can decompose the polynomials $q^{(k+2)}(t_A)$ into homogeneous polynomials $F_{g,n}(t_A)$ of degree $n$:\footnote{In these notes, we use $\mathbb{N}$ to denote the set of non-negative integers, and $\mathbb{N}^*$ to denote the set of positive integers.}
\begin{equation}
q^{(k+2)}(t_A) = \sum_{\substack{g \in \frac{1}{2} \mathbb{N},  n \in \mathbb{N}^* \\ 2g-2+n=k}} F_{g,n}(t_A), 
\end{equation}
so that the partition function is rewritten as
\begin{equation}\label{eq:pfFgn}
Z = \exp\left( \sum_{\substack{g \in \frac{1}{2} \mathbb{N}, n \in \mathbb{N}^* \\ 2g-2+n > 0}} \hslash^{2g-2+n} F_{g,n}(t_A) \right).
\end{equation}
We also often expand the homogeneous polynomials $F_{g,n}(t_A)$ as
\begin{equation}
F_{g,n}(t_A) = \sum_{k_1, \ldots, k_n \in A} \frac{1}{n!} F_{g,n}[k_1, \ldots,k_n]  t_{k_1} \cdots t_{k_n}
\end{equation}
which brings the partition function in the standard form
\begin{equation}\label{eq:pfFgncoeff}
Z = \exp\left( \sum_{\substack{g \in \frac{1}{2} \mathbb{N}, n \in \mathbb{N}^* \\ 2g-2+n > 0}} \frac{\hslash^{2g-2+n}}{n!}  \sum_{k_1, \ldots, k_n \in A}  F_{g,n}[k_1, \ldots,k_n]  t_{k_1} \cdots t_{k_n}\right).
\end{equation}

A natural question is to ask whether there exists non-trivial differential operators $P$ in the variables $t_A$ and that are polynomial in $\hslash$ such that $P Z = 0$. For a generic partition function $Z$, there will be no such non-zero operators. But for some specific choices of partition functions, such operators may exist. 

The reason that this is interesting is the following. On the one hand, if such an operator $P$ exists, the differential constraint $P Z = 0$ leads to recursion relations for the coefficients $F_{g,n}[k_1,\ldots,k_n]$, since the operator $P$ is polynomial in $\hslash$ while the argument of the exponential is a formal series in $\hslash$. On the other hand, in general we will be interested in calculating the coefficients $F_{g,n}[k_1,\ldots,k_n]$, as those may give interesting enumerative invariants, such as integrals over the moduli space of curves or correlation functions of a physical theory. Thus, the existence of non-trivial operators $P$ that annihilate $Z$ lead to non-trivial recursion relations between the quantities of interest.

In fact, we can ask an even stronger question: for what $Z$ can we find a ``maximal'' set of differential operators that annihilate $Z$? This is even more interesting, because if there are enough such independent operators (and hence the notion of ``maximality''), then the resulting recursion relations could potentially be used to fully reconstruct $Z$ uniquely, and hence completely determine the quantities of interest.

This is the question that Airy ideals address.

\subsection{The Rees Weyl algebra}

Let us now make these concepts more precise. We follow the approach in \cite{BCJ22}.  Let $A$ be a finite or countably infinite index set (think of $A = \mathbb{N}$). We use the notation $t_A$ for the set of variables $\{t_a \}_{a \in A}$, and $\partial_A$ for the set of differential operators $\left\{ \frac{\partial}{\partial t_a} \right \}_{a \in A}$. 

\begin{definition}\label{d:Weyl}
The \emph{Weyl algebra} $ \mathbb{C}[t_A]\langle \partial_A \rangle$ is the algebra of differential operators with polynomial coefficients. We define the \emph{completed Weyl algebra} $\mathcal{D}_A$ to be the completion of the Weyl algebra where we allow infinite sums in the derivatives (when $A$ is a countably infinite index set) but not in the variables. For simplicity, we will usually call $\mathcal{D}_A$ the Weyl algebra in the variables $t_A$.
\end{definition}

\begin{example}
If $A = \mathbb{N}$, then $P = \sum_{a \in A} \partial_a \in \DA$, but $Q = \sum_{a \in A} t_a \notin \DA$. 
\end{example}

The Weyl algebra $\mathcal{D}_A$ is a filtered algebra, but not a graded algebra. Recall that an (exhaustive ascending) filtration on $\mathcal{D}_A$ is an increasing sequence of subspaces $F_i \mathcal{D}_A \subseteq \mathcal{D}_A$, for $i \in \mathbb{N}$, such that
\begin{equation}
\{0\} \subseteq F_0 \mathcal{D}_A \subseteq F_1 \mathcal{D}_A \subseteq F_2 \mathcal{D}_A \subseteq \ldots \subseteq \mathcal{D}_A,
\end{equation}
with $ \cup_{i \in \mathbb{N}} F_i \mathcal{D}_A = \mathcal{D}_A$ and $F_i \mathcal{D}_A \cdot F_j \mathcal{D}_A \subseteq F_{i+j} \mathcal{D}_A$ for all $i,j \in \mathbb{N}$.

It turns out that there are many filtrations on the Weyl algebra. The one that we will be interested in is called the Bernstein filtration.

\begin{definition}\label{d:bernstein}
The \emph{Bernstein filtration} on $\DA$ is the exhaustive ascending filtration
 which  is defined by assigning degree one to all variables $ t_a $ and partial derivatives $ \partial_a $ and defining the subspaces $ F_i \mathcal D_A $ to contain all operators in $ \mathcal D_A $ of degree $ \leq i $. Mathematically, we can write
$F_i \mathcal{D}_A$ as
\begin{equation}
 F_i \mathcal{D}_A = \left \{ \sum_{\substack{m, k \in \mathbb{N} \\ m+k = i}} \sum_{a_1, \ldots, a_m \in A} p_{a_1 \cdots a_m}^{(k)}(t_A) \partial_{a_1} \cdots \partial_{a_m} \right \},
\end{equation}
where the $ p_{a_1 \cdots a_m}^{(k)}(t_A)$ are polynomials of degree $\leq k$.
 Here, $F_0 \mathcal{D}_A = \mathbb{C}$.
\end{definition}

\begin{exercise}
Show that this is an exhaustive ascending filtration of $\mathcal{D}_A$.
\end{exercise}

The Weyl algebra is filtered but not graded. However, there is a natural way to construct a graded algebra out of any filtered algebra, which is sometimes called the ``Rees construction''.

\begin{definition}\label{d:rees}
The \emph{Rees Weyl algebra} $\mathcal{D}_A^\hslash$ associated to $\mathcal{D}_A$ with the Bernstein filtration is
\begin{equation}
\mathcal{D}_A^\hslash = \bigoplus_{n \in \mathbb{N}} \hslash^n F_n \mathcal{D}_A.
\end{equation}
\end{definition}
It is easy to see that $\DAh$ is a graded algebra, graded by powers of $\hslash$.
 
  \begin{remark}
A word of caution:  our $\hslash$ differs from the $\hslash$ introduced in \cite{KS17} where Airy structures were originally defined. To connect the two conventions, one should start with the original definition of Airy structures in \cite{KS17}, rescale the variables as $x_i \mapsto \hslash^{1/2} t_i$, and then redefine $\hslash \mapsto \hslash^2$. With this transformation, the grading defined in \cite{KS17} becomes the natural $\hslash$-grading on the Rees algebra that we introduce here. Although the two conventions are ultimately equivalent, we find the introduction of $\hslash$ via the Bernstein filtration more natural and transparent.
 \end{remark}

When $A$ is countably infinite, we also need to introduce a technical condition on collections of differential operators in $\DAh$. The point is that we want to be able to take infinite linear combinations of operators $P_a$ without divergent sums appearing. To this end, we define the notion of a bounded collection of differential operators:
\begin{definition}\label{d:bounded}
Let $I$ be a finite or countably infinite index set, and $\{P_i\}_{i \in I}$ a collection of differential operators $P_i \in \DAh$ of the form
\begin{equation}
P_i = \sum_{n \in \mathbb N} \hslash^n \sum_{\substack{m,k \in \mathbb N \\ m+k = n} }\sum_{a_1,\ldots, a_m \in A} p^{(n,k)}_{i;a_1,\ldots, a_m}(t_A) \partial_{a_1}\ldots \partial_{a_m}.
\end{equation}
We say that the collection is \emph{bounded} if, for all fixed choices of indices $ a_1,\ldots, a_m, n  $ and $ k $, the polynomials  $ p^{(n,k)}_{i;a_1,\ldots, a_m}(t_A) $ vanish for all but finitely many indices $ i \in I $.
\end{definition} 
It is easy to see that for any bounded collection $\{P_i\}_{i \in I}$, linear combinations $\sum_{i \in I} c_i P_i$ for any $c_i \in \DAh$ are well defined operators in $\DAh$, regardless of whether $I$ is finite or countably infinite.
\begin{exercise}
Check this!
\end{exercise}

\subsection{Airy ideals}

Our original question was to dermine for what partition functions $Z$ does there exist a ``maximal'' set of differential operators that annihilate $Z$. Following the concepts introduced in the previous section, we would like these operators to live in the Rees Weyl algebra $\DAh$. However, as is usual in the theory of differential equations, it is better to formulate the question in terms of left ideals, instead of a specific collection of differential operators. This is because if $P Z =0$, then $Q P Z = 0$ for any differential operator $Q$. So it is more natural to talk about the left ideal in $\DAh$ that annihilates $Z$. This is called the ``annihilator ideal'' of $Z$.

\begin{definition}
Let $Z$ be a partition function as in \eqref{eq:pfFgncoeff}. The \emph{annihilator ideal} $\mathcal{I} = \text{Ann}_{\DAh}(Z)$ of $Z$ in $\DAh$ is the left ideal in $\DAh$ defined by
\begin{equation}
 \text{Ann}_{\DAh}(Z) = \{ P \in \DAh~|~ P Z = 0 \}. 
\end{equation}
\end{definition}

So we can reformulate our  original question as follows: when is the annihilator ideal of a partition function $Z$ in $\DAh$ as large as possible?

To this end, we define the notion of an ``Airy ideal'', which is a class of left ideals that will arise as annihilator ideals of partition functions.

 \begin{definition}\label{d:airy}
 Let $\mathcal{I} \subseteq \DAh$ be a left ideal. We say that it is an \emph{Airy ideal} if there exists a bounded generating set $\{H_a \}_{a \in A}$ for $\mathcal{I}$ such that:\footnote{We abuse notation slightly here. We say that $\mathcal{I}$ is generated by the $H_a$, even though in standard terminology the ideal generated by the $H_a$ should only contain finite linear combinations of the generators. Here we allow finite and infinite (when $A$ is countably infinite) linear combinations, which is allowed since the collection $\{H_a \}_{a \in A}$ is bounded.}
  \begin{enumerate}
    \item The operators $H_a$ take the form
    \begin{equation}\label{eq:form}
       H_a= \hslash \partial_a + O(\hslash^2).
    \end{equation}
    \item The left ideal $\mathcal{I}$ satisfies the property:
    \begin{equation}
      [ \mathcal{I}, \mathcal{I}] \subseteq \hslash^2 \mathcal{I}.
    \end{equation}
  \end{enumerate}
 \end{definition}    
 
\begin{remark}
We note that condition (2) is highly non-trivial. First, it is clear that, for any left ideal $\mathcal{I} \subseteq \DAh$, $[\mathcal{I}, \mathcal{I}] \subseteq \mathcal{I}$, by definition of a left ideal. Second, because of the definition of the Rees Weyl algebra, it is also clear that $[\mathcal{I},\mathcal{I}] \subseteq \hslash^2 \DAh$. However, the requirement that $[\mathcal{I},\mathcal{I}] \subseteq \hslash^2 \mathcal{I}$ is highly non-trivial; most left ideals do not satisfy this condition.

As a simple example of a left ideal that does not satisfy condition (2), take $A = \{1,2 \}$, and consider the left ideal $I$ generated by the two operators $H_1 = \hslash \partial_1$ and $H_2 = \hslash \partial_2 + \hslash^2 t_1$. Then $[H_1, H_2] = \hslash^3 = \hslash^2 \cdot \hslash$. While $\hslash \in \DAh$, and $\hslash^3 \in \mathcal{I}$ since $\hslash^3 = H_1 H_2 - H_2 H_1$, it is easy to see that $\hslash \notin \mathcal{I}$, and thus $\mathcal{I}$ does not satisfy condition (2).
\end{remark}

 Definition \ref{d:airy} proposes an answer to the question that we asked.

 \begin{definition}\label{d:airypf}
 Let  $Z$ be a partition function as in \eqref{eq:pfFgncoeff}. We say that $Z$ is an \emph{Airy partition function} if its annihilator ideal $\mathcal{I} = \text{Ann}_{\DAh}(Z) \subset \DA$ is an Airy ideal. 
 \end{definition}
 
 The two conditions in \eqref{d:airy} have a very concrete meaning:
 \begin{itemize}
\item Condition (1) is a precise statement of what it means for the annihilator ideal to be ``as large as possible''. 
\item Condition (2) is a necessary condition for the left ideal $\mathcal{I}$ to be the annihilator ideal of a partition function of the form \eqref{eq:pfFgncoeff}. To see this, suppose that $\I$ is the annihilator ideal of a partition function $Z$. Then $[\mathcal{I}, \mathcal{I} ] \cdot Z = 0$. From the definition of the Rees Weyl algebra, we also know that for all $P,Q \in \DAh$, $[P, Q ] = \hslash^2 S$ for some $S \in \DAh$. But since the action of $\hslash^2$ on $Z$ is just multiplication, we conclude that, for all $P,Q \in \mathcal{I}$, $[P,Q]  \cdot Z = \hslash^2 S \cdot Z = 0$, which implies that $S \cdot Z = 0$, that is, $S \in \mathcal{I}$. In other words, $[ \mathcal{I}, \mathcal{I} ] \subseteq \hslash^2 \mathcal{I}$.
\end{itemize}

We say that a partition function is Airy if its annihilator ideal is an Airy ideal; this is a characterization of partition functions for which the annihilator ideal is as large as it gets. The resulting differential constraints $H_a Z = 0$ give rise to recursion relations for the coefficients $F_{g,n}[k_1,\ldots,k_n]$ in $Z$. But two questions remain:
\begin{enumerate}
\item Are the recursion relations sufficient to fully reconstruct the partition function $Z$ uniquely?
\item Given an Airy ideal $\I$, is it always the case that it is the annihilator ideal of a partition function of the form \eqref{eq:pf}?
\end{enumerate}
In other words, given an Airy ideal $\I$, does there always exist a unique solution $Z$ of the form \eqref{eq:pfFgncoeff} to the constraints $\I Z = 0$? The main theorem in the theory of Airy ideals, which was first proved in \cite{KS17}, gives a positive answer to this question.

 \begin{theorem}\label{t:airy}
  Let $\mathcal{I} \subset \DAh$ be an Airy ideal. Then there exists a unique partition function $Z$ of the form \eqref{eq:pfFgncoeff} such that $\mathcal{I}$ is the annihilator ideal of $Z$ in $\DAh$.
 \end{theorem}
 
 In words, given any Airy ideal $\I$, there always exists a partition function $Z$ such that $\I Z =0$, and the resulting recursion relations can be used to fully reconstruct $Z$ uniquely.
 
\subsection{Sketch of the proof}

We will not prove Theorem \ref{t:airy} in these lecture notes, but rather sketch the main ideas of the proof. The original proof was provided in \cite{KS17}, but we follow here the approach of \cite{BCJ22}, which is perhaps more instructive.

The starting point of the proof is to work in the $\hslash$-adic completion $\DAhC$ of the Rees Weyl algebra $\DAh$:
\begin{equation}
\DAhC = \prod_{n \in \mathbb{N}} \hslash^n F_n \mathcal{D}_A.
\end{equation}
The difference between $\DAh$ and $\DAhC$ is that operators that are formal series in $\hslash$ are included in the latter, while the former only includes operators that are polynomial in $\hslash$. Thus $\DAh \subset \DAhC$.

Let $\I \subset \DAh$ be an Airy ideal, and $\widehat{\I} \subset \DAhC$ be the left ideal generated by the $\{H_a\}_{a \in A}$ in the completion $\DAhC$. Then it is straightforward to see that if $\I$ is the annihilator ideal in $\DAh$ for some partition function $Z$, then $\widehat{\I}$ is the annihilator ideal in $\DAhC$ of the same partition function $Z$.

The key realization then is that the ideal $\widehat{\I}$ in the completion $\DAhC$ is actually quite simple. What do we mean by that?  We proceed in four steps.

\begin{enumerate}
\item In the first step, we show that there exists a collection of operators $\{\bar{H}_a\}_{a\in A}$ in $\widehat{\I}$ of the form
\begin{equation}\label{eq:Hbar}
\bar{H}_a= \hslash \partial_a -\sum_{n=2}^\infty \hslash^n p_a^{(n)}(t_A) ,
\end{equation}
where the $p_a^{(n)}(t_A)$ are polynomials of degree $\leq n$. Note that those operators are generally much simpler than the $H_a$, which may include terms with derivatives at all orders in $\hslash$.

To show this, write the generators of $\widehat{\I}$ as $H_a = \hslash \partial_a + P_a$, with  $P_a = O (\hslash^2)$. Pick one of them, say $H_a$, and consider first the terms in $P_a$ of order $\hslash^2$. There are two types of terms (we assume everything is normal ordered, with derivatives on the right of variables): terms that have no derivatives (``polynomial terms'') and terms with at  least one derivative. Keep the terms in the first class untouched, and for all terms in the second class, replace the right-most derivatives $\hslash \partial_k$ by $H_k - P_k$. This creates new terms with a $H_k$ on the right (those are in the left ideal $\widehat{\I}$), and new terms of order $\hslash^3$. We then do the same thing at order $\hslash^3$, and so on and so forth. As we are now working in the completed algebra, where formal power series in $\hslash$ are included, we can keep doing this process inductively, and we end up with the statement that
\begin{equation}
H_a = \hslash \partial_a -\sum_{n=2}^\infty \hslash^n p_a^{(n)}(t_A) + Q_a,
\end{equation}
where the $p_a^{(n)}(t_A)$ are polynomials of degree $\leq n$, and $Q_a \in \widehat{\I}$. Since $H_a \in \widehat{\I}$ and $Q_a \in \widehat{\I}$, we conclude that $\bar{H}_a := \hslash \partial_a -\sum_{n=2}^\infty \hslash^n p_a^{(n)}(t_A) $ is also in $\widehat{\I}$. Of course, we can do that for all $H_a$ in the generating set.

Note that for this process to work, it is very important that we work in the completion $\DAhC$, so that the inductive process can keep going forever.

\item In the second step, we show that, not only are the simpler operators $\bar{H}_a$ in $\widehat{\I}$, but, in fact, the collection $\{ \bar{H}_a \}_{a \in A}$ generates $\widehat{\I}$. In other words, we can write any element in $\widehat{\I}$ as a (potentially infinite) linear combination of the $\bar{H}_a$.

\item The third step consists in showing that, while the operators $H_a$ do not necessarily commute, the $\bar H_a$ do. That is, $[ \bar H_a, \bar H_b ] =0$ for all $a,b \in A$. To show this, we use the condition $[\widehat{\I}, \widehat{\I}] \subseteq \hslash^2 \widehat{\I}$, which implies that there are no non-zero operators in $\widehat{\I}$ that are formal power series in $\hslash$ with polynomial coefficents (no derivatives).

\item In the fourth step, we use Poincare lemma to show that the condition $[\bar{H}_a, \bar{H}_b] = 0$ implies that we can rewrite the polynomials $p_a^{(n)}(t_A)$ as derivatives of a single polynomial $q^{(n+1)}(t_A)$; that is,
\begin{equation}\label{eq:hbarfinal}
\bar{H}_a = \hslash \partial_a -\sum_{n=2}^\infty \hslash^n \partial_a q^{(n+1)}(t_A).
\end{equation}
\end{enumerate}

The result is that the ideal $\widehat{\I}$ in the completed Rees Weyl algebra is very simple; it is generated by the collection of operators $\bar{H}_a$ from \eqref{eq:hbarfinal}. Using the standard action of derivatives on exponentials, it is then clear that the $\bar{H}_a$ annihilate the partition function $Z$ of the form
\begin{equation}\label{eq:pp}
Z = \exp \left(\sum_{n=2}^\infty \hslash^{n-1} q^{(n+1)}(t_A)\right).
\end{equation}
More precisely, $\widehat{\I}$ is the annihilator ideal of $Z$ in $\DAhC$. It then follows that $H_a Z = 0$ as well, for all $a \in A$, since $H_a \in \widehat{\I}$. Furthermore, $Z$ is uniquely determined, after fixing the initial conditions $q^{(n+1)}(0) = 0$.  This essentially concludes the proof of the theorem. Tadam!

\begin{remark}
We note that we can generalize the definition of Airy ideals, Definition \ref{d:airy}, slightly. We can include in the $H_a$ extra terms of degree $\hslash^1$ of the form $\hslash p _a^{(1)}(t_A)$ for some linear polynomials $p_a^{(1)}(t_A)$. Everything goes through, with the caveat that the sum over $n$ in the resulting partition function \eqref{eq:pp} starts at $n=1$. In the notation of \eqref{eq:pfFgncoeff}, this means that unstable terms with $2g-2+n = 0$ are included.

We could also try to include degree $\hslash^0$ terms in the $H_a$, but this is a lot more subtle. Except for rather trivial cases, introducing degree zero terms requires refining the construction of Airy ideals to work over formal power series in these coefficients. This is related to the construction of Whittaker states, which are relevant for supersymmetric gauge theories, via Airy ideals in \cite{BBCC21,BCU24}.
\end{remark}

\subsection{Examples}

A definition is interesting only if it has non-trivial examples. Perhaps there are no left ideal in $\DAh$ that are Airy ideals! How boring would that be! Fortunately, it turns out that there many interesting Airy ideals, and that, in fact, many of the differential constraints already known in enumerative geometry fall into the Airy framework.

We could start by searching for Airy ideals in $\DAh$ with $A$ a finite index set (i.e., the Rees Weyl algebra in a finite number of variables). In fact, an interesting question is to classify all Airy ideals that can be realized in this way. See for instance \cite{ABCO17,HR19}.

However, most interesting examples arise when $A$ is a countably infinite index set, and hence this is what we focus on in this section.

\subsubsection{Kontsevich-Witten Virasoro constraints}
\label{s:KW}

The prototypical example of an Airy  ideal comes from the Virasoro constraints satisfied by the Kontsevich-Witten partition function originally formulated by Dikgraaf, Verlinde and Verlinde in \cite{DVV91}.

Let $A$ be the set of odd natural numbers, and define the following differential operators in $\DAh$:
\begin{equation}\label{eq:KW}
L_{k} = \hslash J_{2k+3}- \hslash^2 \left( \frac{1}{2}\sum_{\substack{m_1+m_2 = 2 k\\ m_1, m_2 \text{ odd}}} : J_{m_1} J_{m_2}: + \frac{1}{8}\delta_{k,0} \right), \qquad k \geq -1,
\end{equation}
where, for $m \geq 1$, 
\begin{equation}
J_m =\partial_m, \qquad J_{-m} = m t_m,
\end{equation}
and $: \quad :$ denotes normal ordering, where we put the $J_k$ with positive $k$'s to the right of those with negative $k$'s (i.e. derivatives to the right).

First, one can show that the left ideal $\I \subset \DAh$ generated by the collection $\{L_k\}_{k \geq -1}$ is an Airy ideal. Indeed, by direct calculation, we get that
\begin{equation}\label{eq:vir1}
[L_m, L_n] = \hslash^2 (m-n) L_{m+n}, \qquad m,n \geq -1.
\end{equation}
That is, the differential operators above form a representation of a subalgebra of the Virasoro algebra (appropriately rescaled by $\hslash$ -- see Section \ref{s:uea}). 

\begin{exercise}
Show that the operators in \eqref{eq:KW} satisfy the Virasoro commutation relations \eqref{eq:vir1}.
\end{exercise}

A consequence of \eqref{eq:vir1} is that the ideal $\I$ satisfies condition (2) of Definition \ref{d:airy}, namely $[\I,\I ] \subseteq \hslash^2 \I$. Indeed, pick any two $P, Q \in \I$. We can write $P = \sum_{a \in A} c_a L_a$ and $Q = \sum_{b \in A} d_b L_b$. Calculating the commutator, we get
\begin{align}\label{eq:calc}
[P,Q] =& [ \sum_{a \in A} c_a L_a, \sum_{b \in A} d_b L_b ]\\
=&\sum_{a,b \in A} \left( c_a d_b [L_a, L_b]+ c_a [ L_a, d_b] L_b + d_b [c_a, L_b] L_a + [c_a, d_b] L_b L_a \right).
\end{align}
The last three terms on the RHS are clearly in $\hslash^2 \I$, and the first term is also in $\hslash^2 \I$ because of \eqref{eq:vir1}.

As for condition (1), looking at \eqref{eq:KW} we see that
\begin{equation}
L_k = \hslash \partial_{2k+3} + O(\hslash^2), \qquad k \geq -1,
\end{equation}
and as we are working over the set $A$ of odd natural numbers, condition (1) is also satisfied (after trivially redefining $H_{2k+3} = L_k$). It follows that $\I$ is an Airy ideal.\footnote{To be precise, we should also check that the collection $\{ L_k \}_{k \geq -1}$ is bounded, which is easy to show.}

Now that we know that the collection $\{ L_k \}_{k \geq -1}$ from \eqref{eq:KW} generates an Airy ideal, Theorem \ref{t:airy} gives us for free that there is a unique solution $Z$ of the form \eqref{eq:pfFgncoeff} to the constraints
\begin{equation}
L_k Z = 0, \qquad k \geq -1.
\end{equation}
One can see that these Virasoro constraints (after trivial index relabeling) are the Virasoso constraints from \cite{DVV91}. Since the solution is unique, this means that $Z$ is nothing but the Kontsevich-Witten tau-function of the KdV integrable hierarchy (with a suitable choice of normalization), written explicitly as
\begin{equation}
Z = \exp\left( \sum_{\substack{g \in  \mathbb{N}, n \in \mathbb{N}^* \\ 2g-2+n > 0}} \frac{\hslash^{2g-2+n}}{n!}  \sum_{k_1, \ldots, k_n \in \mathbb{N}}  F_{g,n}[2k_1+1, \ldots,2k_n+1]  t_{2k_1+1} \cdots t_{2k_n+1}\right)
\end{equation}
with the coefficients given by
\begin{equation}\label{eq:fgnairy}
F_{g,n}[2k_1+1, \ldots,2k_n+1] = \prod_{i=1}^n (2 k_i+1)!! \int_{\overline{\mathcal{M}}_{g,n}} \psi_1^{k_1} \cdots \psi_n^{k_n}.
\end{equation}
The Virasoro constraints $L_k Z = 0$, $k \geq -1$, therefore determine intersection numbers over $\overline{\mathcal{M}}_{g,n}$ uniquely recursively, which is a well known consequence of Witten's conjecture.

\subsubsection{BGW Virasoro constraints}
\label{s:BGW}

We can do a small variation on the previous example. Let $A$ still be the set of odd natural numbers, but define the following differential operators in $\DAh$:
\begin{equation}\label{eq:BGW}
L_{k} = \hslash J_{2k+1}- \hslash^2 \left( \frac{1}{2}\sum_{\substack{m_1+m_2 =2 k\\ m_1, m_2 \text{ odd}}} : J_{m_1} J_{m_2}: + \frac{1}{8}\delta_{k,0} \right), \qquad k \geq 0.
\end{equation}
Note that we changed the $O(\hslash^1)$ terms, and also we kept only the operators with $k \geq 0$. One can show that those operators still form a representation of the Virasoro algebra:
\begin{equation}
[L_m, L_n] = \hslash^2 (m-n) L_{m+n}, \qquad m,n \geq 0.
\end{equation}
It then follows, as above, that they generate an Airy ideal in $\DAh$ (with $H_{2k+1} = L_k$), and hence there is a unique partition function $Z$ of the form \eqref{eq:pfFgncoeff} satisfying the constraints $L_k Z = 0$, $k \geq 0$. 

\begin{exercise}
In this example, show that the constraints $L_k Z = 0$, $k \geq 0$ imply that $F_{0,n}[2k_1+1, \ldots, 2k_n+1] =0$ for all $n \in \mathbb{N}^*$.
\end{exercise}

What is this partition function $Z$? Can we write the coefficients $F_{g,n}[k_1, \ldots, k_n]$ in terms of intersection numbers over $\overline{\mathcal{M}}_{g,n}$? Already in this simple example the answers to these equations are quite interesting!

It turns out that $Z$ is still a tau-function of the KdV integrable hierarchy, but a different one: the BGW tau-function (it provides a different solution of the integrable system, satisfying different initial conditions). Moreover, the enumerative geometric interpretation of the coefficients $F_{g,n}[k_1,\ldots,k_n]$ was only obtained recently by Norbury in \cite{No17}, and involves constructing new cohomology classes over $\overline{\mathcal{M}}_{g,n}$. Fascinating!

More precisely, the geometric statement is that
\begin{equation}\label{eq:norbury}
F_{g,n}[2k_1+1, \ldots, 2k_n+1] = \prod_{i=1}^n (2k_i+1)!! \int_{\overline{\mathcal{M}}_{g,n}} \Theta_{g,n} \psi_1^{k_1} \cdots \psi_n^{k_n},
\end{equation}
where $\Theta_{g,n}$ is a cohomology class on $\overline{\mathcal{M}}_{g,n}$ now known as the Norbury class (it is a particular case of a construction due to Chiodo \cite{Ch06}). The Virasoro constraints  $L_k Z = 0$, $k \geq 0$ determine these intersection numbers uniquely, by recursion on $2g-2+n$. The $\Theta_{g,n}$ cohomology classes form a cohomological field theory --- see the Les Houches lecture notes \cite{GL24} for a brief introduction to cohomological field theories and more on the Norbury cohomological field theory. 

As we will see, this is only the first example of a large class of interesting Airy ideals that appear as ``building blocks'' of the general theory. Much remains to be understood for these building blocks, as far as the relations with integrable systems and enumerative geometry go.

\subsubsection{Quadratic operators}

In both previous examples, the operators that generate the Airy ideal are quadratic in $\hslash$. Let us now look at the general case of Airy ideals generated by operators that are quadratic in $\hslash$. In fact, this is the original case that was studied by Kontsevich and Soibelman in \cite{KS17}.

Let $A$ be an index set. We write general operators in $\DAh$ that are quadratic in $\hslash$ and that satisfy condition (1) of Definition \ref{d:airy} as:
\begin{equation}
H_a = \hslash \partial_a -\hslash^2 \left(  \frac{1}{2} A_{abc} t_b t_c + B_{abc} t_b \partial_C + \frac{1}{2}C_{abc} \partial_b \partial_c + D_a \right).
\end{equation}
In this section we use the convenient notation that repeated indices are summed over the index set $A$.
Note that we could also include terms that are linear in $t$'s and $\partial$'s in the $O(\hslash^2)$ term, but for simplicity we omit them, as was originally done in \cite{KS17}.

Let $\I$ be the left ideal generated by $\{H_a \}_{a \in A}$. Condition (1) of Definition \ref{d:airy} is satisfied by construction. We also assume that the collection $\{H_a\}_{a \in A}$ is bounded. What about condition (2)? From the calculation in \eqref{eq:calc}, it is easy to see that $[\I,\I] \subseteq \hslash^2 \I$ if and only if 
\begin{equation}\label{eq:Hab}
[H_a, H_b] = \hslash^2 f_{abc} H_c
\end{equation}
for some $f_{abc} \in \DAh$. But since the $H_a$ are quadratic in $\hslash$, it follows that the $f_{abc}$ must be  $O(\hslash^0)$. That is, $f_{abc} \in \mathbb{C}$. In other words, condition (2) is satisfied if and only if the $H_a$ form a representation of a Lie algebra! Note that this only holds for operators that are quadratic in $\hslash$; if we include higher order terms, Airy ideals do not have to come from representations of Lie algebras.

Assuming that \eqref{eq:Hab} holds, from Theorem \ref{t:airy} we know that there is a unique partition function $Z$ of the form \eqref{eq:pfFgncoeff} that satisfies the constraints $H_a Z = 0$, $a \in A$. Its coefficients $F_{g,n}[k_1, \ldots, k_n]$ are determined recursively. In fact, we can write down the recursion explicitly here. It reads \cite{ABCO17}:
\begin{align}\label{eq:trA}
F_{g,n}&[k_1, \ldots, k_n]= \sum_{m=2}^\infty B_{k_1 k_m a} F_{g,n-1} [a, k_2, \ldots, \widehat{k}_m, \ldots, k_n] \nonumber\\
& + \frac{1}{2} C_{k_1 a b} \left[F_{g-1,n+1}[a,b, k_2, \ldots,k_n] + \sum_{\substack{g_1+g_2=g \\ I \cup J = \{k_2, \ldots, k_n\}}} F_{g_1, |I|+1}[a,I] F_{g_2, |J|+1} [b, J] \right],
\end{align}
where the notation $\widehat{k}_m$ means that $k_m$ is omitted, and in the last sum we sum over all possibly empty disjoint subsets $I,J \subseteq \{ k_2,\ldots, k_n \}$ such that $I \cup J =  \{ k_2 ,\ldots, k_n \}$. The initial conditions of the recursion are
\begin{equation}\label{eq:trB}
F_{0,3}[a,b,c] = A_{abc}, \qquad F_{1,1}[a] = D_a.
\end{equation}

\begin{exercise}
Obtain the recursion formula \eqref{eq:trA} and the initial conditions \eqref{eq:trB} from the differential constraints $H_a Z = 0$.
\end{exercise}

This is our first take on ``topological recursion''; in the context of Airy ideals, it appears simply as the recursive structure satisfied by the coefficients $F_{g,n}[k_1,\ldots,k_n]$ as a consequence of the differential constraints $\I Z = 0$. We wrote it explicitly here for the case of operators that are quadratic in $\hslash$, but it can of course be written as well for operators of any order in $\hslash$ (see \cite{BBCCN18}).

\begin{remark}
The astute reader may ask the following question. What if we start directly from the topological recursion of \eqref{eq:trA}, with the initial conditions in \eqref{eq:trB}? Doesn't it construct a partition function $Z$ recursively that solves the constraints $H_a Z = 0$, \emph{for any choice of tensors} $A_{abc}$, $B_{abc}$, $C_{abc}$ and $D_a$? That would be rather strange, as it would mean that any collection of quadratic operators $\{H_a \}_{a \in A}$ satisfies the condition $[\I,\I]\subseteq \hslash^2 \I$ in the definition of Airy ideals. Surely that can't be right!

The subtlety is in the recursive formula \eqref{eq:trA}. For this formula to reconstruct a partition function $Z$, it must produce coefficients $F_{g,n}[k_1,\ldots, k_n]$ that are \emph{symmetric} under permutations of the indices $k_1,\ldots,k_n$. This is highly non-trivial! The formula is manifestly symmetric under permutations that leave $k_1$ fixed, but it is far from obviously symmetric for permutations that involve $k_1$. In fact, the conditions on the tensors $A_{abc}$, $B_{abc}$, $C_{abc}$ and $D_a$ that are necessary for the recursion \eqref{eq:trA} to produce symmetric $F_{g,n}[k_1,\ldots,k_n]$ are quite involved; see \cite{ABCO17} for details. Nevertheless, it turns out that symmetry is implied by the condition $[\I,\I] \subseteq \hslash^2 \I$, as it should, and everything is well.
\end{remark}

%

\subsubsection{Universal enveloping algebras and $\mathcal{W}$-algebras}
\label{s:uea}

Our previous examples of Airy ideals are constructed via representations of Lie algebras. This is not a coincidence; many interesting examples of Airy ideals can be constructed in this way, or more generally as representations of non-linear Lie algebras such as $\mathcal{W}$-algebras. Let me explain the construction in more general terms. This subsection is a little more technical, and may be skipped.

Recall that a Lie algebra $\mathfrak{g}$ is a vector space $\mathfrak{g}$ with an alternating bilinear map $[,]: \mathfrak{g} \times \mathfrak{g} \to \mathfrak{g}$. For the purpose of the construction, we can be more general and allow $\mathfrak{g}$ to be a ``non-linear Lie algebra'' (see for instance Section 3 of \cite{DK05} for a precise definition). Roughly speaking, a non-linear Lie algebra is like a Lie algebra, but the bilinear map now takes the form $[,]: \mathfrak{g} \times \mathfrak{g} \to \mathcal{T}(\mathfrak{g})$, where $\mathcal{T}(\mathfrak{g})$ is the tensor algebra over $\mathfrak{g}$.\footnote{The definition is actually a little more involved than this. In particular, one needs to impose a grading condition on the algebra and the bilinear map. We refer the reader to Definition 3.1 of \cite{DK05} for a precise definition.} In other words, if $\{ e_i \}$ is a basis for $\mathfrak{g}$, in a Lie algebra, the commutator $[p,q]$ for any $p,q \in \mathfrak{g}$ is a linear combination of the $e_i$, while in a non-linear Lie algebra $[p,q]$ is a ``polynomial'' in the $e_i$ (with product given by tensor product). Many of the properties of Lie algebras carry through to non-linear Lie algebras; for instance, one can construct the universal enveloping algebra $U(\mathfrak{g})$ as usual, and $\mathfrak{g}$ is a PBW basis for $U(\mathfrak{g})$ (see \cite{DK05}). 

Let $\mathfrak{g}$ be a Lie algebra or a non-linear Lie algebra, and $U(\mathfrak{g})$ the universal enveloping algebra. Suppose that there is an exhaustive ascending filtration 
\begin{equation}
\{0\} \subseteq F_0 U(\mathfrak{g}) \subseteq F_1 U(\mathfrak{g}) \subseteq \ldots \subseteq U(\mathfrak{g}).
\end{equation}
Then we can construct the Rees universal enveloping algebra $U^\hslash(\mathfrak{g}) = \bigoplus_{n \in \mathbb{N}} \hslash^n F_n U(\mathfrak{g})$ using the Rees construction as in Definition \ref{d:rees}. We assume that the filtration is such that
\begin{equation}\label{eq:filtcond}
[ F_i U(\mathfrak{g}), F_j U(\mathfrak{g})] \subseteq F_{i+j-2} U(\mathfrak{g}),
\end{equation}
in which case
\begin{equation}
[U^\hslash(\mathfrak{g}), U^\hslash(\mathfrak{g})] \subseteq \hslash^2 U^\hslash(\mathfrak{g}).
\end{equation}
To construct Airy ideals, we proceed as follows:
\begin{lemma}\label{l:uea}
Let $\rho: U^\hslash(\mathfrak{g}) \to \DAh$ be a representation of the Rees enveloping algebra in the Rees Weyl algebra, for some index set $A$. Let $\mathcal{I}_{U^\hslash} \subseteq U^\hslash(\mathfrak{g})$ be a left ideal in $U^\hslash(\mathfrak{g})$, and $\mathcal{I} = \DAh \rho(\mathcal{I}_{U^\hslash}) \subseteq \DAh$ be the corresponding left ideal in $\DAh$ generated by $\rho(\mathcal{I}_{U^\hslash})$.

Suppose that  $\mathcal{I}_{U^\hslash}$ satisfies the property $[\mathcal{I}_{U^\hslash},\mathcal{I}_{U^\hslash}] \subseteq \hslash^2 \mathcal{I}_{U^\hslash}$, and that there exists a generating set $\{H_a \}_{a \in A}$ for $\mathcal{I}_{U^\hslash}$ such that $\rho(H_a) = \hslash \partial_a + O(\hslash^2)$ and the collection $\{\rho(H_a)\}_{a \in A}$ is bounded. Then $\mathcal{I}$ is an Airy ideal.
\end{lemma}

\begin{proof}
This is clear. It is easy to show that the condition $[\mathcal{I}_{U^\hslash},\mathcal{I}_{U^\hslash}] \subseteq \hslash^2 \mathcal{I}_{U^\hslash}$ implies that $[\mathcal{I}, \mathcal{I} ] \subseteq \hslash^2 \mathcal{I}$, and since the set $\{ \rho(H_a) \}_{a \in A}$ generates $\mathcal{I}$, we conclude that $\mathcal{I}$ is an Airy ideal.
\end{proof}

In this construction, we see that the two conditions in the definition of Airy ideals are realized independently. On the one hand, the condition $[\mathcal{I}_{U^\hslash},\mathcal{I}_{U^\hslash}] \subseteq \hslash^2 \mathcal{I}_{U^\hslash}$ is a condition on the left ideal $\mathcal{I}_{U^\hslash} \subseteq U^\hslash(\mathfrak{g})$ in the Rees universal enveloping algebra, which is independent of the choice of representation $\rho: U^\hslash(\mathfrak{g}) \to \DAh$. On the other hand, the condition that there exists a generating set with $\rho(H_a) = \hslash \partial_a  + O(\hslash^2)$ is very much representation-dependent.

The condition $[\mathcal{I}_{U^\hslash},\mathcal{I}_{U^\hslash}] \subseteq \hslash^2\mathcal{I}_{U^\hslash}$ is in fact fairly easy to satisfy, as it is naturally obtained from left  ideals in $U(\mathfrak{g})$ as follows. We first define an operation that maps elements of $U(\mathfrak{g})$ to elements of $U^\hslash(\mathfrak{g})$ (see for instance Section 1.2 in \cite{SST99}). 

\begin{definition}\label{d:reesation}
Let 
 $p \in U(\mathfrak{g})$, and let  $i = \min\{ k \in \mathbb{N}~|~ p \in F_k U(\mathfrak{g}) \}$. We define the \emph{homogenization} $h(p)$ of $p$ to be $h(p) = \hslash^i p \in U^\hslash(\mathfrak{g})$. We define the homogenization $h(\mathcal{I}_U)$ of a left ideal $\mathcal{I}_U \subseteq U(\mathfrak{g})$ to be the left ideal in $U^\hslash(\mathfrak{g})$ generated by all homogenized elements $h(p)$, $p \in \mathcal{I}_U$. 
 \end{definition}
 
 Then we have the following simple lemma:
 
\begin{lemma}
Let $\mathcal{I}_U \subseteq U(\mathfrak{g})$ be a left ideal. Then its homogenization $h(\mathcal{I}_U) \subseteq U^\hslash(\mathfrak{g})$ satisfies $[ h(\mathcal{I}_U), h(\mathcal{I}_U)] \subseteq \hslash^2 h(\mathcal{I}_U)$.
\end{lemma}

\begin{proof}
For any $p,q \in \mathcal{I}_U$, $[p,q] = s$ for some $s \in \mathcal{I}_U$. Moreover, by our requirement \eqref{eq:filtcond}, we know that if $p \in F_i U(\mathfrak{g})$ and $q \in F_j U(\mathfrak{g})$, then $s \in F_{i+j-2} U(\mathfrak{g})$. This means that the homogenizations satisfy $[h(p),h(q)] = \hslash^k h(s)$ for some $k \geq 2$. Since $h(\mathcal{I}_U)$ is generated by all homogenized elements $h(p)$, $p \in \mathcal{I}_U$, it is easy to show that this implies that $[h(\mathcal{I}_U), h(\mathcal{I}_U)] \subseteq \hslash^2 h(\mathcal{I}_U)$.
\end{proof}

Thus any left ideal $\mathcal{I}_{U^\hslash} \subseteq U^\hslash(\mathfrak{g})$ that is obtained as the homogenization of a left ideal in $U(\mathfrak{g})$ automatically satisfies $[\mathcal{I}_{U^\hslash},\mathcal{I}_{U^\hslash}] \subseteq \hslash^2\mathcal{I}_{U^\hslash}$. This perhaps helps demystify the meaning of this condition.

\begin{remark}
One has to be careful however; if $\{p_a\}_{a \in A}$ is a generating set for a left ideal $\mathcal{I}_U \subseteq U(\mathfrak{g})$, then it is not necessarily true that the homogenizations $\{h(p_a)\}_{a \in A}$ form a generating set for the homogenization of the left ideal $h(\mathcal{I}_U) \subseteq U^\hslash(\mathfrak{g})$ (see Proposition 1.2.11 in \cite{SST99}). Let us illustrate this with an example.
Let $\mathfrak{g}$ be the Virasoro algebra, spanned by the modes $\{L_m\}_{m \in \mathbb{Z}}$ and the central charge $c$. Consider the left ideal $\mathcal{I}_U \subseteq U(\mathfrak{g})$ generated by $\{L_1, L_2 \}$. Since $[L_2,L_1] = L_3 \in \mathcal{I}_U$, $[L_3,L_1] = 2 L_4 \in \mathcal{I}_U$, and so on and so forth, we know that $L_m \in \mathcal{I}_U$ for all $m \geq 1$. We use the filtration by conformal weight, where $L_m \in F_2 U(\mathfrak{g})$ for all $m \in \mathbb{Z}$, and thus $h(L_m) = \hslash^2 L_m$. Let $h(\mathcal{I}_U) \subseteq U^\hslash(\mathfrak{g})$ be the homogenization of the left ideal. Then $h(L_m)=\hslash^2 L_m \in h(\mathcal{I}_U)$ for all $m \geq 1$. This means that $h(\mathcal{I}_U)$ is not generated by $\{ h(L_1), h(L_2) \}=\{\hslash^2 L_1, \hslash^2 L_2\}$, since for instance $h(L_3) = \hslash^2 L_3$ is not in the ideal generated by $\{h(L_1), h(L_2) \}$ in $U^\hslash(\mathfrak{g})$ --- indeed, $[h(L_2), h(L_1)] = \hslash^2 h(L_3)$, which means that $\hslash^2 h(L_3)$ is in the  ideal generated by $\{h(L_1), h(L_2) \}$, but not $h(L_3)$.
\end{remark}

Nevertheless, this gives a clear recipe on how to obtain Airy ideals from universal enveloping algebras.

\begin{enumerate}
\item We start with a left ideal $\mathcal{I}_U \subseteq U(\mathfrak{g})$ or, equivalently, a cyclic left module $M \simeq U(\mathfrak{g})/ \mathcal{I}_U$ generated by a vector $v$ whose annihilator is $\mathcal{I}_U = \text{Ann}_{U(\mathfrak{g})}(v)$.
\item We construct the homogenization $\mathcal{I}_{U^\hslash}= h(\mathcal{I}_U)$, which is a left ideal in $U^\hslash(\mathfrak{g})$.  By construction, we know that $[\mathcal{I}_{U^\hslash}, \mathcal{I}_{U^\hslash}] \subseteq \hslash^2 \mathcal{I}_{U^\hslash}$. From the point of view of modules, we obtain a cyclic left module $M[\hslash] \simeq U^\hslash(\mathfrak{g})/\mathcal{I}_{U^\hslash}$ generated by the vector $v$ and where $\hslash$ acts by multiplication; the annihilator of $v$ in $U^\hslash(\mathfrak{g})$ is $\mathcal{I}_{U^\hslash} = \text{Ann}_{U^\hslash(\mathfrak{g})}(v)$.
\item We find a representation $\rho: U^\hslash(\mathfrak{g}) \to \DAh$, for some index set $A$, such that there exists a  generating set $\{H_a\}_{a \in A}$ for $\mathcal{I}_{U^\hslash}$ with $\rho(H_a) = \hslash \partial_a + O(\hslash^2)$ and the collection $\{\rho(H_a)\}_{a \in A}$ bounded.
\end{enumerate}
By Lemma \ref{l:uea}, the left ideal $\mathcal{I} \subseteq \DAh$ generated by $\{\rho(H_a)\}_{a \in A}$ is an Airy ideal.

The examples presented previously were obtained in this way.

\begin{example}
Let $\mathfrak{g}$ be the Virasoro algebra with central charge $c$. $\mathfrak{g}$ is spanned by the modes $\{L_m \}_{m \in \mathbb{Z}}$ and the central charge $c$, with commutation relations
\begin{equation}
[L_m, L_n ] = (m-n) L_{m+n} + \frac{c}{12}(m^3-m) \delta_{m,-n}.
\end{equation}
There is a natural filtration on $U(\mathfrak{g})$ that satisfies \eqref{eq:filtcond}, where the $L_m$ have degree $2$ and $c$  has degree $0$ (this is the filtration by conformal weight). The homogenizations are then $h(L_m) = \hslash^2 L_m$ and $h(c) = c$. Thus, in $U^\hslash(\mathfrak{g})$, we have
\begin{equation}
[h(L_m), h(L_n)] = \hslash^2 (m-n) h(L_{m+n}) + \hslash^4 \frac{h(c)}{12}(m^3-m) \delta_{m,-n}.
\end{equation}

Consider the left ideal $\mathcal{I}_U \subseteq U(\mathfrak{g})$ generated by $\{L_m\}_{m \geq -1}$. Then $U(\mathfrak{g})/\mathcal{I}_U$ is the cyclic module generated by the vacuum vector $v$, which is such that $L_m v = 0$ for all $v \geq -1$. It is easy to see that the homogenization $\mathcal{I}_{U^\hslash} = h(\mathcal{I}_U) \subseteq U^\hslash(\mathfrak{g})$ is the left ideal generated by $\{h(L_m)\}_{m \geq -1}$. By construction, it satisfies $[\mathcal{I}_{U^\hslash}, \mathcal{I}_{U^\hslash}] \subseteq \hslash^2 \mathcal{I}_{U^\hslash}$. To get the Kontsevich-Witten Virasoro constraints from Section \ref{s:KW}, we construct the representation of $U^\hslash(\mathfrak{g})$ in \eqref{eq:KW} (in the language of this section, the $L_k$ from \eqref{eq:KW} should be $h(L_k)$ in the Rees universal enveloping algebra), which exists for central charge $c=1$.
\end{example}

\begin{example}
The BGW example from Section \ref{s:BGW} is very similar. The Lie algebra $\mathfrak{g}$ is the same, with the same filtration. We now look at the left ideal $\mathcal{I}_U$ in $U(\mathfrak{g})$ generated by $\{L_m \}_{m \geq 0}$, and construct its homogenization $h(\mathcal{I}_U)$, which is generated by $\{ h(L_m) \}_{m \geq 0}$. The BGW Virasoro constraints are obtained via the representation \eqref{eq:BGW} for the $h(L_m)$, which again exists for central charge $c=1$.

This construction could be generalized a little bit. We could have started with the left ideal $\mathcal{I}_U$ in $U(\mathfrak{g})$ generated by $\{L_m\}_{m \geq 1}$ and $L_0 - \Lambda$ for some $\Lambda \in \mathbb{C}$. Then $U(\mathfrak{g})/ \mathcal{I}_U$ is a highest weight module for the Virasoro algebra with highest weight $\Lambda$, generated by a highest weight vector $v$ that satisfies $L_m v =0$ for $m \geq 1$ and $L_0 v = \Lambda v$. Its homogenization $\mathcal{I}_{U^\hslash}=h(\mathcal{I}_U)$ is generated by $h(L_m) = \hslash^2 L_m$ for $m \geq 1$ and $h(L_0 - \Lambda) = \hslash^2(L_0 - \Lambda)$. Again, by construction $[\mathcal{I}_{U^\hslash}, \mathcal{I}_{U^\hslash}] \subseteq \hslash^2 \mathcal{I}_{U^\hslash}$. Using the same representation \eqref{eq:BGW} for the $h(L_m)$, we obtain an Airy ideal. As a result, we obtain a partition function $Z$ uniquely fixed by the constraints  $h(L_m) Z = 0$ for $ m\geq 1$ and $h(L_0) Z = \hslash^2 \Lambda Z$, with the $h(L_m)$ represented as in \eqref{eq:BGW}. Explicitly, this gives the constraints
\begin{align}
 \hslash J_{2k+1} Z =& \frac{\hslash^2}{2}\sum_{\substack{m_1+m_2 =2 k\\ m_1, m_2 \text{ odd}}} : J_{m_1} J_{m_2}:  Z \qquad \text{for $k \geq 1$,}\\
 \hslash J_1 Z =& \hslash^2 \left(\frac{1}{2}\sum_{\substack{m_1+m_2 =0\\ m_1, m_2 \text{ odd}}} : J_{m_1} J_{m_2}:+ \frac{1}{8} + \Lambda \right) Z.
\end{align}
The result is simply a shift of the constant term $1/8$ by the highest weight of the original module for the Virasoro algebra. This gives a new $\Lambda$-dependent partition function, which can be understood as the highest weight vector generating the highest weight module, with the BGW case corresponding to the zero weight $\Lambda=0$.
\end{example}

\begin{example}
So far all the examples started with a Lie algebra $\mathfrak{g}$. But many interesting examples start with a non-linear Lie algebra. A typical example of a non-linear Lie algebra $\mathfrak{g}$ is the vector space spanned by the modes of the strong generators of a $\mathcal{W}$-algebra. Explaining what $\mathcal{W}$-algebras are is beyond the scope of these notes; it will suffice to say here that one can think of $\mathcal{W}$-algebras as extensions of the Virasoro algebra that appear as symmetry algebras for two-dimensional conformal field theories. See for instance \cite{Ar16} for an introduction to $\mathcal{W}$-algebras.

Following the strategy above, many examples of Airy ideals have been constructed via representations of the Rees universal enveloping algebras of $\mathcal{W}$-algebras -- see for instance \cite{BBCCN18, BKS20,BM20, BBCC21,  BCJ22,BCU24}. In particular, a class of Airy ideals that arise from representations of the $\mathcal{W}(\mathfrak{gl}_n)$-algebras at self-dual level reproduce the topological recursion of \cite{BE13} which we will encounter later on.
\end{example}

\section{Topological recursion}

\label{s:TR}

Let us now study a different recursive structure, which at first will appear completely unrelated. The recursive structure in question has become known as the ``Eynard-Orantin topological recursion'', or simply ``topological recursion'' (affectionately shortened as ``TR''). It was originally formulated in \cite{EO07} as a method for solving loop equations to recursively calculate the correlators of matrix models. But it was proposed to be much more general than that, and indeed it has now become a unifying theme in many different contexts, from Hurwitz theory, to Gromov-Witten theory, to knot theory, to topological and cohomological field theories, to topological string theory. 

\subsection{Spectral curves}

The starting point is the geometry of a spectral curve. Note that many of the concepts and results on the geometry of Riemann surfaces used in this section were beautifully explained in Marco Bertola's notes for his lecture series at the current school \cite{Be24}.

\begin{definition}\label{d:sc}
A \emph{spectral curve} is a quadruple $\mathcal{S} = (\Sigma,x,\omega_{0,1}, \omega_{0,2})$, where:\footnote{We note here that the definition of spectral curves can be generalized in many different ways. We take here a fairly simple approach, but general enough for our purposes.}
\begin{itemize}
\item $\Sigma$ is a Riemann surface;
\item $x: \Sigma \to \mathbb{P}^1$ is a holomorphic map between Riemann surfaces;
\item $\omega_{0,1}$ is a meromorphic one-form on $\Sigma$;
\item $\omega_{0,2}$ is a \emph{fundamental bidifferential} (sometimes also called \emph{Bergman kernel} and denoted by $B$), which is a symmetric meromorphic bidifferential on $\Sigma \times \Sigma$ whose only pole consists of a double pole on the diagonal with biresidue $1$. 
\end{itemize}
\end{definition}
We will come back to the mysterious object $\omega_{0,2}$ in a second; you can forget about it for the time being.

First, we remark that specifying a holomorphic map $x: \Sigma \to \mathbb{P}^1$ is the same as specifying a meromorphic function $x$ on $\Sigma$. Second, specifying a one-form $\omega_{0,1}$ in $\Sigma$ is equivalent to specifying a second holomorphic map $y : \Sigma \to \mathbb{P}^1$ (or meromorphic function $y$) such that $\omega_{0,1} = y \ dx$. This is why spectral curves are sometimes defined by specifying two holomorphic maps $x,y: \Sigma \to \mathbb{P}^1$ instead of a map $x$ and a one-form $\omega_{0,1}$. 

When $\Sigma$ is a compact Riemann surface, the holomorphic map $x: \Sigma \to \mathbb{P}^1$ is a finite degree branched covering. Moreover, in this case there always exists an irreducible polynomial $P(x,y)$ such that the two holomorphic maps $x$ and $y$ identically satisfy the polynomial equation $P(x,y) = 0$.
So we can construct a spectral curve by starting with an algebraic curve 
\begin{equation}\label{eq:alg}
\{ P(x,y) = 0 \} \subset \mathbb{C}^2.
\end{equation} 
We take $\Sigma$ to be the normalization of the corresponding Riemann surface, $x: \Sigma \to \mathbb{P}^1$ to be projection on the $x$-axis, and $\omega_{0,1} =  y\ dx$. Note however that our definition is more general, as it allows non-compact Riemann surfaces $\Sigma$; those spectral curves may not come from algebraic curves.

Locally, every non-constant holomorphic map between Riemann surfaces looks like a power map. More precisely, let $p \in \Sigma$, and $x(p) \in \mathbb{P}^1$. Then there exists a local coordinate $\zeta$ on $\Sigma$ centered at $p$ and a local coordinate on $\mathbb{P}^1$ centered at $x(p)$ such that the map $x: \Sigma \to \mathbb{P}^1$ takes the local normal form
\begin{equation}
x: \zeta \mapsto \zeta^r
\end{equation}
for some positive integer $r$. If we think of $x$ as a meromorphic function on $\Sigma$, then locally it can be written as $x = x(p) + \zeta^r$ if $x(p) \in \mathbb{C}$, and $x = \zeta^{-r}$ if $p$ is a pole of $x$.

The integer $r$ is uniquely defined and called the \emph{ramification order} of $x$ at $p$. We say that $p$ is a \emph{ramification point} if $r \geq 2$, and its image $x(p)$ is called a \emph{branch point}. Let us introduce the notation $\Ra \subset \Sigma$ for the set of ramification points of $x$. Those correspond to the zeros of the differential $dx$ and the poles of $x$ of order $\geq 2$.

The ramification order specifies the local behaviour of $x$ near a point $p$. But what about the one-form $\omega_{0,1}$? We can expand it near $p$ in a local coordinate $\zeta$. We get:
\begin{equation}
\omega_{0,1} = \left( \sum_{k \geq \bar s} \tau_k \zeta^k \right) \frac{d \zeta}{\zeta},
\end{equation}
where $\bar s$ is the minimal exponent appearing in the expansion. For further use, we also define $s$, which is the minimal exponent not divisible by $r$:
\begin{equation}
s = \min\{ k \in \mathbb{Z} ~|~ \tau_k \neq 0 \text{ and $r$ does not divide $k$} \}.
\end{equation}
The three parameters $(r,\bar s, s)$ characterize the local behaviour of a spectral curve near a point $p$.

To make sense of topological recursion, we will need to impose an admissibility condition on spectral curves, which constrains the local behaviour of a spectral curve.

\begin{definition}\label{d:admissible}
Let  $\mathcal{S} = (\Sigma,x,\omega_{0,1}, \omega_{0,2})$ be a spectral curve. We say that it is \emph{admissible at a point $p \in \Sigma$} if either $x$ is unramified at $p$, or $x$ is ramified at $p$ and the three following conditions on the local behaviour at $p$ are satisfied:
\begin{enumerate}
\item $r$ and $s$ are coprime;
\item either $s \leq -1$, or $s \in \{1, \ldots, r+1 \}$ and $r = \pm 1 \mod s$;
\item $\bar s = s$ or $\bar s = s -1$.
\end{enumerate}
We say that the spectral curve is \emph{admissible} if it is admissible everywhere on $\Sigma$.
\end{definition}

\begin{example}\label{e:Airy}
Consider the spectral curve $\mathcal{S}=(\Sigma, x, \omega_{0,1}, \omega_{0,2})$ with 
\begin{equation}
\Sigma = \mathbb{C}, \qquad  x = \frac{1}{2} z^2, \qquad \omega_{0,1} = z^2\ dz \ (\text{equivalently, } y=z), \qquad \omega_{0,2} = \frac{dz_1 dz_2}{(z_1-z_2)^2}. 
\end{equation}
This spectral curve is called the \emph{Airy spectral curve}. As $x$ has no pole on $\mathbb{C}$ and $dx = z dz$, there is a single ramification point at $z=0$, with ramification order $r=2$. Looking at $\omega_{0,1}$ we see that $s=\bar s = 3$. Thus it is admissible.
\end{example}

\begin{example}\label{e:Bessel}
Consider the spectral curve $\mathcal{S}=(\Sigma, x, \omega_{0,1}, \omega_{0,2})$ with 
\begin{equation}
\Sigma = \mathbb{C}, \qquad  x = \frac{1}{2} z^2, \qquad \omega_{0,1} =  dz \ (\text{equivalently, } y=\frac{1}{z}), \qquad \omega_{0,2} = \frac{dz_1 dz_2}{(z_1-z_2)^2}. 
\end{equation}
This spectral curve is called the \emph{Bessel spectral curve}. As $x$ has no pole on $\mathbb{C}$ and $dx = z dz$, there is a single ramification point at $z=0$, with ramification order $r=2$. Looking at $\omega_{0,1}$ we see that $s=\bar s = 1$, and the spectral curve is admissible.
\end{example}

Those two spectral curves are the ``building blocks'' of topological recursion, as they encapsulate the two relevant possible local behaviours near simple ramification points (that is, ramification points with $r=2$).

There is a similar class of spectral curves that encapsulate the local behaviour near ramification points of arbitrary order:

\begin{example}\label{e:rs}
Consider the spectral curve $\mathcal{S}=(\Sigma, x, \omega_{0,1}, \omega_{0,2})$ with 
\begin{equation}
\Sigma = \mathbb{C}, \qquad  x = \frac{1}{r} z^r, \qquad \omega_{0,1} = z^{s-1}\ dz \ (\text{equivalently, } y=z^{s-r}), \qquad \omega_{0,2} = \frac{dz_1 dz_2}{(z_1-z_2)^2},
\end{equation}
with $s \in \{1,\ldots,r+1\}$ and $r= \pm 1 \mod s$. This spectral curve is called the \emph{$(r,s)$ spectral curve}. As $x$ has no pole on $\mathbb{C}$ and $dx = z^{r-1} dz$, there is a single ramification point at $z=0$, with ramification order $r$. Looking at $\omega_{0,1}$ we see that what we called $s$ above is the local parameter $s$, and $s=\bar s$. Thus the spectral curve is admissible.
\end{example}

For simplicity, from now on in these notes we will generally assume that $x$ only has simple ramification points (that is, $r=2$ at all ramification points). There is no particular reason to assume this in the general framework of topological recursion,\footnote{Although one can argue that this is the generic case, since ramification points of a holomorphic maps are generically simple, but can become of higher ramification order after colliding (see for instance \cite{BBCKS23}). Nevertheless, many cases of interest correspond to holomorphic maps $x: \Sigma \to \mathbb{P}^1$ with higher ramification.} but it makes the formulae much simpler, which is useful pedagogically in the context of lecture notes. For the general definition of topological recursion for arbitrary ramification, see for instance \cite{BE13,BE17,BBCCN18,BKS20,BBCKS23,}.

When $r=2$ at a ramification point, admissibility simplifies. More precisely, these are the only possible choices of $s$:
\begin{itemize}
\item $s \leq -1$, which will make the ramification point irrelevant;
\item $s=1$, in which case we say that the ramification point is of ``Bessel-type'';
\item $s=3$, in which case we say that the ramification point is of ``Airy-type''.
\end{itemize}
For simplicity, we will also assume that ramification points of the first type do not appear; these would not contribute to topological recursion anyway, and hence we do not really lose generality by making such an assumption.

 For clarity, let us introduce a name for this class of spectral curves:
\begin{definition}
Let  $\mathcal{S} = (\Sigma,x,\omega_{0,1}, \omega_{0,2})$ be an admissible spectral curve. We say that it is \emph{simple} if all ramification points in $\Ra$ are simple and of either Bessel-type ($s=1$) or Airy-type ($s=3$).\footnote{In the standard literature on topological recursion, for spectral curves with simple ramification point it is usually required that $dx$ and $dy$ do not have common zeros \cite{EO07,EO08}. This is equivalent to requiring that $s$ is no greater than $3$ at a simple ramification point, since if $dy$  had a zero at a simple zero of $dx$, then $s > 3$ at that ramification point.}
\end{definition}

\subsection{Projection property and fundamental bidifferential}

So far we did not discuss the mysterious fundamental bidifferential $\omega_{0,2}$, which is part of the data of a spectral curve. What role does it play?

First, let us recall its definition: $\omega_{0,2}$ is a symmetric meromorphic bidifferential on $\Sigma \times \Sigma$ whose only pole consists of a double pole on the diagonal with biresidue $1$. 

When $\Sigma$ is a compact Riemann surface, it turns out that the fundamental bidifferential is uniquely fixed by a choice of Torelli marking.

\begin{lemma}\label{l:bergman}
Let $\Sigma$ be a compact Riemann surface, and fix a Torelli marking on $\Sigma$, which is a choice of symplectic basis $\{a_1, \ldots, a_{\bar g}, b_1, \ldots, b_{\bar g} \}$ for the homology group $H_1(\Sigma, \mathbb{Z})$, where $\bar{g}$ is the genus of $\Sigma$. Then there is a unique fundamental bidifferential $\omega_{0,2}(z_1, z_2)$ normalized such that
\begin{equation}
\oint_{a_i} \omega_{0,2}(\cdot, z_2) = 0, \qquad \text{for all $i \in \{1,\ldots,\bar{g} \}$.}
\end{equation}
\end{lemma}

Note however that, when $\Sigma$ is non-compact, there is no similar notion of uniqueness, so there can be more choices of fundamental bidifferential, and this is why it is part of the data of a spectral curve.

The genus zero case is particularly important.

\begin{example}
Let $\Sigma = \mathbb{P}^1$, and let $z$ be a uniformizing coordinate on $\mathbb{P}^1$. Since for genus zero compact Riemann surfaces $H_1(\Sigma,\mathbb{Z}) = \{0\}$ (i.e. the first homology is trivial), there is no choice of  Torelli marking. The unique fundamental bidifferential simply reads
\begin{equation}\label{eq:w02}
\omega_{0,2}(z_1,z_2) = \frac{d z_1 dz_2}{(z_1-z_2)^2}.
\end{equation}
\end{example}
This is the fundamental bidifferential that was used in the previous examples, and is the only one that we will use in these lecture notes. 

We now understand better what the object is. But why do we need to specify a fundamental bidifferential? The main reason is that it defines a projection operator as follows.

\begin{definition}\label{d:proj}
Let $\omega_{0,2}$ be a fundamental bidifferential on $\Sigma$.  Let $\alpha$ be a meromorphic one-form on $\Sigma$. Pick a finite set of points $P \subset \Sigma$. We define the projection $\hat{B}_P[\alpha](z')$ of $\alpha(z)$ at $P$ to be the one-form on $\Sigma$ given by:
\begin{equation}
\hat{B}_P[\alpha](z') = \sum_{p \in P} \Res_{z=p} \left( \int_p^z \omega_{0,2}(\cdot, z') \right) \alpha(z).
\end{equation}
\end{definition}

The cool thing about this operation is the following lemma:

\begin{lemma}
The projection $\hat{B}_P[\alpha]$ is a one-form on $\Sigma$ such that:
\begin{itemize}
\item $\hat{B}_P[\alpha]$ is residueless on $P$;
\item $\hat{B}_P[\alpha] - \alpha = \omega$ where $\omega$ is a one-form that has at most simple poles on $P$.
\end{itemize}
 In other words, $\hat{B}_P[\alpha]$ and $\alpha$ have the same principal parts on $P$ except for the $-1$-terms in their Laurent expansions, as $\hat{B}_P[\alpha]$ is residueless on $P$ while $\alpha$ may not be.

It then follows that $\hat{B}_P$ is a projection, as $\hat{B}_P \circ \hat{B}_P = \hat{B}_P$.
\end{lemma}

\begin{proof}
This is easiest to see in local coordinates. Pick a point $p \in P$ and a local coordinate $\zeta$ centered at $p$. We can expand the fundamental bidifferential and $\alpha$ in this local coordinate:
\begin{align}
\omega_{0,2}\simeq& \frac{d \zeta_1 d \zeta_2}{(\zeta_1-\zeta_2)^2} + \sum_{l,m> 0} \phi_{lm}  \zeta_1^{l-1} \zeta_2^{m-1} d \zeta_1 d \zeta_2,\\
\alpha \simeq& \sum_{k=-M}^\infty \alpha_k \zeta^{k} d \zeta.
\end{align}
As $\zeta \to 0$, we get
\begin{equation}
\int^\zeta_0 \omega_{0,2}(\cdot,\zeta') \simeq  d \zeta' \sum_{n=1}^\infty \frac{\zeta^n}{(\zeta')^{n+1}} + \sum_{l,m >0} \phi_{l,m} \frac{\zeta^l(\zeta')^{m-1}}{l}  d \zeta' .
\end{equation}
Therefore, evaluating the projection at  $p$ in the local coordinate, we get:
\begin{align}
\hat{B}_p[\alpha](\zeta') \simeq& \Res_{\zeta=0} \left( \int^\zeta_0 \omega_{0,2}(\cdot,\zeta') \right) \alpha(\zeta) \\
=& \left( \sum_{k=-M}^{-2} \alpha_k (\zeta')^k  + O(1) \right) d \zeta'.
\end{align}
As the sum stops at $k=-2$, we see that $\hat{B}_p[\alpha]$ is residueless at $p$. Moreover, by direct comparison we get that
\begin{equation}
\hat{B}_p[\alpha](\zeta') - \alpha(\zeta') = \left( \frac{\alpha_{-1}}{\zeta'} + O(1) \right) d \zeta'.
\end{equation}
In other words, the one-forms $\hat{B}_p[\alpha]$ and $\alpha$ have the same principal part at $p$ except for the $-1$-term in their Laurent expansion as $\hat{B}_p[\alpha]$ is residueless while $\alpha$ may not be. A similar statement holds true for the one-form $\hat{B}_P[\alpha]$, where we now sum over all $q \in P$, since the projection at the other points $q \neq p$ in $P$ are holomorphic at $p \in P$.

As for showing that it is a projection, for any one-form $\alpha$ on $\Sigma$, we can write
\begin{equation}
\hat{B}_P[\alpha]=  \alpha + \omega
\end{equation}
where $\omega$ has at most simple poles on $P$. But then, it is easy to see that $\hat{B}_P[\omega] = 0$, and hence $\hat{B}_P \circ \hat{B}_P = \hat{B}_P$.
\end{proof}

In fact, one can think of this projection operation as a ``local-to-global'' operation. Indeed, in Definition \ref{d:proj}, on the right-hand-side we don't really need $\alpha$ to be a well-defined one-form on $\Sigma$; all we need is its local data near the points $p \in P$. In other words, we could take as input only the local data of differentials on open sets around the points $p \in P$ (what is called the ``germ'' of a one-form in fancy language), and the output of the operation would be a globally defined one-form on $\Sigma$ that has the same principal part (excluding the $-1$-terms) at the points $p \in P$ as the given local data. Neat!

Along these lines, this projection allows us to define a natural basis of one-forms on $\Sigma$ with prescribed poles (with no residue) at a point $p$:

\begin{definition}\label{d:xi}
Let $\omega_{0,2}$ be a fundamental bidifferential on $\Sigma$. Pick a point $p \in \Sigma$, with $\zeta$ a local coordinate centered at $p$. For $k \in \mathbb{N}^*$, we define the one-forms
\begin{align}
 \xi^{(p)}_{-k}(z') =& \hat{B}_p\left[ \frac{d \zeta(z)}{\zeta^{k+1}(z)} \right] (z') \\
 =&
  \Res_{z=p} \left( \int^{z}_{p} \omega_{0,2}(\cdot, z') \frac{d \zeta(z)}{\zeta^{k+1}(z)} \right).
\end{align}
\end{definition}

In other words, we start with the locally defined one-form $\frac{d \zeta}{\zeta^{k+1}}$ at $p$, and construct a globally defined one-form $\xi_{-k}(z)$ with the same principal part at $p$. That is, in local coordinate $\zeta$ centered at $p$, 
\begin{equation}
\xi^{(p)}_{-k}(z) \simeq \left( \frac{1}{\zeta(z)^{k+1}} + \text{holomorphic} \right) d \zeta(z).
\end{equation}

This is a particularly nice basis of one-forms because one-forms obtained by projection have a finite expansion in this basis:
\begin{lemma}\label{l:finite}
Let $\omega_{0,2}$ be a fundamental bidifferential on $\Sigma$, and $P \subset \Sigma$ a finite set of points. Let $\alpha$ be a meromorphic one-form on $\Sigma$. Then
\begin{equation}
\hat{B}_P[\alpha](z) = \sum_{p \in P} \sum_{k\in \mathbb{N}^*} \alpha_{p,k} \xi_{-k}(z),
\end{equation}
where only finitely many coefficients $\alpha_{p,k}$ are non-vanishing.
\end{lemma}

\begin{proof}
Expand $\alpha$ in a local coordinate $\zeta$ at $p \in P$:
\begin{equation}
\alpha \simeq \sum_{k=-M}^\infty \alpha_k \zeta^{k} d \zeta.
\end{equation}
The projection operator turns the $\zeta^{-k-1}$ into $\xi_{-k}$ and kills the holomorphic and $-1$ part.
\end{proof}
This finite decomposition will prove to be very useful in the following. To end this section, we define the ``projection property''.

\begin{definition}\label{d:projprop1}
Let $\omega_{0,2}$ be a fundamental bidifferential on $\Sigma$, and $P \subset \Sigma$ a finite set of points. Let $\alpha$ be a meromorphic one-form on $\Sigma$. We say that $\alpha$ satisfies the \emph{projection property on $P$} if
\begin{equation}
\hat{B}_P[\alpha] = \alpha.
\end{equation}
By Lemma \ref{l:finite} this implies that $\alpha$ has a finite decomposition in the basis of one-forms $\xi_{-k}^{(p)}(z)$ for $p \in P$.
\end{definition}

\subsection{Topological recursion}

With this under our belt, we are now ready to define topological recursion. As mentioned previously, we will focus on simple spectral curves, for clarity. However, everything can be naturally generalized to spectral curves with arbitrary ramification, with the only price to pay being complicated combinatorics and ugly looking formulae.

Let $\mathcal{S} = (\Sigma,x,\omega_{0,1}, \omega_{0,2})$ be a simple spectral curve. Let $\zeta$ be a local coordinate centered at a ramification point $a \in \Ra$. Since $a$ is simple, locally near $a$ we have $x(\zeta) = x(a) + \zeta^2$ if $x(a) \in \mathbb{C}$, or $x(\zeta) = \frac{1}{\zeta^2}$ if $a$ is a pole of $x$. In both cases, there is a natural involution $\sigma_a: \zeta \mapsto - \zeta$ near $a$ that fixes $a$ and such that $x(\sigma_a(\zeta)) = x(\zeta)$. This is the involution that exchanges the two sheets of the branched covering $x: \Sigma \to \mathbb{P}^1$ that meet at the ramification point $a \in \Sigma$.

The goal of topological recursion is to construct a sequence of meromorphic differentials $\{\omega_{g,n} \}_{g \in \mathbb{N}, n \in \mathbb{N}^*}$, where $\omega_{g,n}$ is a symmetric meromorphic differential on $\Sigma^n$. We also impose that those differentials only have poles at the ramification points in $\Ra$ and that they are residueless. We will call those ``correlators''.

\begin{definition}\label{d:system}
Let $\mathcal{S} = (\Sigma,x,\omega_{0,1}, \omega_{0,2})$ be a simple spectral curve. A \emph{system of correlators on $\mathcal{S}$} is a collection $\{\omega_{g,n} \}_{g \in \mathbb{N}, n \in \mathbb{N}^* }$ of symmetric $n$-differentials on $\Sigma^n$, where $\omega_{0,1}$ and $\omega_{0,2}$ are already specified by $\mathcal{S}$, and for $2g-2+n>0$, $\omega_{g,n}$ only has poles on $\Ra$ with vanishing residues.
\end{definition}

To define topological recursion, we will need the following particular combination of correlators:
\begin{definition}\label{d:comb}
Let $\mathcal{S} = (\Sigma,x,\omega_{0,1}, \omega_{0,2})$ be a simple spectral curve, and $\{\omega_{g,n} \}_{g \in \mathbb{N}, n \in \mathbb{N}^* }$ a system of correlators on $\mathcal{S}$. Let $a \in \Ra$. We define
\begin{multline}\label{eq:comb}
\tilde{\omega}^{(a)}_{g,n}(z, z_2, \ldots, z_n) = \frac{1}{\omega_{0,1}(z) - \omega_{0,1}(\sigma_a(z))} \Big( \omega_{g-1,n+1}(z, \sigma_a(z), z_2, \ldots, z_n) \\+ \sum'_{\substack{g_1+g_2=g \\ I \cup J = \{z_2, \ldots, z_n \}}} \omega_{g_1, |I|+1}(z, I) \omega_{g_2, |J|+1} (\sigma_a(z), J) \Big),
\end{multline}
which is a locally defined one-form (in the variable $z$) in the neighborhood of $z=a$. Here, the second sum is over disjoint, possibly empty, subsets $I, J \subset \{ z_2, \ldots, z_n\}$ such that $I \cup J = \{z_2, \ldots, z_n\}$, and the prime means that we omit all terms involving $\omega_{0,1}$.  We note that $\tilde{\omega}_{g,n}^{(a)}$ is constructed out of $\omega_{g',n'}$ with $2g'-2+n' < 2g-2+n$.
\end{definition}

The definition may seem strange at first, as we are ``dividing by a one-form''. But this is just a formal manipulation; the term in brackets is a quadratic differential in the variable $z$, so by  ``dividing by a one-form'' we just mean cancelling a differential so that the result is a one-form in $z$.

We finally define topological recursion, which amounts to constructing globally defined differentials $\omega_{g,n}$ from the locally defined one-forms $\tilde{\omega}_{g,n}^{(a)}$ from Definition \ref{d:comb} using the projection operator from Definition \ref{d:proj}. In other words, topological recursion constructs globally defined differentials $\omega_{g,n}(z_1,\ldots,z_n)$ on $\Sigma^n$ that have the same principal parts (in $z_1$) as the local $\tilde{\omega}^{(a)}_{g,n}(z_1,\ldots,z_n)$, but with no residues.

\begin{definition}\label{d:TR}
We say that a system of correlators $\{\omega_{g,n} \}_{g \in \mathbb{N}, n \in \mathbb{N}^* }$ on a simple spectral curve $\mathcal{S} = (\Sigma,x,\omega_{0,1}, \omega_{0,2})$ satisfies \emph{topological recursion} if, for all $g,n$ such that $2g-2+n > 0$, 
\begin{align}\label{eq:TR}
\omega_{g,n}(z_1,\ldots,z_n) =&\sum_{a \in \Ra} \hat{B}_{a}[\tilde{\omega}^{(a)}_{g,n}(\cdot, z_2, \ldots, z_n)](z_1)\\
=& \sum_{a \in \Ra} \Res_{z=a} \left(\int_a^z  \omega_{0,2}(\cdot, z_1) \right) \tilde{\omega}^{(a)}_{g,n}(z, z_2, \ldots, z_n).
\end{align}
\end{definition}

We note that the correlators are uniquely reconstructed by topological recursion since, as noted in Definition \ref{d:comb}, $\tilde{\omega}^{(a)}_{g,n}$ only involves contributions from $\omega_{g',n'}$ with $2g'-2+n' < 2g-2+n$. Thus the formula is recursive on $2g-2+n$, with initial data given by $\omega_{0,1}$ and $\omega_{0,2}$. For $2g-2+n>0$, the correlators $\omega_{g,n}$ only have poles on $\Ra$, with order at most $6g-4+2n$.

\begin{remark}
A common way of writing topological recursion in the literature is to introduce the \emph{recursion kernel}
\begin{equation}
K^{(a)}(z_1,z) =  \frac{\int^z_a \omega_{0,2}(\cdot, z_1) }{\omega_{0,1}(z) - \omega_{0,1}(\sigma_a(z))},
\end{equation}
which is locally defined (in $z$) near $z=a$. Then topological recursion can be rewritten as
\begin{multline}\label{eq:TR2}
\omega_{g,n}(z_1, \ldots, z_n) = \sum_{a \in \Ra} \Res_{z=a} K^{(a)}(z_1,z)  \Big( \omega_{g-1,n+1}(z, \sigma_a(z), z_2, \ldots, z_n) \\+ \sum'_{\substack{g_1+g_2=g \\ I \cup J = \{z_2, \ldots, z_n \}}} \omega_{g_1, |I|+1}(z, I) \omega_{g_2, |J|+1} (\sigma_a(z), J) \Big).
\end{multline}
\end{remark}

\subsubsection{Symmetry}

You may have noticed a strong similarity between \eqref{eq:comb}-\eqref{eq:TR} (or \eqref{eq:TR2}) and the recursive formula for the coefficients $F_{g,n}$ of a partition function $Z$ that we obtained from the differential constraints $\I Z = 0$ when the Airy ideal $\I$ is generated by differential operators $H_a$ that are quadratic in $\hslash$, see \eqref{eq:trA}.  Indeed, the combinatorics of topological recursion are roughly the same as the combinatorics of the action of differential operators $H_a$ that are degree two in $\hslash$  on a partition function $Z$. This is of course not a coincidence; we will see how the two formalisms are related in the next section.

As was the case with the recursion \eqref{eq:trA}, symmetry of the correlators produced by \eqref{eq:TR} is far from obvious, since $z_1$ plays a very different role from the other variables $z_2, \ldots, z_n$ in the formula. As symmetry is a defining property of a system of correlators in Definition \ref{d:system}, saying that a system of correlators satisfies topological recursion on a given spectral curve is equivalent to saying that \eqref{eq:TR} produces symmetric correlators on that spectral curve. Note that symmetry of the correlators produced by \eqref{eq:TR} on any admissible simple spectral curve has been proved by direct calculation in \cite{EO07}. It will also follow as a consequence of the relation with Airy ideals, as we will see.

In fact, as mentioned previously the topological recursion formula can be generalized to spectral curves with arbitrary ramification, see \cite{BE13}. As in the particular case above, symmetry of the correlators is then far from obvious. It turns out that the correlators produced by the general form of topological recursion are symmetric whenever the spectral curve is admissible according to Definition \ref{d:admissible} --- this is in fact the reason for introducing this particular admissibility condition. However, in contrast to the case of simple ramification above, proof of symmetry in the general context has not been done by direct calculation, as it gets rather messy. It is rather obtained  in a much nicer way as a consequence of the relation with Airy ideals (see \cite{BBCCN18}).

We can summarize these statements in the following theorem:

\begin{theorem}
Let $\mathcal{S} = (\Sigma,x,\omega_{0,1}, \omega_{0,2})$ be an admissible spectral curve. Construct a collection of differentials $\{\omega_{g,n} \}_{g \in \mathbb{N}, n \in \mathbb{N}^* }$ via \eqref{eq:TR} if the curve is simple, and via the natural generalization of topological recursion in \cite{BE13} for spectral curves with higher ramification points. Then the differentials $\omega_{g,n}$ are symmetric.
\end{theorem}

\subsubsection{Projection property}

Another key property of the correlators produced by topological recursion is that they satisfy the projection property. Recall the definition of the projection property in Definition \ref{d:projprop1}.

\begin{lemma}
Let $\{\omega_{g,n} \}_{g \in \mathbb{N}, n \in \mathbb{N}^* }$ be a system of correlators on a simple spectral curve $\mathcal{S} = (\Sigma,x,\omega_{0,1}, \omega_{0,2})$ that satisfies topological recursion. Then the correlators satisfy the projection property on $\Ra$ (in the variable $z_1$). That is,
\begin{align}
\omega_{g,n}(z_1, \ldots, z_n) =&  \hat{B}_{\Ra}[ \omega_{g,n}(\cdot, z_2, \ldots, z_n)](z_1)\\
=& \sum_{a \in \Ra} \Res_{z=a} \left( \int^z_a \omega_{0,2}(\cdot, z_1) \right) \omega_{g,n}(z, z_2, \ldots, z_n).
\end{align}
We note that the statement remains true on any admissible spectral curve, with the correlators computed via the generalization of topological recursion from \cite{BE13}.
\end{lemma}

\begin{proof}
This is clear by definition of topological recursion (see Definition \ref{d:TR}) and the fact that $\hat{B}_{\Ra}$ is a projection. For general admissible spectral curves, this is explained in \cite{BBCKS23}.
\end{proof}

As explained in Lemma \ref{l:finite}, a direct consequence of the projection property is the existence of a finite expansion for the correlators in terms of the basis of one-forms $ \xi_{-k}^{(a)}$ introduced in Definition \ref{d:xi}.

\begin{lemma}
Let $\{\omega_{g,n} \}_{g \in \mathbb{N}, n \in \mathbb{N}^* }$ be a system of correlators on a spectral curve $\mathcal{S} = (\Sigma,x,\omega_{0,1}, \omega_{0,2})$ that satisfy the projection property on $\Ra$. Define the basis of one-forms $ \xi_{-k}^{(a)}$ for each ramification point $a \in \Ra$ as in Definition \ref{d:xi}. Then, for $2g-2+n>0$, the correlators take the form
\begin{equation}\label{eq:expansion}
\omega_{g,n}(z_1,\ldots,z_n) = \sum_{a_1, \ldots, a_n \in \Ra} \sum_{k_1, \ldots, k_n \in \mathbb{N}^*} F_{g,n} \begin{bmatrix} a_1 & \ldots & a_n \\ k_1 & \ldots & k_n \end{bmatrix} \xi_{-k_1}^{(a_1)}(z_1) \cdots  \xi_{-k_n}^{(a_n)}(z_n),
\end{equation}
where only a finite number of coefficients are non-zero.
\end{lemma}

This property is crucial. As we will see, generally speaking the coefficients $F_{g,n} $ of this expansion will have an interesting interpretation in terms of the moduli space of curves.

\begin{remark}\label{r:acycles}
When $\Sigma$ is a compact Riemann surface with a Torelli marking, we can think of the projection property in a different way. In this case, recall from Lemma \ref{l:bergman} that $\omega_{0,2}$ is uniquely fixed by imposing the normalization condition
\begin{equation}
\oint_{a_i} \omega_{0,2}(\cdot, z_2) = 0, \qquad \text{for all $i \in \{1,\ldots,\bar{g} \}$.}
\end{equation}
With this choice of fundamental bidifferential, it follows from the Riemannn bilinear identity that the statement that the correlators $\omega_{g,n}$ satisfy the projection property on $\Ra$ is equivalent to saying that they are normalized on $a$-cycles. That is, for all $2g-2+n>0$,
\begin{equation}
\oint_{a_i} \omega_{g,n}(\cdot, z_2, \ldots, z_n) = 0, \qquad \text{for all $i \in \{1,\ldots,\bar{g} \}$.}
\end{equation}
This is how the projection property was usually stated in the older literature on topological recursion.
\end{remark}

\subsubsection{Graphical interpretation}

The combinatorics involved in the topological recursion formula \eqref{eq:TR2} have a nice graphical interpretation. For now, this is simply a mnemonic trick to remember the terms on the right-hand-side of \eqref{eq:TR2}. However, the graphical interpretation takes roots in the geometry of the moduli space of curves, which is the natural setup to interpret the correlators produced by topological recursion. 

The idea of the graphical interpretation goes as follows. For each ramification point $a$, we do the following.
\begin{itemize}
\item To each $\omega_{g,n}(z_1,\ldots,z_n)$, we attach a genus $g$ Riemann surface with $n$ boundary, and label the boundaries with the variables $z_1,\ldots,z_n$;
\item To the recursion kernel $K^{(a)}(z_1,z)$, we attach a ``pair of pants'' (that is, a genus $0$ Riemann surface with three boundaries), and label the boundaries with $z_1$, $z$ and $\sigma_a(z)$. 
\end{itemize}

To get the terms on the right-hand-side of the recursion for $\omega_{g,n}(z_1,\ldots,z_n)$, we start with a genus $g$ Riemann surface with $n$ boundaries labeled by the the variables $z_1,\ldots,z_n$, and we cut off a pair of pants such that it includes the boundary labeled by $z_1$. This creates two new boundaries on the resulting Riemann surface (which may or may not be connected), which we label by $z$ and $\sigma_a(z)$ to match the two other boundaries of the pair of pants.

We do this in all possible ways, with the only constraint that we never allow discs (i.e. genus zero surfaces with only one boundary). This creates all terms on the right-hand-side of topological recursion \eqref{eq:TR}. This is exemplified for $\omega_{1,2}(z_1,z_2)$ in figure \ref{f:graphical}.

\begin{center}
\begin{figure}[H]
\includegraphics[width=\textwidth]{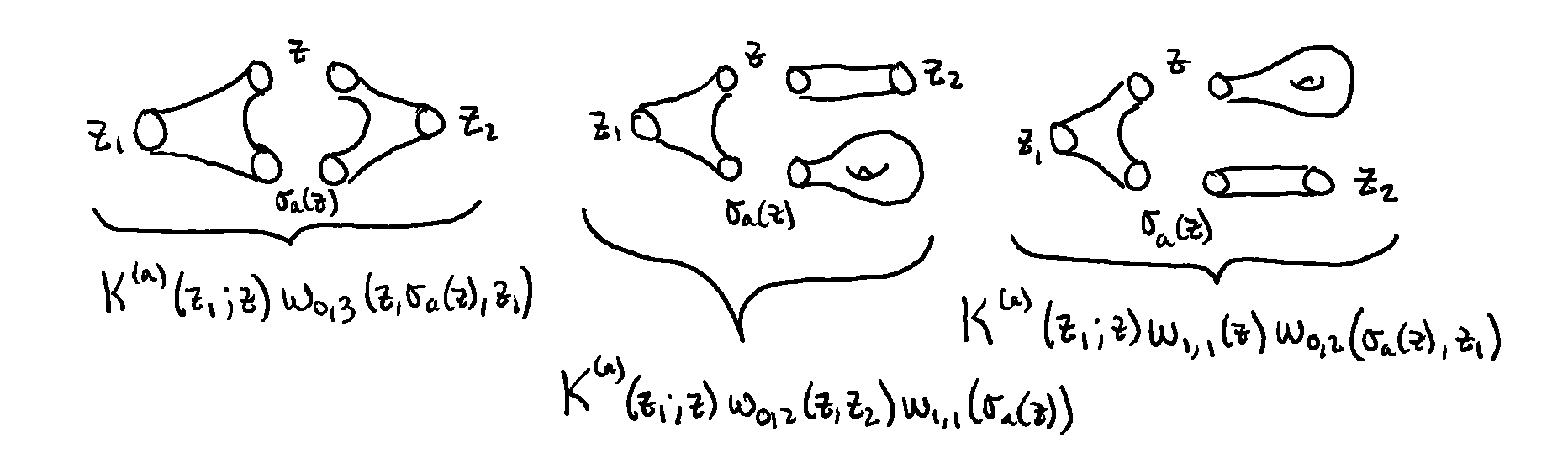}
\caption{A graphical interpretation in terms of pairs of pants of the right-hand-side of the topological recursion formula \eqref{eq:TR} for $\omega_{1,2}(z_1,z_2)$.}
\label{f:graphical}
\end{figure}
\end{center}

Another way to think of the recursive structure of \eqref{eq:TR} is to represent the terms on the right-hand-side of \eqref{eq:TR} as nodal Riemann surfaces (see the lecture notes on moduli spaces of Riemann surfaces from this school for this kind of drawing \cite{GL24}). To get the terms on the right-hand-side of the recursion for $\omega_{g,n}(z_1,\ldots,z_n)$, we draw all possible topologically inequivalent connected nodal Riemann surfaces of genus $g$ with $n$ marked points and such that:
\begin{itemize}
\item There is one component that is a sphere with three marked points, one of which is the point $z_1$ and the other two are nodes;
\item There are no other nodes;
\item There is no component that is a sphere with only one marked point.
\end{itemize}
We then take the normalization of the nodal surfaces, attach the kernel $K^{(a)}(z_1;z)$ to the singled out sphere with three marked points, and attach appropriate correlators $\omega_{g,n}$ to the other components as above.

Note that those are not necessarily stable nodal Riemann surfaces, as we may have components that are spheres with two marked points. What this pictorial representation shows is that the recursive structure seems to be somewhat related to going to the boundary of the compactified moduli space of curves $\overline{\mathcal{M}}_{g,n}$, but in a not so straightforward way. This is exemplified for $\omega_{1,2}(z_1,z_2)$ in figure \ref{f:graphical2}.

\begin{center}
\begin{figure}[H]
\includegraphics[width=\textwidth]{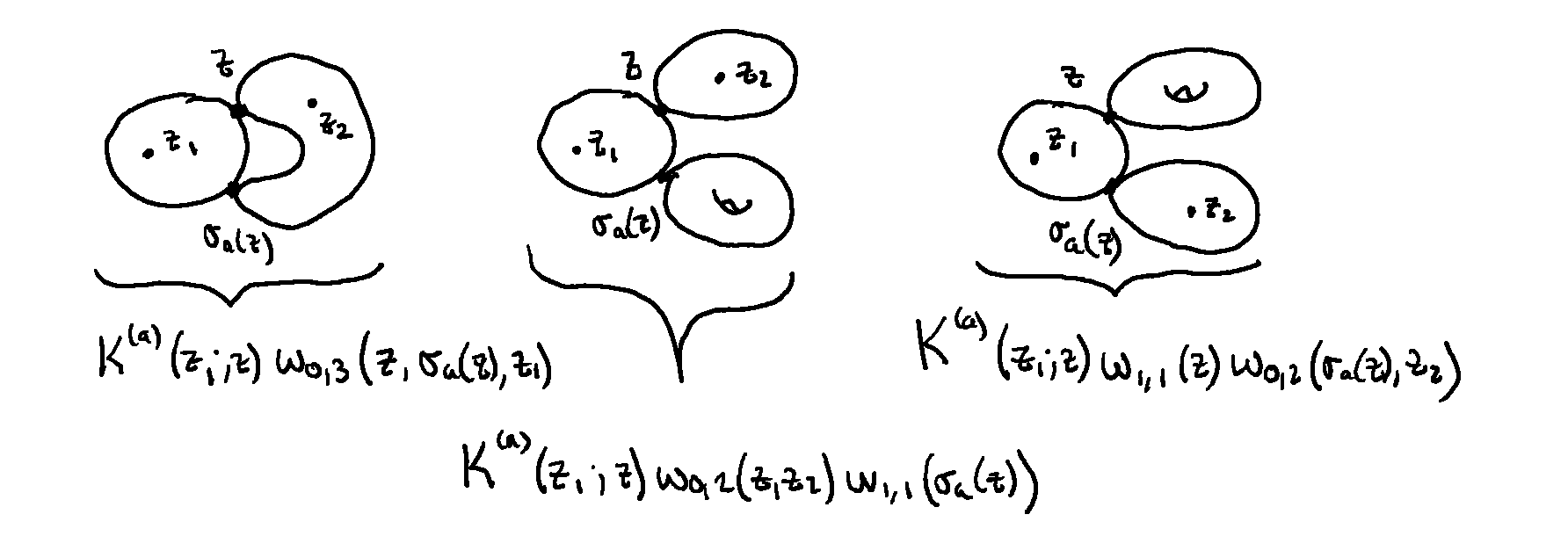}
\caption{A graphical interpretation in terms of nodal surfaces of the right-hand-side of the topological recursion formula \eqref{eq:TR} for $\omega_{1,2}(z_1,z_2)$.}
\label{f:graphical2}
\end{figure}
\end{center}

\begin{remark}
Using this graphical interpretation it is easy to explain the combinatorics that arise in topological recursion for ramification points with ramification order $r > 2$. The idea is that, instead of extracting only a pair of pants, if a ramification point has order $r$, we extract spheres with $m+1$ boundaries, for all $m \in \{2,\ldots,r\}$. We do this in all possible ways, following the rules above, and this gives rise to all the terms on the right-hand-side of the generalized topological recursion from \cite{BE13}. You can find pretty pictures of this process in Section 2.2.3 of \cite{BBCKS23}. 

Equivalently, we can also thinks of the generalized case from the point of view of nodal surfaces. The rules are as above, but we now draw all connected nodal surfaces with one component that is a sphere with $m+1$ marked points, for all $m \in \{2,\ldots,r\}$, one of which is the point $z_1$ and the other $m$ are nodes. Pretty pictures of this type can be found in \cite{BHLMR13}.

Note that this is also roughly the type of combinatorics that arise from Airy ideals, when the differential operators $H_a$ are of degree $r$ in $\hslash$, which is again not a coincidence.
\end{remark}

\subsubsection{Dilaton equation}

Correlators that satisfy topological recursion share many more interesting properties; we refer the reader to \cite{EO07,EO08} for more on this. In this section we simply state one more important property satisfied by the correlators, which is known as the dilaton equation.

\begin{theorem}\label{t:dilaton}

Let $\mathcal{S} = (\Sigma,x,\omega_{0,1}, \omega_{0,2})$ be an admissible spectral curve. Construct a collection of differentials $\{\omega_{g,n} \}_{g \in \mathbb{N}, n \in \mathbb{N}^* }$ via \eqref{eq:TR} if the curve is simple, and via the natural generalization of topological recursion in \cite{BE13} for spectral curves with higher ramification points. Then the differentials $\omega_{g,n}$, for $2g-2+n>0$, satisfy the dilaton equation:
\begin{equation}\label{eq:dilaton}
\omega_{g,n}(z_1,\ldots,z_n) = \frac{1}{2g-2+n} \sum_{a \in \Ra} \Res_{z=a} \Phi(z) \omega_{g,n+1}(z,z_1,\ldots,z_n),
\end{equation}
where $\Phi(z)$ is an arbitrary antiderivative of $\omega_{0,1}(z)$; that is, $\omega_{0,1}(z) = d \Phi(z)$.
\end{theorem}

For simple spectral curves, this was proved in \cite{EO07}. For the general case, one can get the dilaton equation from limits as in \cite{BBCKS23} or from the equivalence with Airy ideals.

The dilaton equation allows us to extend the definition of correlators to $\omega_{g,0}$, which are simply numbers; those are also called ``free energies'' and denoted by $F_g := \omega_{g,0}$. The definition is:

\begin{definition}\label{d:fe}
Let $\mathcal{S} = (\Sigma,x,\omega_{0,1}, \omega_{0,2})$ be an admissible spectral curve. For $g \geq 2$, the \emph{free energies} $F_g := \omega_{g,0}$ are defined by
\begin{equation}
F_g = \frac{1}{2g-2} \sum_{a \in \Ra} \Res_{z=a} \Phi(z) \omega_{g,1}(z).
\end{equation}
There is also a separate definition for $F_0$ and $F_1$, see \cite{EO07}.
\end{definition}

We will not use the free energies further in these lecture notes, but they play an important role in relations to integrability, and will also appear in other lecture series in this school \cite{Li24}.

\subsection{Examples and enumerative geometry}

In this section we study a few examples of spectral curves, and highlight the enumerative interpretation of the coefficients $F_{g,n}$ of the expansion \eqref{eq:expansion}.

\subsubsection{Airy spectral curve}

Consider the Airy spectral curve introduced in Example \ref{e:Airy}. We have 
\begin{equation}
\Sigma = \mathbb{C}, \qquad x = \frac{1}{2}z^2, \qquad \omega_{0,1} = z^2\ dz, \qquad \omega_{0,2} = \frac{dz_1 dz_2}{(z_1-z_2)^2}.
\end{equation}
There is only one ramification point at $z=0$, which is simple.

\eqref{eq:TR} can be used to construct a system of correlators $\omega_{g,n}$. In fact, for this spectral curve we can calculate the first few correlators by hand. We obtain:
\begin{align}
\omega_{0,3}(z_1,z_2,z_3) =& \frac{dz_1 dz_2 dz_3}{z_1^2 z_2^2 z_3^2},\\
\omega_{1,1}(z_1) =& \frac{dz_1}{8 z_1^4},\\
\omega_{0,4}(z_1,z_2,z_3,z_4) =& 3 \frac{dz_1 dz_2 dz_3 dz_4}{z_1^2 z_2^2 z_3^2 z_4^2} \sum_{i=1}^4 \frac{1}{z_i^2},\\
\omega_{1,2}(z_1, z_2) =& \frac{dz_1 dz_2}{8} \left(\frac{5}{z_1^2 z_2^6} + \frac{5}{z_1^6 z_2^2} + \frac{3}{z_1^4 z_2^4} \right).
\end{align}

\begin{exercise}
It is a good exercise to use the topological recursion formula \eqref{eq:TR} to calculate by hand the first few correlators ($\omega_{0,3}, \omega_{1,1}, \omega_{0,4}, \omega_{1,2}$, \ldots). Do it!
\end{exercise}

For this spectral curve, since $\omega_{0,2}= \frac{dz_1 dz_2}{(z_1-z_2)^2}$, the basis of one-forms from Definition \ref{d:xi} at the ramification point $z=0$ is very simple:
\begin{equation}
\xi_{-k}(z_1) = \Res_{z=0} \left( \int^z_0 \omega_{0,2}(\cdot, z_1) \frac{d z}{z^{k+1}} \right) = \frac{d z_1}{z_1^{k+1}}.
\end{equation} 
We see that the first  few correlators indeed have a finite expansion of the form
\begin{equation}\label{eq:exairy}
\omega_{g,n}(z_1,\ldots,z_n) = \sum_{k_1, \ldots,k_n \in \mathbb{N}^*} F_{g,n}[k_1,\ldots,k_n] \xi_{-k_1}(z_1) \cdots \xi_{-k_n}(z_n),
\end{equation}
as expected. The resulting non-zero coefficients (up to permutations of the entries) for the first few correlators are
\begin{equation}
F_{0,3}[1,1,1] = 1, \qquad F_{1,1}[3] = \frac{1}{8}, \qquad F_{0,4}[1,1,1,3] = 3, \qquad F_{1,2}[1,5] = \frac{5}{8}, \qquad F_{1,2}[3,3] = \frac{3}{8}.
\end{equation}

In general, for the Airy spectral curve one can prove an explicit formula for the correlators which relates them to intersection numbers over the moduli space of curves. Using the expansion \eqref{eq:exairy}, one can prove that the only non-zero coefficients are when all $k_i$'s are odd (in other words, the correlators $\omega_{g,n}$ only have even powers of the variables, and thus are odd under the involution $z_i \mapsto - z_i$). Those non-zero coefficients are given by
\begin{equation}
F_{g,n}[2m_1+1, \ldots, 2m_n+1] = \prod_{i=1}^n (2m_i+1)!! \int_{\overline{\mathcal{M}}_{g,n}} \psi_1^{m_1} \cdots \psi_n^{m_n}.
\end{equation}
Notice a similarity with \eqref{eq:fgnairy}? This is of course not a coincidence! As we will see, topological recursion on the Airy spectral curve is equivalent to the Kontsevich-Witten Virasoro constraints.

Using the explicit calculation of the coefficients $F_{g,n}$  above, we get the first few non-zero intersection numbers (up to permutations of the psi-classes):
\begin{gather}
\int_{\overline{\mathcal{M}}_{0,3}}1 = 1, \qquad \int_{\overline{\mathcal{M}}_{1,1}} \psi_1 = \frac{1}{24}, \qquad \int_{\overline{\mathcal{M}}_{0,4}} \psi_1 = 1, \\
\int_{\overline{\mathcal{M}}_{1,2}} \psi_1^{2} = \frac{1}{24}, \qquad \int_{\overline{\mathcal{M}}_{1,2}} \psi_1 \psi_2 = \frac{1}{24},
\end{gather}
which are well-known intersection numbers. For your interest, Bertrand Eynard has a nice online program that computes intersection numbers for you \cite{EOjava}.

\subsubsection{Bessel spectral curve}

Consider next the Bessel spectral curve introduced in Example \ref{e:Airy}. We have
\begin{equation}
\Sigma = \mathbb{C}, \qquad x = \frac{1}{2}z^2, \qquad \omega_{0,1} = dz, \qquad \omega_{0,2} = \frac{dz_1 dz_2}{(z_1-z_2)^2}.
\end{equation}
There is only one ramification point at $z=0$, which is simple.

For this spectral curve, using \eqref{eq:TR} one can show that
\begin{equation}
\omega_{0,n}(z_1, \ldots, z_n) = 0
\end{equation}
for all $n \in \mathbb{N}^*$. So the only non-zero correlators have $g \geq 1$. This can be proved with a simple pole analysis.

\begin{exercise}
Check this!
\end{exercise}

The first few non-zero correlators can also be calculated by hand. We get:
\begin{align}
\omega_{1,1}(z_1) =& \frac{dz_1}{8 z_1^2},\\
\omega_{1,2}(z_1, z_2) =& \frac{dz_1 dz_2}{8 z_1^2 z_2^2}.
\end{align}

The basis of one-forms at the ramification point at $z=0$ is the same as for the Airy spectral curve, and we observe that we get a finite expansion in this basis as expected. Just as for the Airy spectral curve, it is straightforward to show that the only non-zero coefficients $F_{g,n}[k_1,\ldots,k_n]$ are when all $k_i$'s are odd. 

The coefficients also have an interpretation in terms of intersection numbers over the moduli space of curves, which was only recently found by Norbury. The statement is that
\begin{equation}
F_{g,n}[2m_1+1, \ldots, 2m_n+1] = \prod_{i=1}^n (2m_i+1)!! \int_{\overline{\mathcal{M}}_{g,n}} \Theta_{g,n} \psi_1^{m_1} \cdots \psi_n^{m_n},
\end{equation}
where $\Theta_{g,n}$ is the Norbury cohomology class on $\overline{\mathcal{M}}_{g,n}$, which already appeared in \eqref{eq:norbury} \cite{No17}. Notice the similarity again!

For instance, using the first few correlators above, we get:
\begin{equation}
\int_{\overline{\mathcal{M}}_{1,1}} \Theta_{1,1} = \frac{1}{8}, \qquad \int_{\overline{\mathcal{M}}_{1,2}} \Theta_{1,2} = \frac{1}{8},
\end{equation}
which were calculated in \cite{No17}.

\subsubsection{Mirzakhani spectral curve}

There is another simple spectral curve that plays an important role in applications; we will call it the Mirzakhani spectral curve. It was first studied in \cite{EO07c}. We consider the spectral curve
\begin{equation}
\Sigma = \mathbb{C}, \qquad x = \frac{1}{2}z^2, \qquad \omega_{0,1} = \frac{z \sin(2 \pi z)}{2 \pi}\ dz, \qquad \omega_{0,2} = \frac{dz_1 dz_2}{(z_1-z_2)^2}.
\end{equation}
This is an interesting spectral curve, as it does not come from an algebraic curve. There is only one ramification point at $z=0$, which is simple, and the basis of one-forms is still $\xi_{-k}(z) = \frac{dz}{z^{k+1}}$.

In this case, the first few correlators can still be calculated by hand. We get:
\begin{align}
\omega_{0,3}(z_1,z_2,z_3) =& \frac{dz_1 dz_2 dz_3}{z_1^2 z_2^2 z_3^2},\\
\omega_{1,1}(z_1) =& dz_1 \left( \frac{1}{8 z_1^4} + \frac{\pi^2}{12 z_1^2} \right),\\
\omega_{0,4}(z_1,z_2,z_3,z_4) =&dz_1 dz_2 dz_3 dz_4 \left(  \frac{3}{z_1^2 z_2^2 z_3^2 z_4^2} \sum_{i=1}^4 \frac{1}{z_i^2} + \frac{2 \pi^2}{z_1^2 z_2^2  z_3^2 z_4^2}\right),\\
\omega_{1,2}(z_1, z_2) =&dz_1 dz_2 \left( \frac{5}{8 z_1^2 z_2^6} + \frac{5}{8 z_1^6 z_2^2} + \frac{3}{8 z_1^4 z_2^4} + \frac{\pi^2}{2 z_1^2 z_2^4} + \frac{\pi^2}{2 z_1^4 z_2^2} + \frac{\pi^4}{4 z_1^2 z_2^2} \right).
\end{align}

\begin{exercise}
If you are are not bored yet, use the TR formula \eqref{eq:TR} to calculate these correlators by hand! Or, write a code that you can then use to calculate TR for your favourite spectral curve. :-)
\end{exercise}

Interestingly, we see that the correlators are the correlators of the Airy spectral curve plus corrections that involve $\pi$. This is not surprising, given that near the ramification point $z=0$,
\begin{equation}
\omega_{0,1}(z)= \frac{z \sin(2 \pi z)}{2 \pi}\, dz = \left( z^2 - \frac{2 \pi^2 z^4}{3} + O(z^6) \right) dz,
\end{equation}
and hence to first order this is just the Airy spectral curve.

The enumerative interpretation for this spectral curve is very nice. Expanding in the basis of one-forms $\xi_{-k}(z) = \frac{dz}{z^{k+1}}$, one can prove that the only non-zero coefficients $F_{g,n}[k_1,\ldots,k_n]$ are when all $k_i$'s are odd, and they take the form (see  \cite{EO07c,EO08}):
\begin{equation}
F_{g,n}[2k_1+1, \ldots, 2k_n+1] =(2 \pi^2)^{k_0} \prod_{i=1}^n (2k_i+1)!! \int_{\overline{\mathcal{M}}_{g,n}} \kappa_1^{k_0} \psi_1^{k_1} \cdots \psi_n^{k_n},
\end{equation}
where $\kappa_1$ is a kappa classes. In particular, the correlators $\omega_{g,n}$ are the Laplace transforms of the Weil-Petersson volumes $V_{g,n}(L_1,\ldots,L_n)$ of the moduli spaces of bordered Riemann surfaces with geodesic boundaries of lengths $L_1,\ldots,L_n$ and genus $g$. The inverse Laplace transform of the topological recursion formula recovers Mirzakhani's recursion relations for these volumes \cite{Mi07}, as was shown in \cite{EO07c}. This particular example also plays a significant role in applications of topological recursion to JT gravity (see for instance \cite{SSS19} and the lecture notes on JT gravity from this school \cite{Tu24}).

\begin{remark}\label{r:kappa}
It is worth noting that this example can be generalized to compute a general generating series for intersection numbers of $\psi$ and $\kappa$ classes on $\overline{\mathcal{M}}_{g,n}$ from topological recursion, see \cite{KN21}. The spectral curve looks like 
\begin{equation}
\Sigma=\mathbb{C}, \qquad x = \frac{1}{2}z^2, \qquad \omega_{0,1} =\left( z^2 + \sum_{k=1}^\infty g_k z^{2k+2} \right) dz, \qquad \omega_{0,2} = \frac{dz_1 dz_2}{(z_1-z_2)^2},
\end{equation}
where the $g_k$ are formal variables that appear in the generaring series for intersection numbers. (Correspondingly, through the correspondence with Airy ideals, one can write general Virasoro constraints for the associated partition function.) As a result, it follows from this that for any spectral curve of the form 
\begin{equation}\label{eq:action}
\Sigma=\mathbb{C}, \qquad x = \frac{1}{2}z^2, \qquad \omega_{0,1} = \left(z^2 +O(z^3) \right) dz, \qquad \omega_{0,2} = \frac{dz_1 dz_2}{(z_1-z_2)^2},
\end{equation}
 the $F_{g,n}[k_1,\ldots,k_n]$ computed by topological recursion are particular combinations of integrals of $\psi$ and $\kappa$ classes on $\overline{\mathcal{M}}_{g,n}$ \cite{KN21}.
 
 There is a deep reason for this, which is that the $F_{g,n}[k_1,\ldots,k_n]$ are correlators of a semisimple cohomological field theory, which can be recovered from the trivial cohomological field theory (the Kontsevich-Witten or Airy case, which computes intersection  numbers of $\psi$-classes) via the action of the Givental group. In the case of spectral curves of the form \eqref{eq:action}, the corresponding cohomological field theory is obtained simply via a Givental translation. This is explained in \cite{DOSS12} -- see also the lecture notes on moduli spaces of Riemann surfaces in this summer school \cite{GL24}.
\end{remark}

\subsubsection{Simple admissible spectral curves}

The Airy and Bessel examples are the building blocks for simple spectral curves, as they control the local behaviour of a spectral curve near a simple ramification point. For simple admissible spectral curves, we can construct correlators using the topological recursion formula \eqref{eq:TR}, which produces symmetric correlators. We expand in the natural basis of one-forms $d \xi_{-k}^{(a)}$ at the ramification points $a \in \Ra$ as in \eqref{eq:expansion}. Do the coefficients
\begin{equation}
F_{g,n} \begin{bmatrix} a_1 & \ldots & a_n \\ k_1 & \ldots & k_n \end{bmatrix} 
\end{equation}
of the expansion have an interpretation as integrals over $\overline{\mathcal{M}}_{g,n}$, as was the case for the Airy and Bessel spectral curves?

The answer is yes, at least when all ramification points are of Airy-type; one can write down an expression for these coefficients as integrals over $\overline{\mathcal{M}}_{g,n}$.  As mentioned above in Remark \ref{r:kappa}, the reason is that the $F_{g,n}$ are correlators of a semisimple cohomological field theory, which can be obtained by acting on a product of trivial cohomological field theories (one for each ramification point) via the Givental group action. In this case however we need to act with both translations and rotations. Writing down the explicit resulting expression is beyond the scope of these lecture notes; we refer the reader to \cite{Ey11,DOSS12}  for more details and for the direct connection with cohomological field theory.

\subsubsection{The $(r,s)$ spectral curves}

Our next example is the $(r,s)$ spectral curve introduced in Example \ref{e:rs}. We have $\Sigma = \mathbb{C}$, $x = \frac{1}{r} z^r$, $\omega_{0,1} = z^{s-1} dz$ (or, equivalently, $y=z^{s-r}$), and $\omega_{0,2} = \frac{dz_1 dz_2}{(z_1-z_2)^2}$. There is only one ramification point at $z=0$, but for $r>2$ it is not simple anymore. Thus, the topological recursion formula \eqref{eq:TR} only applies when $r=2$; for $r>2$ we need to use its generalization from \cite{BE13}.

Nevertheless, for any $r \geq 2$ we can proceed as usual. As $\omega_{0,2} = \frac{dz_1 dz_2}{(z_1-z_2)^2}$, the natural basis of one-forms at $z=0$ is still given by $ \xi_{-k}(z) =  \frac{d z}{z^{k+1}}$. We expand the correlators $\omega_{g,n}$ as in \eqref{eq:exairy}. Do the coefficients $F_{g,n}[k_1,\ldots,k_n]$ have a natural enumerative geometric interpretation?

This question turns out to be a lot more subtle than expected. 

For the case $s=r+1$, it was shown that the $F_{g,n}$ calculate intersection numbers over the moduli space of curves with $r$-spin structures \cite{BE17,DNOPS15}, which is an enumerative geometric problem that was first studied by Witten in \cite{Wi93}. Indeed, as we will see, for this case topological recursion is equivalent to a set of $\mathcal{W}(gl_r)$-constraints satisfied by the $r$-spin partition function. It is also follows that $Z$ is a tau-function for the $r$-KdV hierarchy (sometimes known as the ``$r$-spin Witten conjecture''), which was originally proved by Faber, Shadrin and Zvonkine in \cite{FSZ06}.

For the case $s=r-1$, an enumerative interpretation was found very recently in \cite{CGG22}. It is the natural $r$-spin generalization of the intersection numbers appearing for the Bessel spectral curve, with the Norbury classes replaced by the top degree classes of the Chiodo classes, a natural generalization coming from \cite{Ch06}. In this case, the partition function $Z$ is a still a tau-function for the $r$-KdV hierarchy; it is the so-called $r$-BGW tau-function \cite{ABDKS23b}.

However, for other choices of $s \in \{1,\ldots, r-1\}$ with $r = \pm 1 \mod s$, at this point the enumerative interpretation of the coefficients is unknown. The natural candidate is to take the top degree classes of the Chiodo classes as for the case $s=r-1$, but this is incorrect  --- see however \cite{BCGS26} for more on this. Finding an enumerative interpretation for these coefficients is an interesting open question in the field.

\subsubsection{Arbitrary admissible spectral curves}

The $(r,s)$ spectral curves are the building blocks for arbitrary admissible spectral curves, as they control the local behaviour near ramification points. Just as for the case of arbitrary admissible simple spectral curves, we expect the coefficients of the expansion in the natural basis of one-forms $\xi_{-k}^{(a)}$ to have an enumerative interpretation. However, as the enumerative interpretation for ramification points of type $(r,s)$ is still unknown in general, at this stage we do not have a general expression for these coefficients in terms of enumerative geometry of curves.

\subsection{A few remarks}

\label{s:remarks}

We end this section with a few remarks concerning generalizations of topological recursion, beyond the natural generalization to curves with arbitrary ramification.
\begin{itemize}
\item The reader may have noticed that in going from Section 2 to Section 3 we replaced the sum over half-integers $g \in \frac{1}{2}{\mathbb{N}}$ by a sum over integers $g \in \mathbb{N}$. There is actually no reason to do this; topological recursion could (and should) be defined for systems of correlators $\{ \omega_{g,n} \}_{g \in \frac{1}{2} \mathbb{N}, n \in \mathbb{N}^*}$. This introduces no new complexity; however, there is now a third initial condition, namely the correlator $\omega_{\frac{1}{2}, 1}(z)$, as $2(\frac{1}{2}) -2 + 1 = 0$. This correlator should be given as part of the data of a spectral curve. We omitted this straigthforward generalization in this section for brevity.
\item In the definition of spectral curves, we assumed that $x: \Sigma \to \mathbb{P}^1$ is a holomorphic map between Riemann surfaces. Equivalently, we can think of $x$ as a meromorphic function on $\Sigma$. We can however generalize the topological recursion framework to include points in $\Sigma$ where $x$ has exponential singularities. We can think of the exponential singularities as ramification points of infinite order, and naturally generalize the topological recursion formula for arbitrary ramification to this context. This framework has interesting applications in enumerative geometry, for instance in the context of Hurwitz theory. See \cite{BKW23} for more details.
\item One of the most fundamental property of topological recursion is known as ``symplectic invariance'', and, in particular, ``$x-y$ invariance''. The precise statement of $x-y$ invariance is not so easy to formulate, but it proceeds as follows. Let $\mathcal{S}$ be a spectral curve, which we now think of as a Riemann surface $\Sigma$, two functions $x$ and $y$ on $\Sigma$, and $\omega_{0,2}$. Let us define a ``dual'' spectral curve $\mathcal{S}^\vee$, obtained by swapping $x$ and $y$. In other words, $\mathcal{S}^\vee$ is defined by the same Riemann surface $\Sigma$, the functions $x^\vee = y$ and $y^\vee = x$, and $\omega_{0,2}$. The statement of $x-y$ invariance is that the correlators $\omega_{g,n}$ produced by topological recursion on $\mathcal{S}$ can be fully reconstructed from the correlators $\omega_{g,n}^\vee$ produced by topological recursion on $\mathcal{S}^\vee$, and vice-versa. This statement was proposed already in the very early days of topological recursion in \cite{EO07,EO07b}, but a fully explicit reconstruction reformula relating the two systems of correlators has only been written down recently \cite{ABDKS23}. We refer to the reader to \cite{ABDKS23} for more details. Interestingly, for $x-y$ invariance to hold in general, one need to generalize topological recursion slightly. In particular, in the definition of spectral curves, $x$ and $y$ are not required to be meromorphic anymore; only $dx$ and $dy$ have to be (this is not unrelated to the case of exponential singularities mentioned above). The general recursive formula adds correction terms to \eqref{eq:TR}, and has become known as ``log TR'' \cite{ABDKS23} (note that for most interesting curves however the correction terms vanish and log TR reduces to the original TR presented here). The resulting projection property satisfied by the correlators is also modified accordingly. Very recently, a further generalization of topological recursion that takes its roots in symplectic invariance is known as ``generalized TR'' \cite{ABDKS24}; we caution the reader that for curves with higher ramification, the correlators constructed from generalized TR generally differs from those constructed by the topological recursion of \cite{BE13}.
\item There are many more generalizations of topological recursion that we will not cover in these lectures notes, for instance, $\beta$-deformed (or non-commutative) topological recursion \cite{EM08,CEM09,CEM11,BE18,BE19,BBCC21}, $Q$-deformed topological recursion \cite{KO22,Os23}, geometric recursion \cite{ABO17}, super topological recursion and super Airy structures \cite{BCHORS19,BO20}, etc.
\end{itemize}

\section{The bridge: loop equations}
\label{s:loop}

As should have become clear by studying the examples of the Airy and Bessel spectral curves, the topological recursion formula \eqref{eq:TR} should be somehow related to the Airy ideals of Section \ref{s:airy}. What is the connection?

This brings us to the roots of topological recursion. Eynard and Orantin originally formulated topological recursion as a method for calculating correlators of Hermitian matrix models, which are generating functions for expectation values of products of traces of matrices. More precisely, they formulated topological recursion as a method for calculating a formal asymptotic solution to the loop equations of matrix models. However, they realized that the recursion formula is much more general than that, and indeed, it is now known to appear in many different contexts, regardless of whether there is an underlying matrix model or not. Nevertheless, the origin of topological recursion as a solution of loop equations is key; this is how we will connect topological recursion with Airy ideals.

\subsection{The origin of TR: loop equations}

As in the previous section, for the sake of clarity we focus on simple spectral curves. However, all the main statements (suitably generalized of course) hold true for arbitrary admissible spectral curves. 

Loop equations (also sometimes called ``abstract loop equations'') are a set of equations satisfied by correlators on a spectral curve. They concern the local behaviour near ramification points of particular combinations of correlators.

Let $\mathcal{S} = (\Sigma,x,\omega_{0,1}, \omega_{0,2})$ be a simple admissible spectral curve. We split the set of ramification points as $\Ra = \Ra^A \cup \Ra^B$, where $\Ra^A$ and $\Ra^B$ include ramification points of Airy-type and Bessel-type respectively. Moreover, for any $a \in \Ra$, let $\sigma_a(z)$ be the local involution near $a$ that exchanges the two sheets (in a local coordinate $\zeta$ centered at $a$, $\sigma_a(\zeta) = - \zeta$). 

\begin{definition}\label{d:loop}
Let $\{\omega_{g,n} \}_{g \in \mathbb{N}, n \in \mathbb{N}^* }$ be a system of correlators on a simple admissible spectral curve $\mathcal{S} = (\Sigma,x,\omega_{0,1}, \omega_{0,2})$. 

We say that the correlators satisfy the \emph{linear loop equation} if, for all $2g-2+n >0$ and all $a \in \Ra$,
\begin{equation}
\frac{1}{dx(z)} \left(\omega_{g,n}(z, z_2, \ldots, z_n) + \omega_{g,n}(\sigma_a(z), z_2, \ldots, z_n) \right)
\end{equation}
is holomorphic in $z$ at $z=a$.

We say that the correlators satisfy the \emph{quadratic loop equation} if, for all $2g-2+n>0$ and all $a \in \Ra$,
\begin{equation}\label{eq:qle}
\frac{1}{dx(z)^2} \left( \omega_{g-1,n+1}(z, \sigma_a(z), z_2, \ldots, z_n) + \sum_{\substack{g_1+g_2=g\\I\cup J = \{ z_2, \ldots, z_n \}}} \omega_{g_1, |I|+1}(z, I) \omega_{g_2,|J|+1}(\sigma_a(z), J) \right)
\end{equation}
is either holomorphic in $z$ at $z=a$ if $a \in \Ra^A$, or it has at most a double pole at $z=a$ if $a \in \Ra^B$.
\end{definition}

Those equations may seem to come out of nowhere, but this is the type of equations that one obtains for correlators of Hermitian matrix models. In this context, they follow from the Schwinger-Dyson equations for the matrix model.

\begin{remark}
Note that the combinatorics of the quadratic loop equation \eqref{eq:qle} are very similar to the recursive structure of topological recursion in \eqref{eq:comb}, but the sum in \eqref{eq:qle} does not have a ``prime'': this means that terms with $\omega_{0,1}$ are included in \eqref{eq:qle}, while they were not included in \eqref{eq:comb}. This is very important.
\end{remark}

\begin{remark}
For spectral curves with higher ramification, loop equations can be similarly formulated, see \cite{BE17,BBCCN18,BKS20,BBCKS23}. If a ramification point has order $r$, then there are $r$ loop equations at this point, with the first two being suitable generalizations of the linear and quadratic loop equations above.
\end{remark}

Given a spectral curve, one can study the solution space of loop equations. Correlators that satisfy both the linear and quadratic loop equations on a simple admissible spectral curves are said to satisfy ``blobbed topological recursion'' \cite{BS15}.

\begin{definition}\label{d:blobbed}
Let $\{\omega_{g,n} \}_{g \in \mathbb{N}, n \in \mathbb{N}^* }$ be a system of correlators on a simple admissible spectral curve $\mathcal{S} = (\Sigma,x,\omega_{0,1}, \omega_{0,2})$. We say that the correlators satisfy \emph{blobbed topological recursion} if they satisfy both linear and quadratic loop equations.
\end{definition}

Blobbed topological recursion is not really a recursion. Indeed, the linear and quadratic loop equations do not uniquely specify the correlators. What they do, however, is specify the principal parts of the $\omega_{g,n}$ at the ramification points $a \in \Ra$ in terms of the correlators $\omega_{g',n'}$ with $2g'-2+n' < 2g-2+n$. In other words, we get some sort of recursive process on $2g-2+n$, but at each step of the recursion we can add to $\omega_{g,n}$ a differential that is holomorphic at all $a \in \Ra$ and still get a solution to the loop equations.

To fix this ambiguity and uniquely determine a solution of the loop equations, we need to impose a further property on the correlators. The natural condition to impose is the projection property on $\Ra$, which was introduced in Definition \ref{d:projprop1}. This gives rise to the Eynard-Orantin topological recursion.

\begin{theorem}\label{d:unique}
Let $\mathcal{S} = (\Sigma,x,\omega_{0,1}, \omega_{0,2})$ be a simple admissible spectral curve. There is a unique system of correlators $\{\omega_{g,n} \}_{g \in \mathbb{N}, n \in \mathbb{N}^* }$ on $\mathcal{S}$ that satisfies both blobbed topological recursion (i.e. the linear and quadratic loop equations) and the projection property on $\Ra$, and it is given by the topological recursion formula \eqref{eq:TR}.
\end{theorem}

The proof of this theorem can be found in many places, for instance  \cite{BS15,BBCKS23}.

This is a key result in the theory of topological recursion, and is in fact the origin of the recursion formula in the first place. Imposing the projection property fixes the holomorphic ambiguity at each step of the pseudo-recursive process of blobbed topological recursion to get a fully recursive construction of a system of correlators satisfying the linear and quadratic loop equations.

\begin{remark}
In Section \ref{s:remarks} we pointed out that for $x-y$ invariance to hold in general, one needs to add corrections to the topological recursion formula. The resulting formula (which is still recursive) is now known as log TR \cite{ABDKS23}. From the point of view of loop equations, log TR still provides a solution to the linear and quadratic loop equations; that is, the correlators still satisfy blobbed topological recursion. What changes is the projection property; the correlators produced by log TR do not satisfy the usual projection property. In other words, what log TR does is propose an alternative condition on the correlators that uniquely fixes the holomorphic blobs at each step of blobbed topological recursion. The result is a different set of correlators, uniquely constructed by log TR, that satisfy the linear and quadratic loop equations.
\end{remark}

\subsection{Loop equations as differential constraints}

We are now ready to connect topological recursion from Section \ref{s:TR} with the differential constraints provided by Airy ideal in Section \ref{s:airy}. What we will show is that topological recursion can be recast as an example of Airy ideals.

The idea is simple; we reformulate the loop equations as differential constraints on a partition function $Z$. To do so, we proceed in three steps. We give a rough sketch of the idea below, which was carried out in full generality in \cite{BBCCN18}.

Let $\{\omega_{g,n} \}_{g \in \mathbb{N}, n \in \mathbb{N}^* }$ be a system of correlators on a simple admissible spectral curve $\mathcal{S} = (\Sigma,x,\omega_{0,1}, \omega_{0,2})$.  For simplicity of notation, let us assume that the ramification set $\Ra$ consists of a single point $a$. (The general case is not much harder, but the notation becomes cumbersome.)

\begin{enumerate}
\item We first impose the projection property on the correlators. This means that, for $2g-2+n>0$, we have a finite expansion in terms of the basis of one-forms $\xi_{-k}$ at $a$ of the form (we omit the superscript $^{(a)}$ for simplicity):
\begin{equation}\label{eq:exp1}
\omega_{g,n}(z_1,\ldots,z_n) = \sum_{k_1,\ldots,k_n \in \mathbb{N}^*} F_{g,n}[k_1,\ldots,k_n]  \xi_{-k_1}(z_1) \cdots  \xi_{-k_n}(z_n).
\end{equation}
For the unstable correlators $\omega_{0,1}$ and $\omega_{0,2}$, we can also expand near $a$. As $\omega_{0,1}$ is holomorphic at $a$, we get
\begin{equation}\label{eq:exp2}
\omega_{0,1}(z) = \sum_{k \in \mathbb{N}^*} F_{0,1}[k]  \xi_k(z),
\end{equation}
where we introduced the locally defined one-forms $\xi_k(z) = \zeta^{k-1}(z) d \zeta(z)$ in a local coordinate $\zeta$ centered at $z=a$.
For $\omega_{0,2}$, one can check that, near $a$, we can write an expansion of the form
\begin{equation}\label{eq:exp3}
\omega_{0,2}(z_1,z_2) = \sum_{k \in \mathbb{N}^*} k  \xi_k(z_1)  \xi_{-k}(z_2).
\end{equation}
\item Next, we reformulate the loop equations as residue conditions. The linear loop equation can be restated as the condition that, for all $k \in \mathbb{N}^*$,
\begin{equation} 
\Res_{z=a} \frac{ \xi_k(z)}{dx(z)} \left( \omega_{g,n}(z, z_2, \ldots, z_n) + \omega_{g,n}(\sigma_a(z), z_2, \ldots, z_n) \right)= 0.
\end{equation}
Similarly, the quadratic loop equation can be restated as the condition that, for all $k \in \mathbb{N}^*$ if $a$ is of Airy-type, and for all $k \geq 3$ if $a$ is of Bessel-type,
\begin{gather}
\Res_{z=a} \frac{ \xi_k(z)}{dx(z)^2}  \Big(\omega_{g-1,n+1}(z, \sigma_a(z), z_2, \ldots, z_n) \\+ \sum_{\substack{g_1+g_2=g\\I\cup J = \{ z_2, \ldots, z_n \}}} \omega_{g_1, |I|+1}(z, I) \omega_{g_2,|J|+1}(\sigma_a(z), J) \Big)= 0.
\end{gather}
Inserting the expansions  \eqref{eq:exp1}, \eqref{eq:exp2} and \eqref{eq:exp3} in these residue formula, we get an infinite set of relations between the $F_{g,n}[k_1,\ldots,k_n]$, with coefficients given by residues of one-forms obtained by taking products of $\xi_k$'s (and dividing by $dx$'s).
\item We construct  a partition function
\begin{equation}
Z = \exp\left( \sum_{\substack{g \in  \mathbb{N}, n \in \mathbb{N}^* \\ 2g-2+n > 0}} \frac{\hslash^{2g-2+n}}{n!}  \sum_{k_1, \ldots, k_n \in \mathbb{N}^*}  F_{g,n}[k_1, \ldots,k_n]  t_{k_1} \cdots t_{k_n}\right),
\end{equation}
and show that the infinite set of relations between the $F_{g,n}[k_1,\ldots,k_n]$ can be naturally obtained by imposing an infinite set of differential constraints $H_k Z = 0$, $k \in \mathbb{N}^*$. The differential operators $H_k$ can be constructed explicitly; after diagonalization if needed, they take the form
\begin{equation}
H_a = \hslash \partial_a -\hslash^2 \left(  \frac{1}{2} A_{abc} t_b t_c + B_{abc} t_b \partial_C + \frac{1}{2}C_{abc} \partial_b \partial_c + D_a \right),
\end{equation}
where the tensors $A_{abc}$, $B_{abc}$, $C_{abc}$ and $D_a$ are constructed by taking residues of one-forms obtained by taking products of $d \xi_k$'s (and dividing by $dx$'s). The data of the spectral curve is therefore encapsulated in these tensors.
\end{enumerate}

This procedure turns the loop equations into differential constraints for the partition function $Z$. But do these $H_a$ generate an Airy ideal in the Weyl algebra $\DAh$ with $A = \mathbb{N}^*$? To prove this, we would need to check condition (2) of Definition \ref{d:airy}, namely $[\I,\I] \subseteq \hslash^2 \I$.

Fortunately, we do not have to do much here. Proceeding with the calculation above explicitly, one can show that the resulting operators $H_a$ form a representation of a subalgebra of the Virasoro algebra! Roughly speaking, if $a$ is of Airy-type, we obtain a conjugation of the Kontsevich-Witten representation from Section \ref{s:KW}. If $a$ is of Bessel-type, we obtain a conjugation of the BGW representation from Section \ref{s:BGW}.\footnote{To be precise, what we obtain is a representation of a subalgebra of the $\mathcal{W}(\mathfrak{gl}_2)$-algebra at self-dual level, see \cite{BBCCN18}. But after reduction, one can recast it as either the Kontsevich-Witten or BGW representation.} As a result, we know that those generate an Airy ideal. Neat!

Let us recap what we have seen in this section, for simple admissible spectral curves.
\begin{itemize}
\item The data of a system of correlators that satisfies the projection property can be encapsulated in the coefficients $F_{g,n}$ of their expansions at the ramification points in the ``good'' basis of one-forms.
\item The correlators satisfy topological recursion if and only if they satisfy the projection property and the loop equations.
\item The loop equations can be recast as a set of differential constraints for a partition function $Z$ constructed out of the coefficients $F_{g,n}$.
\item These differential constraints generate an Airy ideal, and hence uniquely fix the partition function $Z$.
\end{itemize}
As a result, we see that topological recursion can be reformulated as a particular example of Airy ideals!

\subsection{Examples}

\subsubsection{Airy spectral curve}

It is a good exercise to carry out the procedure outlined in the previous section by hand for the Airy and Bessel spectral curves. For these spectral curves, everything is very explicit, and the procedure can be carried out in detail. The tensors $A_{abc}$, $B_{abc}$, $C_{abc}$ and $D_a$ for the resulting differential operators can be calculated explicitly.

Let us introduce the following operators:
\begin{align}\label{eq:Ws}
W^1_k =&\hslash J_{2k},\nonumber\\
W^2_k =& - \hslash^2 \left( \frac{1}{2}\sum_{m_1+m_2 =2 k}(-1)^{m_1+1} : J_{m_1} J_{m_2}: + \frac{1}{8}\delta_{k,0} \right),
\end{align}
where, for $m \geq 1$, 
\begin{equation}
J_m =\partial_m, \qquad J_{-m} = m t_m, \qquad J_0 = 0
\end{equation}
and $: \quad :$ denotes normal ordering. 
This is a representation for the modes of the strong generators of the $\mathcal{W}(\mathfrak{gl}_2)$-algebra at self-dual level. This particular representation is obtained by first embedding the $\mathcal{W}(\mathfrak{gl}_2)$-algebra in the Heisenberg algebra of two free bosons, and then the representation is obtained via restriction of a $\mathbb{Z}_2$-twisted representation for the Heisenberg algebra, see \cite{BBCCN18}. 

To get the resulting differential constraints for the Airy spectral curve, we conjugate these operators as follows:
\begin{align}\label{eq:conj1}
H^1_k =& e^{-\frac{1}{\hslash} \frac{J_3}{3}} W^1_k e^{\frac{1}{\hslash} \frac{J_3}{3}} = \hslash J_{2k},\nonumber\\
H^2_k =& e^{-\frac{1}{\hslash} \frac{J_3}{3}} W^2_k e^{\frac{1}{\hslash} \frac{J_3}{3}} = \hslash J_{2k+3}- \hslash^2 \left( \frac{1}{2}\sum_{m_1+m_2 =2 k} (-1)^{m_1+1}: J_{m_1} J_{m_2}: + \frac{1}{8}\delta_{k,0} \right).
\end{align}
Conjugating doesn't change the commutation relations, so this is still a representation of $\mathcal{W}(\mathfrak{gl}_2)$-algebra at self-dual level. The differential constraints are $H^1_k Z = 0$ for $k \geq 0$ and $H^2_k Z = 0$ for $k \geq -1$. Those operators generate an Airy ideal. 

Those are not quite the Kontsevich-Witten constraints from Section \ref{s:KW}. However, notice that the constraints $H^1_k Z = 0$ impose that $Z$ does not depend on the even variables $t_{2k}$, $k \geq 1$ . Therefore, we can reduce the constraints by removing all terms with derivatives $\partial_{2k}$ with $k \geq 0$, and the differential constraints reduce precisely to the Kontsevich-Witten Virasoro constraints of Section \ref{s:KW}.

\begin{remark}
We remark that the conjugation in \eqref{eq:conj1} was not arbitrary.  For the Airy spectral curve, $\omega_{0,1} = z^2\ dz$. If we expand in the basis $\xi_k(z) = z^{k-1}\ dz$, with $\omega_{0,1} = \sum_{ k \in \mathbb{N}^*} F_{0,1}[k] \xi_k(z)$, we get that the only non-zero $F_{0,1}[k]$ is $F_{0,1}[3] = 1$. To get the differential constraints, we conjugated by the operator
\begin{equation}
T = e^{\frac{1}{\hslash} \sum_{k \in \mathbb{N}^*} F_{0,1}[k] \frac{J_k}{k}} = e^{\frac{1}{\hslash} \frac{J_3}{3} }.
\end{equation}
This is in fact the general recipe, see \cite{BBCCN18}. It essentially corresponds to the action of the Givental translation operator on the differential constraints for the partition function of the cohomological field theory .
\end{remark}

\subsubsection{Bessel spectral curve}

For the Bessel spectral curve, we start with the same representation of $\mathcal{W}(\mathfrak{gl}_2)$ \eqref{eq:Ws}, but since $\omega_{0,1} = dz$, we now conjugate by the operator
\begin{equation}
T = e^{\frac{1}{\hslash} \sum_{k \in \mathbb{N}^*} F_{0,1}[k] \frac{J_k}{k}}= e^{\frac{1}{\hslash} J_1 }.
\end{equation}
The resulting differential constraints take the form $H^1_k Z = 0$ for $k \geq 0$ and $H^2_k Z = 0$ for $k \geq 0$ with 
\begin{align}
H^1_k =&\hslash J_{2k},\\
H^2_k =&  \hslash J_{2k+1} - \hslash^2 \left( \frac{1}{2}\sum_{m_1+m_2 =2 k} (-1)^{m_1+1} : J_{m_1} J_{m_2}: +\frac{1}{8}\delta_{k,0} \right).
\end{align}
Those operators generate an Airy ideal. Furthermore, after reduction with respect to even variables, we obtain precisely the BGW Virasoro constraints of Section \ref{s:BGW}.

\begin{remark}
Note that for the Bessel case the quadratic constraints ($H^2_k Z = 0$) start at $k \geq 0$ (which is essential to get an Airy ideal), instead of $k \geq -1$ in the Airy case. In other words, we have ``less'' constraints in the Bessel case, which is consistent with the fact that the requirement for the quadratic loop equation is weaker, since \eqref{eq:qle} can have double poles.
\end{remark}

\subsubsection{Generating function for $\psi$ and $\kappa$ classes}

It is easy to generalize the result from the previous question to the case studied in Remark \ref{r:kappa}. Recall that the spectral curve is
\begin{equation}
\Sigma=\mathbb{C}, \qquad x = \frac{1}{2}z^2, \qquad \omega_{0,1} =\left( z^2 + \sum_{k=1}^\infty g_k z^{2k+2} \right) dz, \qquad \omega_{0,2} = \frac{dz_1 dz_2}{(z_1-z_2)^2}.
\end{equation}
Rewriting the loop equations as differential constraints for this curve, we get the same operators \eqref{eq:Ws} but now conjugated by 
\begin{equation}
T = e^{\frac{1}{\hslash} \sum_{k \in \mathbb{N}^*} F_{0,1}[k] \frac{J_k}{k}}= e^{\frac{1}{\hslash} \sum_{k \in \mathbb{N}} g_k \frac{J_{2k+3}}{2k+3}}
\end{equation}
where we set $g_0 = 1$. The resulting differential constraints are $H^1_k Z = 0$ for $k \geq 0$ and $H^2_k Z = 0$ for $k \geq -1$, with
\begin{align}
H^1_k =& \hslash J_{2k},\nonumber\\
H^2_k =&  \hslash \left( \sum_{n \in \mathbb{N}} g_n J_{2k+2n+3}\right)- \hslash^2 \left( \frac{1}{2}\sum_{m_1+m_2 =2 k}(-1)^{m_1+1} : J_{m_1} J_{m_2}: + \frac{1}{8}\delta_{k,0} \right).
\end{align}
This is not quite in the form of an Airy ideal, but one can find (infinite) linear combinations of the $H^2_k$ that take the right form for the $O(\hslash)$ term. This means that the ideal is an Airy ideal. We see that we obtain the ``same'' Virasoro constraints again, but appropriately conjugated (which is where the geometry of the spectral curve comes in via $\omega_{0,1}$). This conjugation is essentially the action of the Givental translation operator on the Virasoro constraints satisfied by the partition function of the trivial cohomological field theory.

\subsubsection{The $(r,s)$ spectral curves}

If we go beyond simple ramification, a similar procedure can be carried out (see \cite{BBCCN18}). While the details are more complicated, the result is quite simple. For the $(r,s)$ spectral curve, topological recursion can be recast as a set of differential constraints that form a representation of a subalgebra of the $\mathcal{W}(\mathfrak{gl}_r)$-algebra at self-dual level. This representation is obtained from a $\mathbb{Z}_r$-twisted representation of the Heisenberg algebra of $r$ free bosons. The particular subalgebra and representation also depends on the local parameter $s$. 

In particular, for the $s=r+1$ curve, the resulting differential constraints match with the known $\mathcal{W}$-constraints for $r$-spin intersection numbers . For the case $s=r-1$, the resulting constraints provide $\mathcal{W}$-constraints for the intersection numbers defined in \cite{CGG22} .

\subsubsection{The general picture}

In general, for an admissible spectral curve, topological recursion can be recast as a set of differential constraints attached to each ramification point. If the ramifiction point is simple, those constraints form an appropriately conjugated representation of the Virasoro algebra, either of Kontsevich-Witten or BGW type. If the ramification is of higher order, then those constraints from a representation of a subalgebra of the $\mathcal{W}(\mathfrak{gl}_r)$-algebra at self-dual level; the particular representation depend on the $(r,s)$ type of the ramification point. The geometry of the spectral curve is encoded in the particular conjugations that need to be done to obtain the differential constraints.

As a result, we can think of topological recursion as a fancy geometric way of encoding $\mathcal{W}$-constraints attached to the ramification points of a spectral curve! 

\begin{remark}
An important consequence of the reformulation of topological recursion in terms of Airy ideals is a proof of symmetry. Recall that the fact that the topological recursion formula \eqref{eq:TR} produces symmetric correlators is far from obvious. For simple admissible spectral curves, this was proven by direct calculation in \cite{EO07}. However, for arbitrary admissible spectral curves, such a direct calculation is rather technical and difficult.

Symmetry now directly follows from the reformulation. Indeed, once topological recursion has been shown to be equivalent to differential constraints that generate an Airy ideal, by Theorem \ref{t:airy} we know that there exists a unique partition function $Z$ that satisfies the constraints $\mathcal{I} Z = 0$. The coefficients $F_{g,n}[k_1,\ldots,k_n]$ in the partition function $Z$ are necessarily symmetric by construction, and hence the correlators reconstructed from these coefficients are also symmetric. As a result, it follows that for all spectral curves such that topological recursion can be reformulated as differential constraints that generate an Airy ideal, topological recursion produces symmetric correlators.

This is in fact the origin of the admissibility condition in Definition \ref{d:admissible}. It was shown in \cite{BBCCN18} that only when this admissibility condition is satisfied do the resulting differential constraints generate an Airy ideal. In particular, the local condition $r = \pm 1 \mod s$ was rather surprising and unexpected. But indeed, it can be shown that when this condition is not satisfied, topological recursion does not produce symmetric correlators.

\end{remark}

\section{A web of connections}
\label{s:connections}

So far in this note we focused on the details of topological recursion itself. We first studied when a partition function $Z$ can be uniquely reconstructed recursively by solving differential constraints, which led us to the notion of Airy ideals. We then independently studied a recursion that originated from matrix models, and is formulated in terms of residues on a spectral curve. We showed that, ultimately, both are are faces of the same coin; topological recursion can be recast as a set of differential constraints that generate an Airy ideal.

What we did not do so far is explore the plethora of connections between topological recursion, Airy ideals, and the rest of the world. In the end, this is what makes this research area fascinating. Topological recursion has deep connections with integrability, enumerative geometry, cohomological field theory, matrix models, quantum curves, WKB analysis and resurgence. Whenever topological recursion applies to a particular area of interest, all these connections come into play and shed new light on the problem at hand.

It is beyond the scope of these lecture notes to review this web of connections. Instead, I will briefly explore two connections: with enumerative geometry and quantum curves. This is only the tip of the iceberg however, there are many more interesting connections that unfortunately I cannot cover in only a few lectures on the subject.

\subsection{Enumerative geometry}

Ultimately, topological recursion is interesting if it calculates something interesting. 

We already saw in Remark \ref{r:kappa} that for any genus zero spectral curve of the form $x = \frac{1}{2}z^2$ and $\omega_{0,1}$ of Airy-type, the expansion of the correlators at the ramification point ($z=0$) produces intersection numbers of $\psi$ and $\kappa$ classes over $\overline{\mathcal{M}}_{g,n}$. In fact, it follows from the direct connection between topological recursion and cohomological field theory that a similar interpretation in terms of intersection numbers can be done for arbitrary simple spectral curves with only ramification points of Airy-type \cite{Ey11,DOSS12}. This is not surprising, as topological recursion in this case can be recast as a collection of (appropriately conjugated) Kontsevich-Witten Virasoro constraints.

For simple spectral curves that have some ramification points of both Airy- and Bessel-type, a similar interpretation in terms of intersection numbers involving the Norbury class is expected, since topological recursion can be understood as a collection of (appropriately conjugated) Kontsevich-Wittten and BGW  Virasoro constraints.

For spectral curves with ramification points of higher order, it is expected that one should be able to write similar expressions for the coefficients of the expansion of the correlators at the ramification points in terms of integrals over the moduli space of curves with $r$-spin structure. However, at the moment such formulae have not been obtained yet.

The morale of the story is that \emph{for any admissible spectral curve, the coefficients of the local expansion of the correlators at ramification points have an enumerative meaning in terms of the geometry of $\overline{\mathcal{M}}_{g,n}$.}

But what if we expand the correlators at other points on the spectral curve? Are the coefficients interesting as well?

 It turns out that for many spectral curves, expanding the correlators at another point (that is not a ramification point) is also very interesting and calculates interesting enumerative invariants. Furthermore, we can take full advantage of the fact that the correlators are globally defined over the whole spectral curve and relate the expansion of the correlators at this point to the expansion at the ramification points. This is a very powerful tool which can be used to relate various enumerative invariants to integrals over $\overline{\mathcal{M}}_{g,n}$, yielding interesting results in enumerative geometry.
 
Let me highlight this procedure with a few examples.
 
 \subsubsection{Hurwitz numbers}
 
 Hurwitz theory is a classical problem in geometry. Simple Hurwitz numbers count ramified coverings $ C \to \mathbb{P}^1$, where $C$ is a compact connected genus $g$ Riemann surface, with a prescribed ramification type at $\infty$ and simple ramification points elsewhere (we count with weight $1/|G|$, where $G$ is the automorphism group of the covering).  A good pedagogical introduction to Hurwitz theory is \cite{Ca16}.
 
 It turns out that simple Hurwitz numbers are computed by topological recursion. Consider the spectral curve $\mathcal{S} = (\Sigma, x, \omega_{0,1}, \omega_{0,2})$, with 
 \begin{equation}
 \mathcal{S} = \mathbb{C}, \qquad x = z e^{-z}, \qquad \omega_{0,1} = (1-z)\ dz, \qquad \omega_{0,2} = \frac{dz_1 dz_2}{(z_1-z_2)^2}.
 \end{equation}
 One may recognize $z= - W(-x)$, where $W$ is the Lambert $W$-function.
 
 There is a single ramification point at $z=1$, which is simple and of Airy type. Topological recursion proceeds as usual, and construct correlators $\omega_{g,n}$ whose expansion in a natural basis of one-forms at $z=1$ has coefficients that can be written as intersection numbers over $\overline{\mathcal{M}}_{g,n}$. 
 
However, if we expand the correlators near $x=0$ (since $x=0$ is not a branch point, this expansion makes sense), the result is
\begin{equation}
\omega_{g,n} = \sum_{k_1,\ldots,k_n \in \mathbb{N}^*}  \left( \prod_{i=1}^n k_i \right) H_{g,n}(k_1,\ldots,k_n) x_1^{k_1-1} \ldots, x_n^{k_n-1} dx_1 \cdots dx_n,
\end{equation}
where $H_{g,n}(k_1,\ldots,k_n)$ is the simple Hurwitz number counting ramified coverings $C \to \mathbb{P}^1$ with $C$ of genus $g$ and ramification profile at $\infty$ specified by the positive integers $(k_1,\ldots,k_n)$. Amazing! This was first conjectured in \cite{BM07} by studying a particular limit from topological string theory (see below), and proved in \cite{EMS09} by relating topological recursion to the cut-and-join equation satisfied by Hurwitz numbers.

What is particularly nice though is that we can use the fact that the correlators $\omega_{g,n}$ are globally defined object on the spectral curve to relate the two expansions and rewrite Hurwitz numbers as particular combinations of integrals over $\overline{\mathcal{M}}_{g,n}$. This gives a new (and quite simple) proof of a celebrated result in Hurwitz theory, the ELSV formula \cite{ELSV00}.

This particular example has now been generalized in a myriad of ways. One can for example study $r$-orbifold Hurwitz numbers \cite{DLN12,BHLM13,DLPS15}; the spectral curve is
  \begin{equation}
 \mathcal{S} = \mathbb{C}, \qquad x = z e^{-z^r}, \qquad \omega_{0,1} = z^{r-1}(1-r z^r)\ dz, \qquad \omega_{0,2} = \frac{dz_1 dz_2}{(z_1-z_2)^2}.
 \end{equation}
 Expanding at $x=0$ gives orbifold Hurwitz numbers, while expanding at the ramification points give integrals over $\overline{\mathcal{M}}_{g,n}$. The relation between the two expansions yield the Johnson-Pandharipande-Tseng formula \cite{JPT08}, which is the orbifold generalization of ELSV.
 
 One could study  $r$-spin Hurwitz numbers (also called $r$-completed cycles Hurwitz numbers) \cite{SSZ13,DKPS19}; the spectral curve is
   \begin{equation}
 \mathcal{S} = \mathbb{C}, \qquad x = z e^{-z^r}, \qquad \omega_{0,1} = (1-r z^r)\ dz, \qquad \omega_{0,2} = \frac{dz_1 dz_2}{(z_1-z_2)^2}.
 \end{equation}
 Relating the two expansions proves a formula (called the $r$-ELSV formula) relating $r$-spin Hurwitz numbers in terms of intersection numbers over the moduli space of curves with $r$-spin structures \cite{DKPS19} , which was originally conjectured by Zvonkine \cite{Zv06}.
 
  Once can also combine the last two examples to compute $q$-orbifold $r$-spin Hurwitz numbers, which results in a new $qr$-ELSV formula \cite{DKPS19}.
 
 \begin{remark}
 The $r$-spin example is in fact interesting as well because one can extend topological recursion to include contributions from the exponential singularity at $\infty$. In other words, if we consider instead the spectral curve 
    \begin{equation}
 \mathcal{S} = \mathbb{P}^1, \qquad x = z e^{-z^r}, \qquad \omega_{0,1} = (1-r z^r)\ dz, \qquad \omega_{0,2} = \frac{dz_1 dz_2}{(z_1-z_2)^2},
 \end{equation}
 with the only difference being that we include the exponential singularity in the domain of the spectral curve, then the suitably extended topological recursion produces different correlators. It turns out that expansion of these correlators at $x=0$ generates another type of Hurwitz numbers, known as Atlantes Hurwitz numbers \cite{BKW23}. Atlantes Hurwitz numbers were originally introduced in \cite{ALS16}.
 \end{remark}
 
 The connection between topological recursion and Hurwitz numbers can be extended much further and formulated in a very general setting. We refer the interested reader to \cite{ACEH20,BDKS20,BDKS20b}, which in particular highlights the relations with KP tau-functions.

\subsubsection{Gromov-Witten theory of $\mathbb{P}^1$}

Gromov-Witten theory is a natural generalization of intersection theory over $\overline{\mathcal{M}}_{g,n}$, where instead of working with the moduli space of curves, we work with the moduli space of stable maps to a given target space. It turns out that topological recursion also has a lot to say about Gromov-Witten theory!

Consider the spectral curve $\mathcal{S} = (\Sigma, x, \omega_{0,1}, \omega_{0,2})$, with 
 \begin{equation}
 \mathcal{S} = \mathbb{C}^*, \qquad x = z + \frac{1}{z}, \qquad \omega_{0,1} = \log(z) \frac{z^2-1}{z^2}\ dz, \qquad \omega_{0,2} = \frac{dz_1 dz_2}{(z_1-z_2)^2}.
 \end{equation}
 This curve has two simple ramification points at $z = \pm 1$, both of which are of Airy-type. As usual, the coefficients of the expansion of the correlators $\omega_{g,n}$ can be written as integrals over $\overline{\mathcal{M}}_{g,n}$.
 
However, it was shown that the coefficients of the expansion of the correlators at $x = \infty$ are stationary Gromov-Witten invariants of $\mathbb{P}^1$ \cite{NS11,DOSS12}. In this case, the relation between the two expansions essentially reproduces Givental's reconstruction for the semisimple cohomological field theory \cite{DOSS12}. 

This example is in fact a simple example in a large class of examples; as shown in \cite{DOSS12}, any semisimple cohomological field theory can be reformulated as topological recursion on simple spectral curves with ramification points of Airy type.

\subsubsection{Gromov-Witten theory of toric Calabi-Yau threefolds}

Another large class of examples computed by topological recursion is Gromov-Witten theory of toric Calabi-Yau threefolds. In this case, topological recursion can be understood as an example of mirror symmetry. Without getting into the details of the statement and the proof, let me simply sketch the idea. I will necessarily sweep many details under the rug here.

 Let $X$ be a (non-compact) toric Calabi-Yau threefold. Gromov-Witten theory of $X$ is a mathematical formulation of A-model topological string theory on $X$. The mirror theory is B-model topological string theory on the mirror Calabi-Yau threefold $Y$. When $X$ is a toric Calabi-Yau threefold, the mirror is given by the equation
\begin{equation}
Y = \{ u v = P(x,y) \} \subset \mathbb{C}^2 \times (\mathbb{C}^*)^2,
\end{equation}
where $P(x,y)$ is a polynomial that can be read off from the toric data of $X$. $Y$ is a conic bundle over $(\mathbb{C}^*)^2$, where the conic fiber degenerates  to two lines over the curve
\begin{equation}
\{ P(x,y) = 0 \} \subset (\mathbb{C}^*)^2.
\end{equation}
It turns out that the geometry of the B-model basically boils down to the geometry of this discriminant curve, and this  is essentially the spectral curve of the problem. There is a nice way to understand the topology of this curve; one starts with the toric graph of $X$, and fattens it to get the mirror curve. This is exemplified for $X = \mathcal{O}_{\mathbb{P}^2}(-3) $ in Figure \ref{f:p2}.

\begin{center}
\begin{figure}[H]
\includegraphics[width=0.3\textwidth]{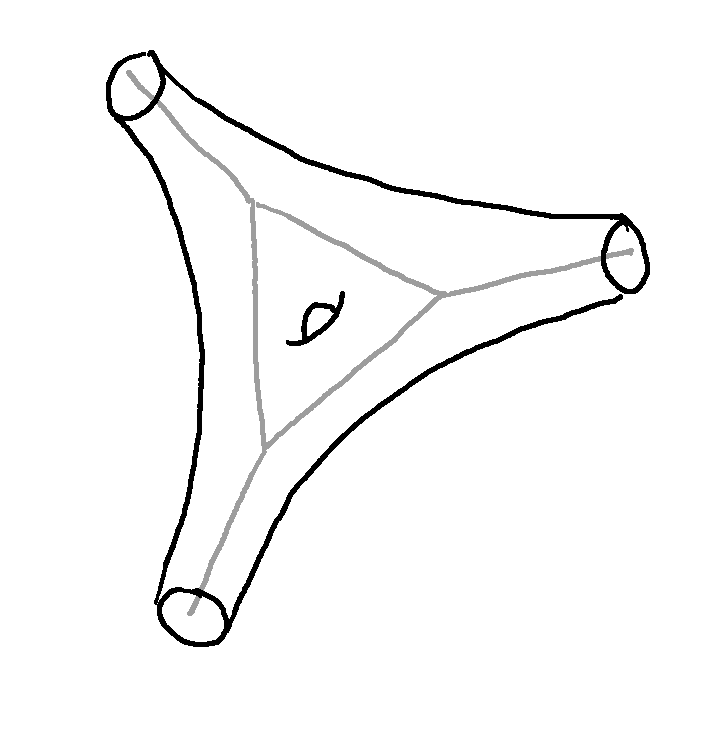}
\caption{The toric diagram and mirror curve for $X = \mathcal{O}_{\mathbb{P}^2}(-3) $.}
\label{f:p2}
\end{figure}
\end{center}

More precisely, since $x,y$ are $\mathbb{C}^*$-coordinates, from the point of view of topological recursion one may think of the curve as
\begin{equation}
C = \{ P(e^x,e^y) = 0 \} \subset (\mathbb{C})^2.
\end{equation}
We then parametrize the curve to get two functions $x$ and $y$ living on a Riemann surface. It turns out that the resulting spectral curves have simple ramification points of Airy-type, and we  proceed with topological recursion to construct correlators $\omega_{g,n}$. We can think of this process as the B-model mirror to A-model topological string theory (that is, Gromov-Witten theory) on $X$.

The idea then is to expand the correlators $\omega_{g,n}$ near a puncture of $C$ (which is a point that one needs to add to $C$ to compactify the curve). The coefficients of this expansion should be related by the so-called ``mirror map'' to open Gromov-Witten invariants of the mirror target space $X$ with a particular choice of Lagrangian brane (so-called ``Aganagic-Vafa branes''). In this context, open Gromov-Witten invariants can be understood as counts of maps from open Riemann surfaces to the target space with the boundaries mapped to the Lagrangian subspace (the ``brane''). Furthermore, the free energies $F_g = \omega_{g,0}$ should be mapped to the closed (that is, the usual) Gromov-Witten invariants of $X$. This claim was originally conjectured in \cite{BKMP07}, and proved in \cite{EO12,FLZ13,FLZ16}.

As the spectral curve $C$ has simple ramification points of Airy-type, as we know there is another expansion at these points whose coefficients are intersection numbers over $\overline{\mathcal{M}}_{g,n}$. What is the relation between these two expansions? It turns out that one can think of it as the B-model mirror to localization in Gromov-Witten theory. Indeed, for Gromov-Witten theory on a toric Calabi-Yau threefold, it turns out that the integrals over the moduli space of stable maps to the target space localize with respect to the torus action and can be related to integrals over $\overline{\mathcal{M}}_{g,n}$. This is what gives rise to the so-called ``topological vertex'' \cite{AKMV03,LLLZ04,Li08}. You will hear more about this in another lecture series at this school \cite{Li24}.

In the end, this gives rise to a very nice picture. On the A-model side, the torus fixed points are the vertices of the toric graph. There is a one-to-one map between these torus fixed points and the ramification points of the mirror curve. As a result, one can think of the local expansion of the correlators $\omega_{g,n}$ at the ramification points as the natural localization process mirror to localization of Gromov-Witten theory on $X$! This is essentially the idea that is made rigorous in the proofs of \cite{EO12,FLZ13,FLZ16} by comparing the graphical decomposition of topological recursion to the graphical reconstruction of Gromov-Witten invariants of $X$ via localization.

\subsection{Quantum curves}

Another interesting connection between topological recursion and the rest of the world is what I would call the ``topological recursion/quantum curve correspondence'' (or TR/QC correspondence). The intuition is that topological recursion should provide a procedure for quantizating the spectral curve. The origin of the statement comes from Hermitian matrix models, where it  is related to the relation between expectation values of products of traces of matrices and the expectation value of the determinant \cite{BE09}. As is common in topological recursion, the intuition comes from matrix models, but the statement is now formulated abstractly without any reference to matrix models.

The correspondence can be stated as follows. We will focus on spectral curves that are constructed from an algebraic curve (see the paragraph around \eqref{eq:alg}):
\begin{equation}
C=\{P(x,y) = 0 \} \subset \mathbb{C}^2.
\end{equation}
We will call these spectral curves ``algebraic''. Let $d$ be the degree of $P$ in $y$. Note however that the TR/QC correspondence could be stated more generally for spectral curves that do not come from algebraic curves.

Apply topological recursion to $C$ (with the usual fundamental bidifferential on the normalization of $C$) to construct correlators $\omega_{g,n}$. Out of those, one can construct the \emph{wave-function}:
\begin{equation}\label{eq:wf}
\psi(z) = \exp \left(\sum_{g \in \mathbb{N}, n \in \mathbb{N}^*} \frac{\hslash^{2g-2+n}}{n!} \left(\int^z_\alpha \cdots \int^z_\alpha \omega_{g,n} - \delta_{g,0} \delta_{n,2} \frac{dx(z_1) dx(z_2)}{(x(z_1)-x(z_2))^2} \right) \right),
\end{equation}
where $\alpha$ is a base point on the normalization of $C$ (that is not a ramification point of $x$) -- it is usually taken to be a pole of $x$. Here we are integrating the correlators $\omega_{g,n}$ in all variables from $\alpha$ to the same variable $z$. 

To state the TR/QC correspondence, we introduce the notion of a quantum curve.

\begin{definition}\label{d:QC}
A \emph{quantum curve} $\hat{P}$ of an algebraic spectral curve $C = \{P(x,y) = 0\} \subset \mathbb{C}^2$ is an order $d$ linear differential operator in $x$, such that, after normal ordering, it takes the form
\begin{equation}
\hat{P}\left(  x, \hslash \frac{d}{dx}; \hslash \right) = P\left(  x,\hslash \frac{d}{dx}\right) + \sum_{n \geq 1} \hslash^n P_n\left(x,\hslash \frac{d}{dx}\right),
\end{equation}
where the leading term $P$ is the original polynomial defining the spectral curve, and the $P_n$ are (normal-ordered) polynomials of degree $< d$. We usually impose that only finitely many correction terms $P_n$ are non-vanishing.
\end{definition}

This is a quantization of the spectral curve, as it amounts to replacing $(x,y) \mapsto \left( x, \hslash \frac{d}{dx} \right)$. But of course, this process is not unique, since the operators $x$ and $\hslash \frac{d}{dx}$ do not commute, and hence the quantization may include $\hslash$ corrections.

The claim of the TR/QC correspondence is the following:

\begin{claim}[TR/QC correspondence]
Let $C = \{P(x,y) = 0\} \subset \mathbb{C}^2$ be a spectral curve, and construct the wave-function $\psi(z)$ as in \eqref{eq:wf} from the correlators obtained by topological recursion. Then there exists a quantum curve $\hat{P}$ such that
\begin{equation}
\hat{P} \psi = 0.
\end{equation}
\end{claim}

This statement should be qualified in a number of ways:
\begin{itemize}
\item The construction of the wave-function depends on a choice of base point $\alpha$. Correspondingly, the quantum curve, if it  exists, should also depend on $\alpha$. Generally speaking, the claim is only expected to hold for $\alpha$ a pole of $x$ (that is not a ramification point, otherwise the integral of the $\omega_{g,n}$ may diverge), but other choices of base points may also lead to quantum curves \cite{BE17}.
\item As stated, the TR/QC correspondence is only expected to hold for genus 0 spectral curves. It can however be formulated for spectral curves of any genus; one simply needs to replace the wave-function \eqref{eq:wf} by its non-perturbative counterpart, which amounts to multiplying it by appropriate theta functions \cite{EM08,EGMO21}. 
\item It is probably better to state the correspondence in terms of a system of first-order linear equations (instead of a single higher-order differential equation), which is more natural from the point of view of integrability. For simplicity we decided not to go in this direction in these lecture notes, but we refer the reader to \cite{BE11,MO19,EGMO21} for more details. 
\end{itemize}

Going back to TR/QC correspondence, the statement was first proved in various papers for a number of specific spectral curves relevant to enumerative geometry (see for instance \cite{BE09,BHLM13,DDM14,DM13,DN14,DM14,DM13b,DMNPS13,LMS13,MSS13,MS12,IKT18,IKT18b,MO19} -- I most certainly forgot some here). It was also studied in the context of knot theory, as it provides a constructive approach to the AJ conjecture \cite{Ga03} -- see for instance  \cite{DFM10,BE12,AFGS12,GS11}. In the context of mirror symmetry for A-model topological string theory on toric Calabi-Yau threefolds, it should relate to the recent work of Mari\~no and collaborators (see for instance \cite{GM22,IM23} and many more papers); you can learn more about this in Mari\~no's lecture notes from this school \cite{Ma24}. 

 Going beyond the case-by-case approach, the correspondence was proved in \cite{BE17} for a class of genus zero algebraic spectral curves with arbitrary ramification (the class corresponds to all genus zero spectral curves whose Newton polygon has no interior point and that are smooth as affine curves). More recently, it was proved in \cite{EGMO21} for all algebraic spectral curves (any genus) that only have simple ramification points. We can summarize these results in the following theorem:
\begin{theorem}[\cite{BE17,EGMO21}]
The TR/QC correspondence holds for all genus zero algebraic spectral curves in the class considered in \cite{BE17}, and for all algebraic spectral curves (any genus) that only have simple ramification points.
\end{theorem}
As a generic spectral curve only has simple ramification points, and in principle spectral curves with higher ramification can be obtained as limit points in families of curves with only simple ramification (see \cite{BBCKS23}), the correspondence is expected to hold in full generality for all algebraic spectral curves.

To summarize, what this correspondence says is that topological recursion can be used to reconstruct the WKB asymptotic solution to a quantum differential operator; as such, it can be understood as a procedure to quantize the spectral curve. But this opens the way to more questions: once we have an asymptotic solution, we can ask whether we can ``resum'' the solution to reconstruct actual non-perturbative solutions to the differential equation. This is the idea of Borel resummation and resurgence, which is a very active research area --- more on this in Iwaki's lecture notes from this school \cite{Iw26}. Can topological recursion be used to go beyond the perturbative, asymptotic solution to the quantum curve?

Some work in this direction appeared in \cite{IK21,IM23}. Very recently, \cite{EGGGL23} also obtained new results in this direction, by studying the large genus asymptotics of the solution of the quantum curves for the Airy and Bessel spectral curves (in other words, they are calculating large genus asymptotics for intersection numbers of $\psi$-classes with and without the Norbury class). They also obtain large genus asymptotics for the $r$-spin intersection numbers (the quantum curve for the $s=r+1$ spectral curve). It would be very interesting to continue in this direction and study large genus asymptotics and Borel resummations for quantum curves obtained via topological recursion in general.

{\setlength\emergencystretch{\hsize}\hbadness=10000}

\printbibliography

@article{ABCO17,
   title={The ABCD of topological recursion},
   volume={439},
   ISSN={0001-8708},
   url={http://dx.doi.org/10.1016/j.aim.2023.109473},
   DOI={10.1016/j.aim.2023.109473},
   journal={Advances in Mathematics},
   publisher={Elsevier BV},
   author={Andersen, Jørgen Ellegaard and Borot, Gaëtan and Chekhov, Leonid O. and Orantin, Nicolas},
   year={2024},
   month=mar, pages={109473},
	url={https://arxiv.org/abs/1703.03307},}

@article{BBCKS23,
   title={Taking limits in topological recursion},
   volume={112},
   ISSN={1469-7750},
   url={http://dx.doi.org/10.1112/jlms.70286},
   DOI={10.1112/jlms.70286},
   number={3},
   journal={Journal of the London Mathematical Society},
   publisher={Wiley},
   author={Borot, Gaëtan and Bouchard, Vincent and Chidambaram, Nitin Kumar and Kramer, Reinier and Shadrin, Sergey},
   year={2025},
   month=Sept,
	url= {https://arxiv.org/abs/2309.01654},}

@misc{BCU24,
      title={Whittaker vectors at finite energy scale, topological recursion and Hurwitz numbers}, 
      author={Gaëtan Borot and Nitin Kumar Chidambaram and Giacomo Umer},
      year={2024},
      eprint={2403.16938},
      archivePrefix={arXiv},
      primaryClass={math-ph},
      url={https://arxiv.org/abs/2403.16938}, }

@article{BBCC21,
   title={Whittaker vectors for $\mathcal{W}$-algebras from topological recursion},
   volume={30},
   ISSN={1420-9020},
   url={http://dx.doi.org/10.1007/s00029-024-00921-x},
   DOI={10.1007/s00029-024-00921-x},
   number={2},
   journal={Selecta Mathematica},
   publisher={Springer Science and Business Media LLC},
   author={Borot, Gaëtan and Bouchard, Vincent and Chidambaram, Nitin K. and Creutzig, Thomas},
   year={2024},
   month=mar,
	url={https://arxiv.org/abs/2104.04516},}

@article{BBCCN18,
   title={Higher Airy Structures, $\mathcal{W}$-Algebras and Topological Recursion},
   volume={296},
   ISSN={1947-6221},
   url={http://dx.doi.org/10.1090/memo/1476},
   DOI={10.1090/memo/1476},
   number={1476},
   journal={Memoirs of the American Mathematical Society},
   publisher={American Mathematical Society (AMS)},
   author={Borot, Gaëtan and Bouchard, Vincent and Chidambaram, Nitin and Creutzig, Thomas and Noshchenko, Dmitry},
   year={2024},
   month=apr,
	url={https://arxiv.org/abs/1812.08738},}

@article{BKW23,
   title={Topological recursion on transalgebraic spectral curves and Atlantes Hurwitz numbers},
   volume={206},
   ISSN={0393-0440},
   url={http://dx.doi.org/10.1016/j.geomphys.2024.105306},
   DOI={10.1016/j.geomphys.2024.105306},
   journal={Journal of Geometry and Physics},
   publisher={Elsevier BV},
   author={Bouchard, Vincent and Kramer, Reinier and Weller, Quinten},
   year={2024},
   month=Dec, pages={105306},
	     url={https://arxiv.org/abs/2304.07433},}

@article{BCJ22,
   title={Airy Ideals, Transvections, and ${\mathcal {W}}(\mathfrak {sp}_{2N})$-Algebras},
   volume={25},
   ISSN={1424-0661},
   url={http://dx.doi.org/10.1007/s00023-023-01374-2},
   DOI={10.1007/s00023-023-01374-2},
   number={5},
   journal={Annales Henri Poincaré},
   publisher={Springer Science and Business Media LLC},
   author={Bouchard, Vincent and Creutzig, Thomas and Joshi, Aniket},
   year={2023},
   month=Oct, pages={2669–2730},
      url={https://arxiv.org/abs/2207.04336}, 
}

@article{BE17,
   title={Reconstructing WKB from topological recursion},
   volume={4},
   ISSN={2270-518X},
   url={http://dx.doi.org/10.5802/jep.58},
   DOI={10.5802/jep.58},
   journal={Journal de l’École polytechnique — Mathématiques},
   publisher={Cellule MathDoc/CEDRAM},
   author={Bouchard, Vincent and Eynard, Bertrand},
   year={2017},
   month=aug, pages={845–908},
	url={https://arxiv.org/abs/1606.04498}, }

@article{BE13,
   title={Think globally, compute locally},
   volume={2013},
   ISSN={1029-8479},
   url={http://dx.doi.org/10.1007/JHEP02(2013)143},
   DOI={10.1007/jhep02(2013)143},
   number={2},
   journal={Journal of High Energy Physics},
   publisher={Springer Science and Business Media LLC},
   author={Bouchard, Vincent and Eynard, Bertrand},
   year={2013},
   month=feb,
	url={https://arxiv.org/abs/1211.2302} }

@article{BHLMR13,
   title={A Generalized Topological Recursion for Arbitrary Ramification},
   volume={15},
   ISSN={1424-0661},
   url={http://dx.doi.org/10.1007/s00023-013-0233-0},
   DOI={10.1007/s00023-013-0233-0},
   number={1},
   journal={Annales Henri Poincaré},
   publisher={Springer Science and Business Media LLC},
   author={Bouchard, Vincent and Hutchinson, Joel and Loliencar, Prachi and Meiers, Michael and Rupert, Matthew},
   year={2013},
   month=feb, pages={143–169},
	url={https://arxiv.org/abs/1208.6035}, }

@misc{BM07,
      title={Hurwitz numbers, matrix models and enumerative geometry}, 
      author={Vincent Bouchard and Marcos Mari\~no},
      year={2008},
      eprint={0709.1458},
      archivePrefix={arXiv},
      primaryClass={math.AG},
      url={https://arxiv.org/abs/0709.1458}, 
}

@article{BKMP07,
   title={Remodeling the B-Model},
   volume={287},
   ISSN={1432-0916},
   url={http://dx.doi.org/10.1007/s00220-008-0620-4},
   DOI={10.1007/s00220-008-0620-4},
   number={1},
   journal={Communications in Mathematical Physics},
   publisher={Springer Science and Business Media LLC},
   author={Bouchard, Vincent and Klemm, Albrecht and Mariño, Marcos and Pasquetti, Sara},
   year={2008},
   month=sep, pages={117–178},
	url={https://arxiv.org/abs/0709.1453}, }

@article{CE06,
   title={Matrix eigenvalue model: Feynman graph technique for all genera},
   volume={2006},
   ISSN={1029-8479},
   url={http://dx.doi.org/10.1088/1126-6708/2006/12/026},
   DOI={10.1088/1126-6708/2006/12/026},
   number={12},
   journal={Journal of High Energy Physics},
   publisher={Springer Science and Business Media LLC},
   author={Chekhov, Leonid and Eynard, Bertrand},
   year={2006},
   month=dec, pages={026–026},
	url={https://arxiv.org/abs/math-ph/0604014} }

@misc{EO07,
      title={Invariants of algebraic curves and topological expansion}, 
      author={Bertrand Eynard and Nicolas Orantin},
      year={2007},
      eprint={math-ph/0702045},
      archivePrefix={arXiv},
      primaryClass={math-ph},
      url={https://arxiv.org/abs/math-ph/0702045}, 
}

@misc{EO07b,
      title={Topological expansion of mixed correlations in the hermitian 2 Matrix Model and $x-y$ symmetry of the $F_g$ invariants}, 
      author={Bertrand Eynard and Nicolas Orantin},
      year={2007},
      eprint={0705.0958},
      archivePrefix={arXiv},
      primaryClass={math-ph},
      url={https://arxiv.org/abs/0705.0958}, 
}

@misc{EO08,
      title={Algebraic methods in random matrices and enumerative geometry}, 
      author={Bertrand Eynard and Nicolas Orantin},
      year={2008},
      eprint={0811.3531},
      archivePrefix={arXiv},
      primaryClass={math-ph},
      url={https://arxiv.org/abs/0811.3531}, 
}

@article{DVV91,
    author = "Dijkgraaf, Robbert and Verlinde, Herman L. and Verlinde, Erik P.",
    editor = "Brezin, E. and Wadia, S. R.",
    title = "{Loop equations and Virasoro constraints in nonperturbative 2-D quantum gravity}",
    reportNumber = "IASSNS-HEP-90-48, PUPT-1184",
    doi = "10.1016/0550-3213(91)90199-8",
    journal = "Nucl. Phys. B",
    volume = "348",
    pages = "435--456",
    year = "1991"
}

@article{Ko91,
    author = "Kontsevich, Maxim",
    title = "{Intersection theory on the moduli space of curves and the matrix Airy function}",
    doi = "10.1007/BF02099526",
    journal = "Commun. Math. Phys.",
    volume = "147",
    pages = "1--23",
    year = "1992"
}

@article{Wi91,
    author = "Witten, Edward",
    title = "{Two-dimensional gravity and intersection theory on moduli space}",
    doi = "10.4310/SDG.1990.v1.n1.a5",
    journal = "Surveys Diff. Geom.",
    volume = "1",
    pages = "243--310",
    year = "1991"
}

@misc{KS17,
      title={Airy structures and symplectic geometry of topological recursion}, 
      author={Maxim Kontsevich and Yan Soibelman},
      year={2017},
      eprint={1701.09137},
      archivePrefix={arXiv},
      primaryClass={math.AG},
      url={https://arxiv.org/abs/1701.09137}, 
}

@article{HR19,
   title={Airy Structures for Semisimple Lie Algebras},
   volume={385},
   ISSN={1432-0916},
   url={http://dx.doi.org/10.1007/s00220-021-04142-7},
   DOI={10.1007/s00220-021-04142-7},
   number={3},
   journal={Communications in Mathematical Physics},
   publisher={Springer Science and Business Media LLC},
   author={Hadasz, Leszek and Ruba, Błażej},
   year={2021},
   month=jun, pages={1535–1569},
	url={https://arxiv.org/abs/1911.10453} }

@article{BKS20,
   title={Higher Airy structures and topological recursion for singular spectral curves},
   volume={11},
   ISSN={2308-5835},
   url={http://dx.doi.org/10.4171/AIHPD/168},
   DOI={10.4171/aihpd/168},
   number={1},
   journal={Annales de l’Institut Henri Poincaré D, Combinatorics, Physics and their Interactions},
   publisher={European Mathematical Society - EMS - Publishing House GmbH},
   author={Borot, Gaëtan and Kramer, Reinier and Schüler, Yannik},
   year={2024},
   month=feb, pages={1–146},
	url={https://arxiv.org/abs/2010.03512},}

@article{CGG22,
   title={Relations on $\overline{\mathcal{M}
}_{g,n}$ and the negative $r$-spin Witten conjecture},
   volume={241},
   ISSN={1432-1297},
   url={http://dx.doi.org/10.1007/s00222-025-01351-y},
   DOI={10.1007/s00222-025-01351-y},
   number={3},
   journal={Inventiones mathematicae},
   publisher={Springer Science and Business Media LLC},
   author={Chidambaram, Nitin Kumar and Garcia-Failde, Elba and Giacchetto, Alessandro},
   year={2025},
   month=Aug, pages={929–997},
      url={https://arxiv.org/abs/2205.15621},  }

@article{No17,
   title={A new cohomology class on the moduli space of curves},
   volume={27},
   ISSN={1465-3060},
   url={http://dx.doi.org/10.2140/gt.2023.27.2695},
   DOI={10.2140/gt.2023.27.2695},
   number={7},
   journal={Geometry and Topology},
   publisher={Mathematical Sciences Publishers},
   author={Norbury, Paul},
   year={2023},
   month=sep, pages={2695–2761},
	url={https://arxiv.org/abs/1712.03662}, }

@article{Ch06,
   title={Towards an enumerative geometry of the moduli space of twisted curves
                    and rth roots},
   volume={144},
   ISSN={1570-5846},
   url={http://dx.doi.org/10.1112/S0010437X08003709},
   DOI={10.1112/s0010437x08003709},
   number={6},
   journal={Compositio Mathematica},
   publisher={Wiley},
   author={Chiodo, Alessandro},
   year={2008},
   month=nov, pages={1461–1496},
	url={https://arxiv.org/abs/math/0607324}, }

@misc{EGMO21,
      title={Quantization of classical spectral curves via topological recursion}, 
      author={Bertrand Eynard and Elba Garcia-Failde and Olivier Marchal and Nicolas Orantin},
      year={2024},
      eprint={2106.04339},
      archivePrefix={arXiv},
      primaryClass={math-ph},
      url={https://arxiv.org/abs/2106.04339}, 
}

@misc{Ey11,
      title={Invariants of spectral curves and intersection theory of moduli spaces of complex curves}, 
      author={Bertrand Eynard},
      year={2011},
      eprint={1110.2949},
      archivePrefix={arXiv},
      primaryClass={math-ph},
      url={https://arxiv.org/abs/1110.2949}, 
}

@article{DOSS12,
   title={Identification of the Givental Formula with the Spectral Curve Topological Recursion Procedure},
   volume={328},
   ISSN={1432-0916},
   url={http://dx.doi.org/10.1007/s00220-014-1887-2},
   DOI={10.1007/s00220-014-1887-2},
   number={2},
   journal={Communications in Mathematical Physics},
   publisher={Springer Science and Business Media LLC},
   author={Dunin-Barkowski, Petr and Orantin, Nicolas and Shadrin, Sergey and Spitz, Loek},
   year={2014},
   month=feb, pages={669–700},
	url={https://arxiv.org/abs/1211.4021},}

@article{BS15,
   title={Blobbed topological recursion: properties and applications},
   volume={162},
   ISSN={1469-8064},
   url={http://dx.doi.org/10.1017/S0305004116000323},
   DOI={10.1017/s0305004116000323},
   number={1},
   journal={Mathematical Proceedings of the Cambridge Philosophical Society},
   publisher={Cambridge University Press (CUP)},
   author={Borot, Gaëtan and Shadrin, Sergey},
   year={2016},
   month=may, pages={39–87},
	url={https://arxiv.org/abs/1502.00981},}

@misc{ABDKS23,
      title={Log topological recursion through the prism of $x-y$ swap}, 
      author={Alexander Alexandrov and Boris Bychkov and Petr Dunin-Barkowski and Maxim Kazarian and Sergey Shadrin},
      year={2024},
      eprint={2312.16950},
      archivePrefix={arXiv},
      primaryClass={math-ph},
      url={https://arxiv.org/abs/2312.16950}, 
}

@article{FSZ06,
   title={Tautological relations and the $r$-spin Witten conjecture},
   volume={43},
   ISSN={1873-2151},
   url={http://dx.doi.org/10.24033/asens.2130},
   DOI={10.24033/asens.2130},
   number={4},
   journal={Annales scientifiques de l’École normale supérieure},
   publisher={Societe Mathematique de France},
   author={Faber, Carel and Shadrin, Sergey and Zvonkine, Dimitri},
   year={2010},
   pages={621–658},
	url={https://arxiv.org/abs/math/0612510},}

@misc{Wi93,
    author = "Witten, Edward",
    title = "{Algebraic geometry associated with matrix models of two-dimensional gravity}",
    year = "1993"
}

@misc{ABDKS23b,
      title={KP integrability through the $x-y$ swap relation}, 
      author={Alexander Alexandrov and Boris Bychkov and Petr Dunin-Barkowski and Maxim Kazarian and Sergey Shadrin},
      year={2023},
      eprint={2309.12176},
      archivePrefix={arXiv},
      primaryClass={math-ph},
      url={https://arxiv.org/abs/2309.12176}, 
}

@article{DNOPS15,
	title={Dubrovin's superpotential as a global spectral curve},
   volume={18},
   ISSN={1475-3030},
   url={http://dx.doi.org/10.1017/S147474801700007X},
   DOI={10.1017/s147474801700007x},
   number={3},
   journal={Journal of the Institute of Mathematics of Jussieu},
   publisher={Cambridge University Press (CUP)},
   author={Dunin-Barkowski, Petr and Norbury, Paul and Orantin, Nicolas and Popolitov, Alexandr and Shadrin, Sergey},
   year={2017},
   month=apr, pages={449-497},
	url={https://arxiv.org/abs/1509.06954}, }

@misc{EO07c,
      title={Weil-Petersson volume of moduli spaces, Mirzakhani's recursion and matrix models}, 
      author={Bertrand Eynard and Nicolas Orantin},
      year={2007},
      eprint={0705.3600},
      archivePrefix={arXiv},
      primaryClass={math-ph},
      url={https://arxiv.org/abs/0705.3600},
}

@misc{EOjava,
	title={Intersection numbers},
	author={Bertrand Eynard},
	url={http://bertrand.eynard.free.fr/IN/IN.html},
}

@article{Mi07,
author = {Mirzakhani, Maryam},
year = {2007},
month = {01},
pages = {179-222},
title = {Simple geodesics and Weil-Petersson volumes of moduli spaces of bordered Riemann surfaces},
volume = {167},
journal = {Inventiones mathematicae},
doi = {10.1007/s00222-006-0013-2}
}

@misc{SSS19,
      title={JT gravity as a matrix integral}, 
      author={Phil Saad and Stephen H. Shenker and Douglas Stanford},
      year={2019},
      eprint={1903.11115},
      archivePrefix={arXiv},
      primaryClass={hep-th},
      url={https://arxiv.org/abs/1903.11115}, 
}

@misc{KN21,
      title={Polynomial relations among kappa classes on the moduli space of curves}, 
      author={Maxim Kazarian and Paul Norbury},
      year={2021},
      eprint={2112.11672},
      archivePrefix={arXiv},
      primaryClass={math.AG},
      url={https://arxiv.org/abs/2112.11672}, 
}

@book{Ca16, 
	place={Cambridge}, 
	series={London Mathematical Society Student Texts}, 
	title={Riemann Surfaces and Algebraic Curves: A First Course in Hurwitz Theory}, 
	publisher={Cambridge University Press}, 
	author={Cavalieri, Renzo and Miles, Eric}, 
	year={2016}, 
	collection={London Mathematical Society Student Texts}}

@misc{EMS09,
      title={The Laplace transform of the cut-and-join equation and the Bouchard-Marino conjecture on Hurwitz numbers}, 
      author={Bertrand Eynard and Motohico Mulase and Brad Safnuk},
      year={2010},
      eprint={0907.5224},
      archivePrefix={arXiv},
      primaryClass={math.AG},
      url={https://arxiv.org/abs/0907.5224}, 
}

@article{ELSV00,
   title={Hurwitz numbers and intersections on moduli spaces of curves},
   volume={146},
   ISSN={1432-1297},
   url={http://dx.doi.org/10.1007/s002220100164},
   DOI={10.1007/s002220100164},
   number={2},
   journal={Inventiones Mathematicae},
   publisher={Springer Science and Business Media LLC},
   author={Ekedahl, Torsten and Lando, Sergei and Shapiro, Michael and Vainshtein, Alek},
   year={2001},
   month=nov, pages={297–327},
	url={https://arxiv.org/abs/math/0004096}, }

@article{DLN12,
   title={Orbifold Hurwitz numbers and Eynard–Orantin invariants},
   volume={23},
   ISSN={1945-001X},
   url={http://dx.doi.org/10.4310/MRL.2016.v23.n5.a3},
   DOI={10.4310/mrl.2016.v23.n5.a3},
   number={5},
   journal={Mathematical Research Letters},
   publisher={International Press of Boston},
   author={Do, Norman and Leigh, Oliver and Norbury, Paul},
   year={2016},
   pages={1281–1327},
	url={https://arxiv.org/abs/1212.6850}, }

@article{BHLM13,
    author = "Bouchard, Vincent and Hernandez Serrano, Daniel and Liu, Xiaojun and Mulase, Motohico",
    title = "{Mirror symmetry for orbifold Hurwitz numbers}",
    eprint = "1301.4871",
    archivePrefix = "arXiv",
    primaryClass = "math.AG",
    journal = "J. Diff. Geom.",
    volume = "98",
    number = "3",
    pages = "375--423",
    year = "2014",
url={https://arxiv.org/abs/1301.4871}, 
}

@article{DLPS15,
   title={Polynomiality of orbifold Hurwitz numbers, spectral curve, and a new proof of the Johnson–Pandharipande–Tseng formula},
   volume={92},
   ISSN={1469-7750},
   url={http://dx.doi.org/10.1112/jlms/jdv047},
   DOI={10.1112/jlms/jdv047},
   number={3},
   journal={Journal of the London Mathematical Society},
   publisher={Wiley},
   author={Dunin-Barkowski, Petr and Lewanski, Danilo and Popolitov, Alexandr and Shadrin, Sergey},
   year={2015},
   month=oct, pages={547–565},
	url={https://arxiv.org/abs/1504.07440}, }

@misc{JPT08,
      title={Abelian Hurwitz-Hodge integrals}, 
      author={Paul Johnson and Rahul Pandharipande and Hsian-Hua Tseng},
      year={2008},
      eprint={0803.0499},
      archivePrefix={arXiv},
      primaryClass={math.AG},
      url={https://arxiv.org/abs/0803.0499}, 
}

@article{SSZ13,
   title={Equivalence of ELSV and Bouchard–Mariño conjectures for $r$-spin Hurwitz numbers},
   volume={361},
   ISSN={1432-1807},
   url={http://dx.doi.org/10.1007/s00208-014-1082-y},
   DOI={10.1007/s00208-014-1082-y},
   number={3–4},
   journal={Mathematische Annalen},
   publisher={Springer Science and Business Media LLC},
   author={Shadrin, Sergey and Spitz, Loek and Zvonkine, Dimitri},
   year={2014},
   month=aug, pages={611–645},
	url={https://arxiv.org/abs/1306.6226}, }

@article{DKPS19,
   title={Loop equations and a proof of Zvonkine’s $qr$-ELSV formula},
   ISSN={1873-2151},
   url={http://dx.doi.org/10.24033/asens.2553},
   DOI={10.24033/asens.2553},
   journal={Annales Scientifiques de l École Normale Supérieure},
   publisher={Societe Mathematique de France},
   author={Dunin-Barkowski, Petr and Kramer, Reinier and Popolitov, Alexandr and Shadrin, Sergey},
   year={2023},
   month=oct,
	url={https://arxiv.org/abs/1905.04524}, }

@misc{Zv06,
	author={Zvonkine, Dimitri},
	title={A preliminary text on the $r$-ELSV formula}, 
	year={2006},
}

@article{ALS16,
   title={Ramifications of Hurwitz theory, KP integrability and quantum curves},
   volume={2016},
   ISSN={1029-8479},
   url={http://dx.doi.org/10.1007/JHEP05(2016)124},
   DOI={10.1007/jhep05(2016)124},
   number={5},
   journal={Journal of High Energy Physics},
   publisher={Springer Science and Business Media LLC},
   author={Alexandrov, Alexandr and Lewanski, Danilo and Shadrin, Sergey},
   year={2016},
   month=may,
	url={https://arxiv.org/abs/1512.07026}, }

@article{ACEH20,
   title={Weighted Hurwitz Numbers and Topological Recursion},
   volume={375},
   ISSN={1432-0916},
   url={http://dx.doi.org/10.1007/s00220-020-03717-0},
   DOI={10.1007/s00220-020-03717-0},
   number={1},
   journal={Communications in Mathematical Physics},
   publisher={Springer Science and Business Media LLC},
   author={Alexandrov, Alexandr and Chapuy, Guillaume and Eynard, Betrand and Harnad, John},
   year={2020},
   month=mar, pages={237–305},
	url={https://arxiv.org/abs/1806.09738}, }

@article{BDKS20,
   title={Explicit closed algebraic formulas for Orlov–Scherbin $n$-point functions},
   volume={9},
   ISSN={2270-518X},
   url={http://dx.doi.org/10.5802/jep.202},
   DOI={10.5802/jep.202},
   journal={Journal de l’École polytechnique — Mathématiques},
   publisher={Cellule MathDoc/CEDRAM},
   author={Bychkov, Boris and Dunin-Barkowski, Petr and Kazarian, Maxim and Shadrin, Sergey},
   year={2022},
   month=jul, pages={1121–1158},
	url={https://arxiv.org/abs/2008.13123}, }

@article{BDKS20b,
   title={Topological recursion for Kadomtsev–Petviashvili tau functions of hypergeometric type},
   volume={109},
   ISSN={1469-7750},
   url={http://dx.doi.org/10.1112/jlms.12946},
   DOI={10.1112/jlms.12946},
   number={6},
   journal={Journal of the London Mathematical Society},
   publisher={Wiley},
   author={Bychkov, Boris and Dunin‐Barkowski, Petr and Kazarian, Maxim and Shadrin, Sergey},
   year={2024},
   month=jun,
	url={https://arxiv.org/abs/2012.14723},}

@article{NS11,
   title={Gromov–Witten invariants of $\mathbb{P}^1$ and Eynard–Orantin invariants},
   volume={18},
   ISSN={1465-3060},
   url={http://dx.doi.org/10.2140/gt.2014.18.1865},
   DOI={10.2140/gt.2014.18.1865},
   number={4},
   journal={Geometry and Topology},
   publisher={Mathematical Sciences Publishers},
   author={Norbury, Paul and Scott, Nick},
   year={2014},
   month=oct, pages={1865–1910},
	url={https://arxiv.org/abs/1106.1337}, }

@misc{FLZ13,
      title={All Genus Open-Closed Mirror Symmetry for Affine Toric Calabi-Yau 3-Orbifolds}, 
      author={Bohan Fang and Chiu-Chu Melissa Liu and Zhengyu Zong},
      year={2019},
      eprint={1310.4818},
      archivePrefix={arXiv},
      primaryClass={math.AG},
      url={https://arxiv.org/abs/1310.4818}, 
}

@misc{EO12,
      title={Computation of open Gromov-Witten invariants for toric Calabi-Yau 3-folds by topological recursion, a proof of the BKMP conjecture}, 
      author={Bertrand Eynard and Nicolas Orantin},
      year={2013},
      eprint={1205.1103},
      archivePrefix={arXiv},
      primaryClass={math-ph},
      url={https://arxiv.org/abs/1205.1103}, 
}

@misc{FLZ16,
      title={On the Remodeling Conjecture for Toric Calabi-Yau 3-Orbifolds}, 
      author={Bohan Fang and Chiu-Chu Melissa Liu and Zhengyu Zong},
      year={2019},
      eprint={1604.07123},
      archivePrefix={arXiv},
      primaryClass={math.AG},
      url={https://arxiv.org/abs/1604.07123}, 
}

@article{AKMV03,
   title={The Topological Vertex},
   volume={254},
   ISSN={1432-0916},
   url={http://dx.doi.org/10.1007/s00220-004-1162-z},
   DOI={10.1007/s00220-004-1162-z},
   number={2},
   journal={Communications in Mathematical Physics},
   publisher={Springer Science and Business Media LLC},
   author={Aganagic, Mina and Klemm, Albrecht and Marino, Marcos and Vafa, Cumrun},
   year={2004},
   month=sep, pages={425–478},
	url={https://arxiv.org/abs/hep-th/0305132},}

@article{LLLZ04,
   title={A mathematical theory of the topological vertex},
   volume={13},
   ISSN={1465-3060},
   url={http://dx.doi.org/10.2140/gt.2009.13.527},
   DOI={10.2140/gt.2009.13.527},
   number={1},
   journal={Geometry and Topology},
   publisher={Mathematical Sciences Publishers},
   author={Li, Jun and Liu, Chiu-Chu Melissa and Liu, Kefeng and Zhou, Jian},
   year={2009},
   month=jan, pages={527–621},
	url={https://arxiv.org/abs/math/0408426},}

@misc{Li08,
      title={Gromov-Witten invariants of toric Calabi-Yau threefolds}, 
      author={Chiu-Chu Melissa Liu},
      year={2009},
      eprint={0811.4703},
      archivePrefix={arXiv},
      primaryClass={math.AG},
      url={https://arxiv.org/abs/0811.4703}, 
}

@misc{EGGGL23,
      title={Resurgent large genus asymptotics of intersection numbers}, 
      author={Bertrand Eynard and Elba Garcia-Failde and Alessandro Giacchetto and Paolo Gregori and Danilo Lewański},
      year={2023},
      eprint={2309.03143},
      archivePrefix={arXiv},
      primaryClass={math.AG},
      url={https://arxiv.org/abs/2309.03143}, 
}

@article{EM08,
   title={A holomorphic and background independent partition function for matrix models and topological strings},
   volume={61},
   ISSN={0393-0440},
   url={http://dx.doi.org/10.1016/j.geomphys.2010.11.012},
   DOI={10.1016/j.geomphys.2010.11.012},
   number={7},
   journal={Journal of Geometry and Physics},
   publisher={Elsevier BV},
   author={Eynard, Bertrand and Mariño, Marcos},
   year={2011},
   month=jul, pages={1181–1202},
	url={https://arxiv.org/abs/0810.4273},}

@article{MO19,
   title={Isomonodromic deformations of a rational differential system and reconstruction with the topological recursion: The sl2 case},
   volume={61},
   ISSN={1089-7658},
   url={http://dx.doi.org/10.1063/5.0002260},
   DOI={10.1063/5.0002260},
   number={6},
   journal={Journal of Mathematical Physics},
   publisher={AIP Publishing},
   author={Marchal, Olivier and Orantin, Nicolas},
   year={2020},
   month=jun,
	url={https://arxiv.org/abs/1901.04344}, }

@article{BE11,
   title={Geometry of Spectral Curves and All Order Dispersive Integrable System},
   ISSN={1815-0659},
   url={http://dx.doi.org/10.3842/SIGMA.2012.100},
   DOI={10.3842/sigma.2012.100},
   journal={Symmetry, Integrability and Geometry: Methods and Applications},
   publisher={SIGMA (Symmetry, Integrability and Geometry: Methods and Application)},
   author={Borot, Gaëtan and Eynard, Bertrand},
   year={2012},
   month=dec,
	url={https://arxiv.org/abs/1110.4936}, }

@misc{BE09,
      title={Determinantal formulae and loop equations}, 
      author={Michel Bergère and Bertrand Eynard},
      year={2009},
      eprint={0901.3273},
      archivePrefix={arXiv},
      primaryClass={math-ph},
      url={https://arxiv.org/abs/0901.3273}, 
}

@misc{DDM14,
      title={Topological recursion and a quantum curve for monotone Hurwitz numbers}, 
      author={Norman Do and Alastair Dyer and Daniel V. Mathews},
      year={2014},
      eprint={1408.3992},
      archivePrefix={arXiv},
      primaryClass={math.GT},
      url={https://arxiv.org/abs/1408.3992}, 
}

@misc{DM13,
      title={Quantum curves for the enumeration of ribbon graphs and hypermaps}, 
      author={Norman Do and David Manescu},
      year={2013},
      eprint={1312.6869},
      archivePrefix={arXiv},
      primaryClass={math.GT},
      url={https://arxiv.org/abs/1312.6869}, 
}

@article{DN14,
   title={Topological recursion for irregular spectral curves},
   volume={97},
   ISSN={1469-7750},
   url={http://dx.doi.org/10.1112/jlms.12112},
   DOI={10.1112/jlms.12112},
   number={3},
   journal={Journal of the London Mathematical Society},
   publisher={Wiley},
   author={Do, Norman and Norbury, Paul},
   year={2018},
   month=mar, pages={398–426},
	url={https://arxiv.org/abs/1412.8334}, }

@misc{DM14,
      title={Quantization of spectral curves for meromorphic Higgs bundles through topological recursion}, 
      author={Olivia Dumitrescu and Motohico Mulase},
      year={2014},
      eprint={1411.1023},
      archivePrefix={arXiv},
      primaryClass={math.AG},
      url={https://arxiv.org/abs/1411.1023}, 
}

@article{DM13b,
   title={Quantum Curves for Hitchin Fibrations and the Eynard–Orantin Theory},
   volume={104},
   ISSN={1573-0530},
   url={http://dx.doi.org/10.1007/s11005-014-0679-0},
   DOI={10.1007/s11005-014-0679-0},
   number={6},
   journal={Letters in Mathematical Physics},
   publisher={Springer Science and Business Media LLC},
   author={Dumitrescu, Olivia and Mulase, Motohico},
   year={2014},
   month=feb, pages={635–671},
	url={https://arxiv.org/abs/1310.6022},}

@article{DMNPS13,
   title={Quantum spectral curve for the Gromov–Witten theory of the complex
projective line},
   volume={2017},
   ISSN={0075-4102},
   url={http://dx.doi.org/10.1515/crelle-2014-0097},
   DOI={10.1515/crelle-2014-0097},
   number={726},
   journal={Journal für die reine und angewandte Mathematik (Crelles Journal)},
   publisher={Walter de Gruyter GmbH},
   author={Dunin-Barkowski, Petr and Mulase, Motohico and Norbury, Paul and Popolitov, Alexander and Shadrin, Sergey},
   year={2014},
   month=nov, pages={267–289},
	url={https://arxiv.org/abs/1312.5336},}

@misc{LMS13,
      title={Quantum curves for simple Hurwitz numbers of an arbitrary base curve}, 
      author={Xiaojun Liu and Motohico Mulase and Adam Sorkin},
      year={2013},
      eprint={1304.0015},
      archivePrefix={arXiv},
      primaryClass={math.AG},
      url={https://arxiv.org/abs/1304.0015}, 
}

@article{MSS13,
   title={The spectral curve and the Schrödinger equation of double Hurwitz numbers and higher spin structures},
   volume={7},
   ISSN={1931-4531},
   url={http://dx.doi.org/10.4310/CNTP.2013.v7.n1.a4},
   DOI={10.4310/cntp.2013.v7.n1.a4},
   number={1},
   journal={Communications in Number Theory and Physics},
   publisher={International Press of Boston},
   author={Mulase, Motohico and Shadrin, Sergey and Spitz, Loek},
   year={2013},
   pages={125–143},
	url={https://arxiv.org/abs/1301.5580}, }

@article{MS12,
   title={Spectral curves and the Schrödinger equations for the Eynard–Orantin recursion},
   volume={19},
   ISSN={1095-0753},
   url={http://dx.doi.org/10.4310/ATMP.2015.v19.n5.a2},
   DOI={10.4310/atmp.2015.v19.n5.a2},
   number={5},
   journal={Advances in Theoretical and Mathematical Physics},
   publisher={International Press of Boston},
   author={Mulase, Motohico and Sulkowski, Piotr},
   year={2015},
   pages={955–1015},
	url={https://arxiv.org/abs/1210.3006},}

@article{DFM10,
   title={The volume conjecture, perturbative knot invariants, and recursion relations for topological strings},
   volume={849},
   ISSN={0550-3213},
   url={http://dx.doi.org/10.1016/j.nuclphysb.2011.03.014},
   DOI={10.1016/j.nuclphysb.2011.03.014},
   number={1},
   journal={Nuclear Physics B},
   publisher={Elsevier BV},
   author={Dijkgraaf, Robbert and Fuji, Hiroyuki and Manabe, Masahide},
   year={2011},
   month=aug, pages={166–211},
	url={https://arxiv.org/abs/1010.4542},}

@misc{Ga03,
      title={Difference and differential equations for the colored Jones function}, 
      author={Stavros Garoufalidis},
      year={2007},
      eprint={math/0306229},
      archivePrefix={arXiv},
      primaryClass={math.GT},
      url={https://arxiv.org/abs/math/0306229}, 
}

@misc{BE12,
      title={All-order asymptotics of hyperbolic knot invariants from non-perturbative topological recursion of A-polynomials}, 
      author={Gaëtan Borot and Bertrand Eynard},
      year={2012},
      eprint={1205.2261},
      archivePrefix={arXiv},
      primaryClass={math-ph},
      url={https://arxiv.org/abs/1205.2261}, 
}

@article{AFGS12,
   title={Volume conjecture: refined and categorified},
   volume={16},
   ISSN={1095-0753},
   url={http://dx.doi.org/10.4310/ATMP.2012.v16.n6.a3},
   DOI={10.4310/atmp.2012.v16.n6.a3},
   number={6},
   journal={Advances in Theoretical and Mathematical Physics},
   publisher={International Press of Boston},
   author={Awata, Hidetoshi and Fuji, Hiroyuki and Gukov, Sergei and Sułkowski, Piotr},
   year={2012},
   pages={1669–1777},
	url={https://arxiv.org/abs/1203.2182},}

@article{GS11,
   title={A-polynomial, B-model, and quantization},
   volume={2012},
   ISSN={1029-8479},
   url={http://dx.doi.org/10.1007/JHEP02(2012)070},
   DOI={10.1007/jhep02(2012)070},
   number={2},
   journal={Journal of High Energy Physics},
   publisher={Springer Science and Business Media LLC},
   author={Gukov, Sergei and Sulkowski, Piotr},
   year={2012},
   month=feb,
	url={https://arxiv.org/abs/1108.0002},}

@misc{IKT18,
      title={Voros Coefficients for the Hypergeometric Differential Equations and Eynard-Orantin's Topological Recursion - Part I : For the Weber Equation}, 
      author={Kohei Iwaki and Tatsuya Koike and Yumiko Takei},
      year={2018},
      eprint={1805.10945},
      archivePrefix={arXiv},
      primaryClass={math.CA},
      url={https://arxiv.org/abs/1805.10945}, 
}

@misc{IKT18b,
      title={Voros Coefficients for the Hypergeometric Differential Equations and Eynard-Orantin's Topological Recursion - Part II : For the Confluent Family of Hypergeometric Equations}, 
      author={Kohei Iwaki and Tatsuya Koike and Yumiko Takei},
      year={2018},
      eprint={1810.02946},
      archivePrefix={arXiv},
      primaryClass={math.CA},
      url={https://arxiv.org/abs/1810.02946}, 
}

@article{IK21,
   title={Topological Recursion and Uncoupled BPS Structures II: Voros Symbols and the $\tau$-Function},
   volume={399},
   ISSN={1432-0916},
   url={http://dx.doi.org/10.1007/s00220-022-04563-y},
   DOI={10.1007/s00220-022-04563-y},
   number={1},
   journal={Communications in Mathematical Physics},
   publisher={Springer Science and Business Media LLC},
   author={Iwaki, Kohei and Kidwai, Omar},
   year={2023},
   month=mar, pages={519–572},
	url={https://arxiv.org/abs/2108.06995},}

@article{IM23,
   title={Resurgent Structure of the Topological String and the First Painlevé Equation},
   ISSN={1815-0659},
   url={http://dx.doi.org/10.3842/SIGMA.2024.028},
   DOI={10.3842/sigma.2024.028},
   journal={Symmetry, Integrability and Geometry: Methods and Applications},
   publisher={SIGMA (Symmetry, Integrability and Geometry: Methods and Application)},
   author={Iwaki, Kohei and Mariño, Marcos},
   year={2024},
   month=apr,
	url={https://arxiv.org/abs/2307.02080},}

@misc{Os23,
      title={Refined topological recursion revisited -- properties and conjectures}, 
      author={Kento Osuga},
      year={2023},
      eprint={2305.02494},
      archivePrefix={arXiv},
      primaryClass={math-ph},
      url={https://arxiv.org/abs/2305.02494}, 
}

@article{KO22,
   title={Quantum curves from refined topological recursion: The genus 0 case},
   volume={432},
   ISSN={0001-8708},
   url={http://dx.doi.org/10.1016/j.aim.2023.109253},
   DOI={10.1016/j.aim.2023.109253},
   journal={Advances in Mathematics},
   publisher={Elsevier BV},
   author={Kidwai, Omar and Osuga, Kento},
   year={2023},
   month=nov, pages={109253},
	url={https://arxiv.org/abs/2204.12431},}

@article{BO20,
   title={${\mathcal {N}}=1$ super topological recursion},
   volume={111},
   ISSN={1573-0530},
   url={http://dx.doi.org/10.1007/s11005-021-01479-x},
   DOI={10.1007/s11005-021-01479-x},
   number={6},
   journal={Letters in Mathematical Physics},
   publisher={Springer Science and Business Media LLC},
   author={Bouchard, Vincent and Osuga, Kento},
   year={2021},
   month=nov,
	url={https://arxiv.org/abs/2007.13186},}

@article{BCHORS19,
   title={Super Quantum Airy Structures},
   volume={380},
   ISSN={1432-0916},
   url={http://dx.doi.org/10.1007/s00220-020-03876-0},
   DOI={10.1007/s00220-020-03876-0},
   number={1},
   journal={Communications in Mathematical Physics},
   publisher={Springer Science and Business Media LLC},
   author={Bouchard, Vincent and Ciosmak, Paweł and Hadasz, Leszek and Osuga, Kento and Ruba, Błażej and Sułkowski, Piotr},
   year={2020},
   month=oct, pages={449–522},
	url={https://arxiv.org/abs/1907.08913},}

@misc{ABO17,
      title={Geometric recursion}, 
      author={Jørgen Ellegaard Andersen and Gaëtan Borot and Nicolas Orantin},
      year={2023},
      eprint={1711.04729},
      archivePrefix={arXiv},
      primaryClass={math.GT},
      url={https://arxiv.org/abs/1711.04729}, 
}

@misc{CEM09,
      title={Topological expansion of the Bethe ansatz, and quantum algebraic geometry}, 
      author={Leonid Chekhov and Bertrand Eynard and Olivier Marchal},
      year={2009},
      eprint={0911.1664},
      archivePrefix={arXiv},
      primaryClass={math-ph},
      url={https://arxiv.org/abs/0911.1664}, 
}

@article{CEM11,
   title={Topological expansion of the $\beta$-ensemble model and quantum algebraic geometry in the sectorwise approach},
   volume={166},
   ISSN={1573-9333},
   url={http://dx.doi.org/10.1007/s11232-011-0012-3},
   DOI={10.1007/s11232-011-0012-3},
   number={2},
   journal={Theoretical and Mathematical Physics},
   publisher={Springer Science and Business Media LLC},
   author={Chekhov, Leonid and Eynard, Betrand and Marchal, Olivier},
   year={2011},
   month=feb, pages={141–185},
	url={https://arxiv.org/abs/1009.6007},}

@misc{BE18,
      title={Integrability of $\mathcal W({\mathfrak{sl}_d})$-symmetric Toda conformal field theories I : Quantum geometry}, 
      author={Raphaël Belliard and Bertrand Eynard},
      year={2019},
      eprint={1801.03433},
      archivePrefix={arXiv},
      primaryClass={math-ph},
      url={https://arxiv.org/abs/1801.03433}, 
}

@misc{BE19,
      title={From the quantum geometry of Fuchsian systems to conformal blocks of W-algebras}, 
      author={Raphaël Belliard and Bertrand Eynard},
      year={2020},
      eprint={1907.10543},
      archivePrefix={arXiv},
      primaryClass={math-ph},
      url={https://arxiv.org/abs/1907.10543}, 
}

@misc{ABDKS24,
      title={Degenerate and irregular topological recursion}, 
      author={Alexander Alexandrov and Boris Bychkov and Petr Dunin-Barkowski and Maxim Kazarian and Sergey Shadrin},
      year={2024},
      eprint={2408.02608},
      archivePrefix={arXiv},
      primaryClass={math-ph},
      url={https://arxiv.org/abs/2408.02608}, 
}

@article{DK05,
   title={Finite vs affine W-algebras},
   volume={1},
   ISSN={1861-3624},
   url={http://dx.doi.org/10.1007/s11537-006-0505-2},
   DOI={10.1007/s11537-006-0505-2},
   number={1},
   journal={Japanese Journal of Mathematics},
   publisher={Springer Science and Business Media LLC},
   author={De Sole, Alberto and Kac, Victor G.},
   year={2006},
   month=apr, pages={137–261},
	url={https://arxiv.org/abs/math-ph/0511055},}

@misc{Ar16,
      title={Introduction to W-algebras and their representation theory}, 
      author={Tomoyuki Arakawa},
      year={2017},
      eprint={1605.00138},
      archivePrefix={arXiv},
      primaryClass={math.RT},
      url={https://arxiv.org/abs/1605.00138}, 
}

@article{BM20,
   title={A new class of higher quantum Airy structures as modules of $\mathcal{W}(\mathfrak{gl}_r)$-algebras},
   volume={14},
   ISSN={2542-4653},
   url={http://dx.doi.org/10.21468/SciPostPhys.14.6.169},
   DOI={10.21468/scipostphys.14.6.169},
   number={6},
   journal={SciPost Physics},
   publisher={Stichting SciPost},
   author={Bouchard, Vincent and Mastel, Kieran},
   year={2023},
   month=jun,
	url={https://arxiv.org/abs/2009.13047}, }

@book{SST99,  
	series={Algorithms and Computation in Mathematics}, 
	title={Gröbner Deformations of Hypergeometric Differential Equations}, 
	publisher={Springer Berlin, Heidelberg}, 
	author={Mutsumi Saito and Bernd Sturmfels and Nobuki Takayama}, 
	year={1999}, 
	DOI={10.1007/978-3-662-04112-3},}

@misc{Be24,
	title={Lecture notes on Riemann surfaces and theta functions},
	author={Marco Bertola},
	year={2024},
}

@article{GL24,
   title={Les Houches lecture notes on moduli spaces of Riemann surfaces},
   ISSN={2590-1990},
   url={http://dx.doi.org/10.21468/SciPostPhysLectNotes.111},
   DOI={10.21468/scipostphyslectnotes.111},
   journal={SciPost Physics Lecture Notes},
   publisher={Stichting SciPost},
   author={Giacchetto, Alessandro and Lewański, Danilo},
   year={2026},
   month=Jan,
	url={https://arxiv.org/abs/2410.13273},
 }

@misc{Li24,
	title={Lecture notes on A-model topological string theory},
	author={Chiu-Chu Melissa Liu},
	year={2024},
}

@article{Tu24,
   	title={Les Houches lectures on two-dimensional gravity and holography},
	pages = {113},
	author = {Turiaci, Gustavo Joaquin},
	journal = {SciPost Phys. Lect. Notes},
	year = {2026},
	publisher = {SciPost},
	doi = {10.21468/SciPostPhysLectNotes.113},
	url = {https://scipost.org/10.21468/SciPostPhysLectNotes.113},
	url = {https://arxiv.org/abs/2412.09537v1},
}

@article{Ma24,
   title={Les Houches lectures on non-perturbative topological strings},
   ISSN={2590-1990},
   url={http://dx.doi.org/10.21468/SciPostPhysLectNotes.112},
   DOI={10.21468/scipostphyslectnotes.112},
   journal={SciPost Physics Lecture Notes},
   publisher={Stichting SciPost},
   author={Mariño, Marcos},
   year={2026},
   month=Jan,
	url ={https://arxiv.org/abs/2411.16211} }

@article{GM22,
   title={Exact multi-instantons in topological string theory},
   volume={15},
   ISSN={2542-4653},
   url={http://dx.doi.org/10.21468/SciPostPhys.15.4.179},
   DOI={10.21468/scipostphys.15.4.179},
   number={4},
   journal={SciPost Physics},
   publisher={Stichting SciPost},
   author={Gu, Jie and Mariño, Marcos},
   year={2023},
   month=oct,
	url={https://arxiv.org/abs/2211.01403},}

@article{BCGS26,
   title={Theta classes: Generalized topological recursion, integrability, and W‐constraints},
   volume={114},
   ISSN={1469-7750},
   url={http://dx.doi.org/10.1112/jlms.70728},
   DOI={10.1112/jlms.70728},
   number={4},
   journal={Journal of the London Mathematical Society},
   publisher={Wiley},
   author={Bouchard, Vincent and Chidambaram, Nitin Kumar and Giacchetto, Alessandro and Shadrin, Sergey},
   year={2026},
   month=Oct,
	url={https://arxiv.org/abs/2505.11291} }

@misc{Iw26,
      title={Les Houches Lectures on Exact WKB Analysis and Painlev\'e Equations}, 
      author={Kohei Iwaki},
      year={2026},
      eprint={2512.17599},
      archivePrefix={arXiv},
      primaryClass={math-ph},
      url={https://arxiv.org/abs/2512.17599}, 
}

\end{document}